\documentclass[12pt,a4paper]{article}

\usepackage{times}
\usepackage{amsmath}
\usepackage{amssymb}
\usepackage{amsfonts}
\usepackage{latexsym}
\usepackage{color}
\usepackage{graphicx}
\usepackage{physics}
\usepackage{cite}
\usepackage{rotating}
\usepackage[normalem]{ulem}

\catcode `\@=11 \@addtoreset{equation}{section}
\def\theequation{\arabic{section}.\arabic{equation}}
\catcode `\@=12

\newcommand{\be}{\begin{equation}}
\newcommand{\en}{\end{equation}}
\newcommand{\bea}{\begin{eqnarray}}
\newcommand{\ena}{\end{eqnarray}}
\newcommand{\beano}{\begin{eqnarray*}}
\newcommand{\enano}{\end{eqnarray*}}
\newcommand{\bee}{\begin{enumerate}}
\newcommand{\ene}{\end{enumerate}}

\newcommand{\mc}{\mathcal}

\newcommand{\Sc}{{\cal S}}
\newcommand{\E}{{\cal E}}
\newcommand{\V}{{\cal V}}
\newcommand{\F}{{\cal F}}

\newcommand{\1}{1 \!\! 1}
\newcommand{\vno}{\varphi_n^{(0)}}
\newcommand{\pno}{\psi_n^{(0)}}
\newcommand{\pmo}{\psi_m^{(0)}}
\newcommand{\vnuno}{\varphi_n^{(1)}}
\newcommand{\pnuno}{\psi_n^{(1)}}
\newcommand{\pmuno}{\psi_m^{(1)}}

\def\scp<#1>{\left\langle #1 \right\rangle}
\def\scpp<#1>{\left\langle\!\left\langle #1\right\rangle\!\right\rangle}
\def\scppp<#1>{\left\langle\!\!\left\langle #1\right\rangle\!\!\right\rangle}
\newcommand{\supzero}{^{(0)}}
\newcommand{\supone}{^{(1)}}

\newcommand{\hop}{t_{h}}

\newcommand{\Hil}{\mc H}

\newtheorem{thm}{Theorem}
\newtheorem{cor}[thm]{Corollary}
\newtheorem{lemma}[thm]{Lemma}
\newtheorem{prop}[thm]{Proposition}
\newtheorem{defn}[thm]{Definition}

\newenvironment{proof}{\noindent {\bf Proof --}}{\hfill$\square$ \vspace{3mm}\endtrivlist}

\catcode `\@=11 \@addtoreset{equation}{section}
\catcode `\@=12

\begin{document}

\thispagestyle{empty}

\vspace*{2cm}

\begin{center}
{\Large \bf A chain of solvable non-Hermitian Hamiltonians constructed by a series of metric operators}   \vspace{2cm}\\

{\large Fabio Bagarello}\\
  Dipartimento di  Ingegneria, \\
Universit\`a di Palermo,\\ I-90128  Palermo, Italy, e\\
I.N.F.N., Sezione di Napoli, Italy\\
e-mail: fabio.bagarello@unipa.it\\
home page: www1.unipa.it/fabio.bagarello

\vspace{2mm}


\vspace{2mm}

{\large Naomichi Hatano}\\
Institute of Industrial Science, The University of Tokyo\\ Kashiwanoha, Kashiwa, Chiba 277-8574, Japan\\
e-mail: hatano@iis.u-tokyo.ac.jp\\
home page: hatano-lab.iis.u-tokyo.ac.jp/hatano

\end{center}

\vspace*{1cm}

\begin{abstract}
\noindent  We show how, given a non-Hermitian Hamiltonian $H$, we can generate new non-Hermitian operators sequentially, producing a virtually infinite chain of non-Hermitian Hamiltonians which are isospectral to $H$ and $H^\dagger$ and whose eigenvectors we can easily deduce in an almost automatic way; no ingredients are necessary other than $H$ and its eigensystem. To set off the chain and keep it running, we use, for the first time in our knowledge, a series of maps all connected to different metric operators. We show how the procedure works in several physically relevant systems. In particular, we apply our method to various versions of the Hatano-Nelson model and to some PT-symmetric Hamiltonians.
\end{abstract}

\vspace{2cm}


\vfill


\newpage

\section{Introduction}\label{sectintr}

In quantum mechanics, quite often the first step in understanding a physical system $\Sc$ consists of deducing the eigenstates and the eigenvalues of the Hamiltonian that drives the time evolution of $\Sc$. This is because the eigenvalues give us the stationary states of the system, while the eigenstates usually produce a basis of the Hilbert space in which $\Sc$ is defined. However, finding eigenvalues and eigenvectors of some Hamiltonian $H$ may not be so simple. For this reason, along the years many techniques have been proposed to simplify this search. 

There exist several lines of research in which the main interest is to deduce new {\em solvable} Hamiltonians (i.e., Hamiltonians whose eigenvalues and eigenvectors can be found in some explicit way) out of some {\em seed} solvable Hamiltonian.
The role of intertwining operators and of ladder operators has proved to be useful~\cite{intop1,intop2,intop3,bagintop,fern} in finding a series of solvable Hamiltonians.
Supersymmetric quantum mechanics (SUSY-QM)~\cite{jun,coop,suku,bagI} is another line of research of this type.  
The factorizability of the Hamiltonian can be also useful in this perspective; see Ref.~\cite{dong} and references therein. 

In most of the approaches cited above, one generates new Hamiltonians with eigenvalues and eigenvectors that are easy to deduce, out of a {\em seed} Hamiltonian for which these objects are known. For instance, if $H_1$ is such a Hamiltonian, so that all the quantities in the eigenvalue equation $H_1 f_n^{(1)}=E_n f_n^{(1)}$ are known, then it may happen that a second Hamiltonian $H_2$ exists, together with a so-called intertwining operator obeying $UH_1=H_2U$. Under these conditions, if $f_n^{(1)}$ does not belong to the kernel of $U$, which is not required to be invertible, then $f_n^{(2)}=Uf_n^{(1)}$ is an eigenstate of $H_2$ with the eigenvalue $E_n$, because $H_2f_n^{(2)} = H_2 Uf_n^{(1)}=UH_1f_n^{(1)} = E_nUf_n^{(1)}= E_nf_n^{(2)}$. If this is possible for every  eigenfunction $f_n^{(1)}$, then $H_1$ and $H_2$ are isospectral to each other. If $U$ is invertible, $H_1$ and $H_2$ are also {\em similar} as in $H_2=UH_1U^{-1}$. If $U$ is invertible but not unitary, it can easily happen that $H_1$ is Hermitian while $H_2$ is not. SUSY-QM works in a slightly different way, relating the eigenvectors and the eigenvalues of two factorized Hamiltonians $H=A^\dagger A$ and $H_\mathrm{susy}=AA^\dagger$, where $A$ and $A^\dagger$ are not required to be ladder operators. The way in which SUSY is useful is that, knowing the details of $H$, one can find many properties of $H_\mathrm{susy}$.

In the present paper we propose a general strategy of producing a virtually infinite chain of not necessarily Hermitian Hamiltonians which are all isospectral  (note that in the present paper, the terms {\em self-adjoint} and {\em Hermitian} will be used as synonymous),  or two related (again, virtually infinite) chains of non-Hermitian Hamiltonians whose discrete spectra behave as follows: they coincide for all Hamiltonians in the same chain, while they are complex conjugate of those of the Hamiltonians of the other chain. In both cases, however, their eigenvectors are all related to each other, and  can be easily found. Our strategy starts from a single seed Hamiltonian $H$, not necessarily Hermitian, with the only assumption that its eigenvalues have multiplicity one, namely no degeneracy. 
To simplify the mathematical aspects of the present paper without losing the physical relevance of our results, we work here only with finite-dimensional Hilbert spaces. As is well known, this is a standard way to avoid dealing with unbounded operators, which would require much more mathematical care; we will consider the issue in a future paper. 

Let us recall how and when, out of a given $H$, we can construct two biorthogonal bases of eigenvectors of $H$ and $H^\dagger$, namely the Dirac adjoint of $H$.
Note that $H$ is generally non-Hermitian here, and hence the so-called right- and left-eigenvectors are not Hermitian conjugate but biorthogonal to each other.
We will use the biorthogonal bases to construct two pieces of self-adjoint and positive operators, one being the inverse of the other, which allow us to define two new types of scalar products and two related adjoint maps. This itself is in fact well known in the literature, where the adjoint maps are strictly connected to {\em metric operators}. What is not known in our knowledge is that we can use the metric operators to generate a series of new metric operators sequentially. Once $H$ and its eigenvectors and eigenvalues are known, the procedure starts producing new isospectral, generally non-Hermitian Hamiltonians. We describe this in Sec.~\ref{sect2}, considering both real and complex eigenvalues of $H$.

Since the publication of Refs.~\cite{Hatano96,BB98}, there is much interest in Hamiltonians which are {\bf not} self-adjoint but are physically motivated, and admit real eigenvalues and a transition to a phase of complex eigenvalues. 
Historically, non-Hermitian Hamiltonians were first introduced in nuclear theory in order to describe decay processes~\cite{Gamow28,Feshbach58,Peierls59}.
In the second half of 1990's, two epoch-making non-Hermitian Hamiltonians were introduced; namely a generalized Anderson model called the Hatano-Nelson model~\cite{Hatano96,Hatano97,Hatano98} and the $PT$-symmetric systems~\cite{BB98,benbook,ben}.
The former Hamiltonian has non-Hermiticity in off-diagonal elements, while the latter has it in diagonal ones.
The former was originally introduced as an effective model mapped from a statistical-physical model of type-II superconductors, but later used to analyze the Anderson localization.
It was revealed that the localized eigenstates have real eigenvalues while the extended ones have complex eigenvalues.
The latter was originally introduced to replace the Hermiticity as a condition to produce real-valued energy eigenvalues, but later triggered various macroscopic and optic experiments with device applications in sight.
The real eigenvalues signify stationary flows of probability from a source to a sink, while the complex eigenvalues indicate exponential growth and decay.
The last couple of years saw an explosive boom of studies on non-Hermitian systems of a wide variety, including non-Hermitian topological systems~\cite{Gong18,Kawabata19} and $PT$-symmetric Floquet theory~\cite{Turker18,Harter20}.

In light of this situation, it is of importance to classify non-Hermitian Hamiltonians.
The present study implies that there are an infinite number of isospectral non-Hermitian Hamiltonians associated with one single Hamiltonian.
We will show in Sec.~\ref{sectexe} various examples for which the newly generated operators are almost always different from the previous ones. After a preliminary application, meant to show the non-triviality of our procedure, we apply our strategy to several versions of the Hatano-Nelson model \cite{Hatano96,Hatano97}, and in particular to an open clean system with asymmetric hopping, and to its closed version. We also discuss what happens when impurities are considered. A PT-symmetric model~\cite{BB98,benbook,ben} is also discussed. In all these cases, we clearly see that the procedure is non-trivial and we discuss its physical meaning. 
Our conclusions are given in Section \ref{sectconcl}.


\section{The general framework}\label{sect2}

\subsection{Starting point: the Dirac adjoint}

We consider a system residing in an $N$-dimensional Hilbert space $\Hil$, where $N<\infty$, with the scalar product $\scp<.,.>$, which satisfies the standard equalities 
\begin{align}
\overline{\scp<f,g>}=\scp<g,f>,\qquad
\scp<f,\alpha g+\beta h>=\alpha\scp<f,g>+\beta\scp<f,h>, \qquad
\scp<f,f>\ge 0, 
\end{align}
for all $f,g,h\in\Hil$ and for all $\alpha, \beta\in\mathbb{C}$, the last of which yields the corresponding norm $\norm{f}:=\sqrt{\scp<f,f>}$.

Let a non-self-adjoint Hamiltonian $H$ with $H\neq H^\dagger$ drive the system.
Here $H^\dagger$ is the Dirac adjoint of $H$, which is defined by $\scp<.,.>$ as in $\scp<Hf,g>=\scp<f,H^\dagger g>$, $^\forall f,g\in\Hil$. 
Note that this adjoint is generally an involution because $\qty(H^\dag)^\dag=H$.
For concreteness, if the scalar product $\scp<.,.>$ is the ordinary product in $\Hil=\mathbb{C}^N$, which is $\scp<f,g>=\sum_{l=1}^{N}\overline{f_l}\,g_l$ with obvious notation, then the adjoint $H^\dagger$ is just the transpose and complex conjugate of the matrix representing $H$. We denote the set of all the bounded operators on $\Hil$ by $B(\Hil)$. Notice that, since $\Hil$ is finite-dimensional, $B(\Hil)$ is nothing but the set of all the matrices on $\Hil$.

Hereafter we will generate various types of adjoint map by generating various types of scalar product, starting with the simplest scalar product $\scp<.,.>$ and the Dirac adjoint.
Our only assumptions all along the present paper are that $H$ is diagonalizable and that all the eigenvalues have multiplicity one:
\begin{align}
H\vno=E_n\vno, 
\label{21}\end{align}
$n=1,2,\ldots,N$, where $E_n\neq E_m$ if $n\neq m$. 
(Note that the eigenvalues can be real or complex.)
Consequently, the set of eigenstates $\F_\varphi\supzero=\{\vno\}$ is a basis for $\Hil$. 

Since $H\neq H^\dagger$, the basis $\F_\varphi\supzero$ needs not be an orthogonal one. 
However, it uniquely admits another basis $\F_\psi\supzero=\{\pno\}$, which satisfies the biorthogonality $\scp<\vno,\pmo>=\delta_{n,m}$ and 
\begin{align}
\label{eq2.3}
f=\sum_{n=1}^{N}\scp<\vno,f>\pno = \sum_{n=1}^{N}\scp<\pno,f>\vno
\end{align}
for all $f\in\Hil$; see Ref.~\cite{chri}.  
We will often rewrite \eqref{eq2.3} as 
\begin{align}
\sum_{n=1}^{N}\dyad{\vno}{\pno}=\sum_{n=1}^{N}\dyad{\pno}{\vno}=\1,
\label{22}\end{align}
where $\dyad{f}{g}$ is an operator defined by $(\dyad{f}{g})h:=\scp<g,h>f$, $^\forall f,g,h\in\Hil$. 
Formula~\eqref{22} is known as a {\em resolution of the identity}.
It is known that $\F_\psi\supzero$ is the set of eigenstates of $H^\dagger$:
\begin{align}
H^\dagger\pno=\overline{E}_{n}\,\pno.
\label{23}\end{align}

The mutually biorthogonal sets $\F_\varphi\supzero$ and  $\F_\psi\supzero$ define two operators which will play a fundamental role in our construction below:
\begin{align}
S_\varphi\supzero:=\sum_{n=1}^{N}\dyad{\vno}{\vno},
\qquad
S_\psi\supzero:=\sum_{n=1}^{N}\dyad{\pno}{\pno}.
\label{24}\end{align}
They are both self-adjoint because \textit{e.g.}\
\begin{align}
\scp<f,\qty(S_\varphi\supzero)^\dagger g>&=\scp<S_\varphi\supzero f,g>
=\sum_{n=1}^N\scp<\scp<\vno,f>\vno,g>
=\sum_{n=1}^N\overline{\scp<\vno,f>}\scp<\vno,g>
\nonumber\\
&=\sum_{n=1}^N\scp<f,\vno>\scp<\vno,g>
=\sum_{n=1}^N\scp<f,\scp<\vno,g>\vno>
=\scp<f,S_\varphi\supzero g>,\quad {}^\forall f,g\in\Hil ,
\end{align}
and strictly positive because \textit{e.g.}\
\begin{align}
\scp<f,S_\varphi\supzero f>=\sum_{n=1}^N \scp<f,\scp<\vno,f>\vno>
=\sum_{n=1}^N \scp<\vno,f>\scp<f,\vno>=\sum_{n=1}^N \abs{\scp<\vno,f>}^2>0,
\end{align}
for all $f\neq0$. 
Moreover, owing to the biorthogonality of the families, we have
\begin{align}
S_\varphi\supzero\pno=\vno, 
\qquad 
S_\psi\supzero\vno=\pno, 
\qquad\mbox{and}\qquad 
S_\varphi\supzero=\qty(S_\psi\supzero)^{-1}.
\label{25}
\end{align}

\subsection{First iteration generating two new adjoints}
\label{subsec22}

We can use the operators~\eqref{24} to define two new scalar products on $\Hil$,
\begin{align}
\scp<f,g>_{\varphi(0)}:=\scp<S_\varphi\supzero f,g>, 
\qquad 
\scp<f,g>_{\psi(0)}:=\scp<S_\psi\supzero f,g>,
\label{26}\end{align}
and the two corresponding adjoint maps, 
\begin{align}
\scp<Xf,g>_{\varphi(0)}=\scp<f,X^{\flat_0}g>_{\varphi(0)}, 
\qquad 
\scp<Xf,g>_{\psi(0)}=\scp<f,X^{\sharp_0}g>_{\psi(0)},
\quad {}^\forall f,g\in\Hil.
\label{27}\end{align}
The fact that $\scp<f,g>_{\varphi(0)}$ and $\scp<f,g>_{\psi(0)}$ are both scalar products that satisfy the standard equalities follows from the properties of $S_\varphi\supzero$ and $S_\psi\supzero$; for instance,
\begin{align}
\overline{\scp<f,g>_{\varphi(0)}} =\overline{\scp<S_\varphi\supzero f,g>}=\scp<g,S_\varphi\supzero f>=\scp<S_\varphi\supzero g,f>=\scp<g,f>_{\varphi(0)}.
\end{align}
The norms
\begin{align}
\norm{f}_{\varphi(0)}:=\sqrt{\scp<f,f>_{\varphi(0)}},
\qquad
\norm{f}_{\psi(0)}:=\sqrt{\scp<f,f>_{\psi(0)}}
\end{align}
are both equivalent to $\norm{f}$ because \textit{e.g.}\
\begin{align}
\|f\|_{\varphi(0)}^2=\langle S_\varphi\supzero f,f\rangle=\langle (S_\varphi\supzero)^{1/2} f,(S_\varphi\supzero)^{1/2}f\rangle=\|(S_\varphi\supzero)^{1/2}f\|^2\leq \|(S_\varphi\supzero)^{1/2}\|^2\|f\|^2;
\end{align}
in a similar way we can find a reverse inequality between $\norm{f}_{\varphi(0)}$ and $\norm{f}$, and hence $\|f\|_{\varphi(0)}=\|(S_\varphi\supzero)^{1/2}\|\|f\|\propto\|f\|$.

On the other hand, simple computations show that the two adjoint maps in Eq.~\eqref{27} are generally different operations. 
We find \textit{e.g.}\ from
\begin{align}
&\scp<f,S_\psi\supzero X^\dagger S_\varphi\supzero g>_{\varphi(0)}
=\scp<S_\varphi\supzero XS_\psi\supzero S_\varphi\supzero f, g>
=\scp<S_\varphi\supzero Xf, g>=\scp<Xf,g>_{\varphi(0)}=\scp<f,X^{\flat_0}g>_{\varphi(0)}
\end{align}
for all operator $X$ on $\Hil$, namely $^\forall X\in B(\Hil)$, and  $^\forall f,g\in \Hil$, that
\begin{align}
X^{\flat_0}=S_\psi\supzero X^\dagger S_\varphi\supzero, \qquad X^{\sharp_0}=S_\varphi\supzero X^\dagger S_\psi\supzero,
\label{28}\end{align}
which together  imply the following relations:
\begin{align}
&X^{\flat_0}=\qty(S_\psi\supzero)^2 X^{\sharp_0} \qty(S_\varphi\supzero)^2,
\\
&\qty(X^{\flat_0})^\dagger=\qty(X^\dagger)^{\sharp_0}, 
\qquad 
\qty(X^{\sharp_0})^\dagger=\qty(X^\dagger)^{\flat_0}.
\label{29}\end{align}

Formula \eqref{29} shows that the ordering of the adjoint maps is not commutative. 
However, there is a specific case in which the ordering becomes irrelevant, as the following lemma shows:
\begin{lemma}\label{lemma1}
	The following statements are equivalent: 
	(i) $(X^{\flat_0})^\dagger=\qty(X^\dagger)^{\flat_0}$; 
	(ii) $\qty[X,\qty(S_\varphi\supzero)^2]=0$; 
	(iii) $\qty[X,\qty(S_\psi\supzero)^2]=0$; 
	(iv) $\qty(X^{\sharp_0})^\dagger=\qty(X^\dagger)^{\sharp_0}$.
\end{lemma}
The proof is easy and will not be given here. 
Lemma~1 implies the specialty of the pseudo-Hermiticity~\cite{mosta}, which we will come back to later.

Let us note that the eigenstates $\vno$ and $\pno$ of $H$ and $H^\dagger$ are also the eigenstates of new Hamiltonians constructed out of $H$, using these alternative adjoints.
More in details, we have
\begin{align}
\qty(H^{\sharp_0})^\dagger\pno=E_n\pno, 
\qquad
H^{\sharp_0}\vno=\overline{E}_{n}\,\vno,
\label{29bis}\end{align} 
for all $n$. 
This means that the pair of $H$ and $\qty(H^{\sharp_0})^\dagger$ are isospectral to each other as well as the pair of $H^\dagger$ and $H^{\sharp_0}$ are.

We can attribute this isospectrality to the fact that each pair of Hamiltonians is related by a simple intertwining relation that preserves the eigenvalues~\cite{intop1,intop2,intop3,bagintop}, at least in finite-dimensional Hilbert spaces. 
For instance, if we multiply from the left of the second equality in \eqref{28}  by $S_\psi\supzero$ for $X=H$ and  use the last identity in \eqref{25}, we find the intertwining relations
\begin{align}\label{218-1}
S_\psi\supzero H^{\sharp_0}=H^\dagger S_\psi\supzero,
\qquad
\mbox{or equivalently}
\quad
H^{\sharp_0}S_\varphi\supzero=S_\varphi\supzero H^\dagger,
\end{align}
which is indeed followed by
\begin{align}
H^{\sharp_0}\vno=
\qty(S_\psi\supzero)^{-1} H^\dagger S_\psi\supzero\vno=S_\varphi\supzero H^\dagger\pno=\overline{E}_{n}\,S_\varphi\supzero\pno=\overline{E}_{n}\,\vno.
\end{align}
We also have the following similar intertwining relations:
\begin{align}
\qty(H^{\sharp_0})^\dag S_\psi\supzero=S_\psi\supzero H,
&\qquad
S_\varphi\supzero\qty(H^{\sharp_0})^\dag =HS_\varphi\supzero,
\\
S_\varphi\supzero H^{\flat_0}=H^\dag S_\varphi\supzero,
&\qquad
H^{\flat_0}S_\psi\supzero=S_\psi\supzero H^\dag ,
\\\label{222-1}
\qty(H^{\flat_0})^\dag S_\varphi\supzero=S_\varphi\supzero H,
&\qquad
S_\psi\supzero\qty(H^{\flat_0})^\dag =HS_\psi\supzero,
\end{align}
the last two of which implies that $H$ and $\qty(H^{\flat_0})^\dag$ are isospectral to each other as well as $H^\dag$ and $H^{\flat_0}$ are.

In fact, the Hamiltonian is called pseudo-Hermitian~\cite{mosta} and all its eigenvalues are real if it satisfies the specific intertwining relation
\begin{align}
S_\psi\supzero H=H^\dag S_\psi\supzero.
\end{align}
This implies that the pseudo-Hermiticity is nothing but the self-adjointness with respect to $\sharp_0$ as in 
\begin{align}\label{223-1}
H=H^{\sharp_0}.
\end{align} 
We will probe more into this in the next subsection.

\vspace{0.5\baselineskip}

\noindent
{\bf Remark--} Notice that if $H=H^\dagger$, then $\F_\varphi\supzero$ is already an orthonormal basis, so that $\vno=\pno$ for all $n$.
Therefore $S_\varphi\supzero=S_\psi\supzero=\1$, and hence all the scalar products and the adjoint maps collapse to the standard ones.
We will also show in Subsect.~\ref{sec3.2.2} below an instance of non-Hermitian Hamiltonian for which $\F_\varphi\supzero=\F_\psi\supzero$ and $S_\varphi\supzero=S_\psi\supzero=\1$ despite that the eigenvalues are mostly complex.

\vspace{0.5\baselineskip}

Let us summarize what we have obtained throughout the present subsection in a slightly different presentation.
Out of the two Hamiltonians, one and its $\dag$-adjoint, 
\begin{align}
H\vno&=E_n\vno, 
\qquad
H^\dagger\pno=\overline{E}_{n}\,\pno, 
\end{align}
we generated two new Hamiltonians, one new adjoint of $H$ ($\sharp_0$) and one of $H^\dag$ ($\flat_0$);
the former is isospectral to $H^\dag$ with the eigenvectors $\vno$ and the latter is to $H$ with $\pno$:
\begin{align}
H^{\sharp_0}\vno&=\overline{E}_{n}\,\vno,
\qquad
\qty(H^\dagger)^{\flat_0} \pno=E_n\pno.
\label{223-2}\end{align} 
We will show in Subsect.~\ref{subsec24} that $H^{\flat_0}$ and $\qty(H^\dagger)^{\sharp_0}$ are also isospectral to $H^\dag$ and $H$, respectively, with new set of eigenvectors.
We will then open up a door to a whole new series of non-Hermitian Hamiltonians.

\vspace{0.5\baselineskip}

\noindent
{\bf Remark--} It may be worth stressing that what we are doing here is not just constructing matrices which are similar to $H$ or to $H^\dagger$, but to do this by means of some {\bf very specific} operators, as in (\ref{28}), for instance, $H^{\flat_0}=S_\psi\supzero H^\dagger S_\varphi\supzero$. This is a big difference, since we only work with a single ingredient, $H$. We do not use, in our construction above and below, any other mathematical object which is not connected to the original $H$.

\subsection{Reality of the eigenvalues and the pseudo-Hermiticity}

Before going to introducing more new Hamiltonians, let us discuss the reality of the eigenvalues.
As we mentioned at the beginning, the eigenvalues of $H$ can be real or complex. 
The following theorem gives equivalent conditions for every piece of $E_n$ to be real specifically in the case of the pseudo-Hermiticity \eqref{223-1}.

\begin{thm}\label{thm1}
	The following statements are equivalent: 
	(i) $E_n\in\mathbb{R}$\, for all $n$; 
	(ii) $S_\psi\supzero H=H^\dagger S_\psi\supzero$; 
	(iii) $S_\varphi\supzero H^\dagger=HS_\varphi\supzero$; 
	(iv) $H=H^{\sharp_0}$; 
	(v) $H^\dagger=\qty(H^\dagger)^{\flat_0}$.
\end{thm}

\begin{proof}
	We first prove that (i) implies (ii). Indeed we have $S_\psi\supzero H\vno=E_n\pno$, while $H^\dagger S_\psi\supzero\vno=\overline{E}_{n}\,\pno$, $^\forall n$. By using the expansion $f=\sum_{n=1}^{N}\scp<\pno,f>\,\vno$, we therefore have $(S_\psi\supzero H-H^\dagger S_\psi\supzero)f=0$, $^\forall f\in\Hil$ if $E_n=\overline{E}_{n}$ for all $n$, which implies (ii).
	
	Let us next assume that (ii) is true and show that $E_n\in\mathbb{R}$ for all $n$. Indeed we have
	\begin{align}
	\overline{E}_{n}\,\pno=H^\dagger S_\psi\supzero\vno=S_\psi\supzero H\vno=E_n\pno,
	\end{align}
	for all $n$. Hence $E_n=\overline{E}_{n}$.
	
	The equivalence between (ii) and (iii) is obvious, while the equivalence between (iii) and (iv) can be proved by using \eqref{28} as follows:
	\begin{align}
	H^{\sharp_0}=H
	\quad \Leftrightarrow \quad 
	S_\varphi\supzero H^\dagger S_\psi\supzero=H 
	\quad \Leftrightarrow \quad 
	S_\varphi\supzero H^\dagger=HS_\varphi\supzero.
	\end{align}
	The equivalence between (iv) and (v) is trivial.
\end{proof}

This theorem states that even if the Hamiltonian $H$ of some system is not self-adjoint with the standard Dirac adjoint $\dag$, it can still have only real eigenvalues if $H$ is self-adjoint with respect to the adjoint $\sharp_0$. 

It is now natural to ask if this claim can be inverted. In other words: if at least one eigenvalue of $H$, say $E_{n_0}$, is complex, is it possible to define a new adjoint $*$ (not necessarily coincident with $\dagger$, $\flat_0$ or $\sharp_0$) such that $H=H^*$? 

The following argument shows that this is not the case: if a single eigenvalue of $H$ is complex, then it is \textit{not} possible to have the self-adjointness $H=H^*$ for any possible adjoint map $*$. 
Suppose that this were possible; hence $H=H^*$. 
Let $\scpp<.,.>$ be the scalar product on $\Hil$ defining $*$:  $\scpp< Xf,g>= \scpp< f,X^*g>$, $^\forall f,g\in\Hil$. 
If we set both $f$ and $g$ to the eigenstate $\varphi_{n_0}$ of $H$ corresponding to the complex eigenvalue $E_{n_0}$ (as in $H\varphi_{n_0}=E_{n_0}\varphi_{n_0}$), we would have
\begin{align}
E_{n_0}\scpp< \varphi_{n_0},\varphi_{n_0}> = \scpp< \varphi_{n_0},H\varphi_{n_0}> = \scpp<  H^*\varphi_{n_0},\varphi_{n_0}> = \scpp<  H\varphi_{n_0},\varphi_{n_0}> =\overline{E}_{{n_0}}\,\scpp< \varphi_{n_0},\varphi_{n_0}>,
\end{align}
which is impossible, since we have assumed $E_{n_0}\neq \overline{E}_{{n_0}}$ while $\scpp< \varphi_{n_0},\varphi_{n_0}>\neq0$. 
Summarizing, we can state that {\em a non-self-adjoint Hamiltonian $H$ can have only real eigenvalues if and only if it is possible to introduce an adjoint with respect to which $H$ becomes self-adjoint.}

This statement is relevant to the $PT$-symmetric non-Hermitian quantum mechanics.
In the $PT$-symmetric quantum mechanics, it has been argued that there should be a so-called $C$-operator~\cite{CMB02} that gives a scalar product in the form $\scp<f,g>_{CPT}:=\scp<CPTf,g>$ and the corresponding $CPT$-conjugation. 
Our argument above implies that the $CPT$-conjugation exists only in the $PT$-unbroken phase, where all eigenvalues are real because of the $PT$-symmetry. However, as we note in the following two remarks as well as in Subsects.~\ref{subusubsec3.2.1} and~\ref{sect3.3.2} below, something interesting can also be stated when all the eigenvalues are pure imaginary.

\vspace{0.5\baselineskip}

\noindent
{\bf Remark 1:--} It is interesting to see what happens when all the eigenvalues, rather than being reals, are purely imaginary. This is relevant in view of some of the examples considered later in Subsects.~\ref{subusubsec3.2.1} and~\ref{sect3.3.2}. 
Suppose that $h\neq h^\dagger$ is a Hamiltonian with purely imaginary eigenvalues $i\epsilon_n$, where $\epsilon_n\in\mathbb{R}$ for all $n$, as in  $h\vno=i\epsilon_n\vno$ and $h^\dagger\pno=-i\epsilon_n\pno$. Then, for a new Hamiltonian $H=-ih$, all eigenvalues are real and its eigenvectors coincide with those of $h$ and $h^\dagger$, as in $H\vno=\epsilon_n\vno$ and $H^\dagger\pno=\epsilon_n\pno$. We can then use Theorem \ref{thm1} to deduce the following properties satisfied by $h$ and $h^\dagger$:
\begin{align}
S_\psi\supzero h=-h^\dagger S_\psi\supzero, \quad S_\varphi\supzero h^\dagger=-hS_\varphi\supzero, \quad h=-h^{\sharp_0}, \quad h^\dagger=-\qty(h^\dagger)^{\flat_0}.
\label{add1}
\end{align}
We stress that these equalities, and those in Theorem \ref{thm1}, do not hold if at least one of the eigenvalues of $h$ (or $H$) has non-zero real and imaginary parts.

\vspace{0.5\baselineskip}

\noindent
{\bf Remark 2:--} Let us consider again the case in which all the eigenvalues are purely imaginary: $h\vno=i\epsilon_n\vno$. This is followed by $h^2\vno=-\epsilon_n^2\vno$. It is interesting to consider the inverse implication: given an operator $h$ such that its square has only negative eigenvalues, $h^2\vno=-\epsilon_n^2\vno$, does  $h$ have only purely imaginary eigenvalues? We can indeed show something similar. Our assumption implies that $-h^2$ has only positive eigenvalues, and hence $-h^2$ is a positive operator. Therefore, a unique positive operator $B$ exists such that $B^2=-h^2$; see Ref.~\cite{rs}. 
The operator $B$ commutes with all the (bounded) operators that commute with $h^2$, and hence $[B,h^2]=-[B,B^2]=0$. Furthermore, defining $\tilde h=iB$, we have $[\tilde h,h^2]=0$. Let us now assume that the multiplicity of each $-\epsilon_n^2$ is one. Then, since $h^2\tilde h\vno=\tilde h h^2\vno=-\epsilon_n^2\tilde h\vno$, the state $\tilde h\vno$ is an eigenstate of $h^2$ with eigenvalue $-\epsilon_n^2$. Hence, $\tilde h\vno$ is necessarily proportional to $\vno$: $\tilde h\vno=\alpha_n\vno$, and $B\vno=-i\alpha_n\vno$. Then we have 
	$$
	\epsilon_n^2\vno=-h^2\vno=B^2\vno=B(-i\alpha_n\vno)=-\alpha_n^2\vno,
	$$
	which implies that $\alpha_n=\pm i|\epsilon_n|$. This shows that all the eigenvalues of $\tilde h$ are purely imaginary. Of course, each choice of sign in the eigenvalues give rise to a possible operator $h$ whose square is as above. In other words, given $h^2$, uniqueness of its square root is obtained only under further assumptions. This is in agreement with a well known result on the square roots of positive bounded operators, which are unique only after the positivity condition is imposed~\cite{rs}. However, all different square roots have purely imaginary eigenvalues.
	

\subsection{Two more isospectral Hamiltonians, four in total}
\label{subsec24}

We have seen in \eqref{25} that the operators $S_\varphi\supzero$ maps $\F_\psi\supzero$ into $\F_\varphi\supzero$ and $S_\psi\supzero$ maps in the opposite way.
If we write \eqref{223-2} in the forms
\begin{align}
H^{\sharp_0}S_\varphi\supzero\pno&=\overline{E}_{n}\,S_\varphi\supzero\pno,
\qquad
\qty(H^\dagger)^{\flat_0} S_\psi\supzero\vno=E_n S_\psi\supzero\vno,
\end{align} 
we realize that $S_\varphi\supzero\pno$ and $S_\psi\supzero\vno$ are the eigenvectors of the new Hamiltonians introduced in Subsect.~\ref{subsec22}.
This is indeed what has been often observed in the literature; see Refs.~\cite{BB98,benbook,ben,baginbagbook,mosta} and references therein. 
Our new point is to see what happens if we apply, for instance, $S_\varphi\supzero$ to $\vno$ instead of to $\pno$. 
We will now show that such states span eigenspaces of $H^{\flat_0}$ and $\qty(H^\dagger)^{\sharp_0}$.

For brevity, let us put
\begin{align}
\vnuno:=S_\varphi\supzero\vno, \qquad \pnuno:=S_\psi\supzero\pno,
\label{210}\end{align}
and $\F_\varphi\supone =\{\vnuno\}$,  $\F_\psi\supone =\{\pnuno\}$. 
The pair $\qty(\F_\varphi\supone ,\F_\psi\supone )$ has similar properties to the original pair $\qty(\F_\varphi\supzero,\F_\psi\supzero)$. 
In fact, we can easily prove that it is a biorthonormal basis for $\Hil$:
$\scp<\vnuno,\pmuno>=\delta_{n,m}$ and $f=\sum_{n=1}^{N}\scp<\vnuno,f>\,\pnuno = \sum_{n=1}^{N}\scp<\pnuno,f>\,\vnuno$, $^\forall f\in\Hil$. 
For example, we have
\begin{align}
\scp<\vnuno,\pmuno>=\scp<S_\varphi\supzero\vno,S_\psi\supzero\pmo>=\scp<\vno,S_\varphi\supzero S_\psi\supzero\pmo>=\scp<\vno,\pmo>=\delta_{n,m},
\end{align}
and
\begin{align}
\sum_{n=1}^{N}\scp<\vnuno,f>\,\pnuno=S_\psi\supzero\sum_{n=1}^{N}\scp<\vno,S_\varphi\supzero f>\,\pno=S_\psi\supzero S_\varphi\supzero f=f,
\end{align}
for all $f\in\Hil$. As before, we write it  more concisely in the form of a resolution of the identity:
\begin{align}
\sum_{n=1}^{N}\dyad{\vnuno}{\pnuno}=\sum_{n=1}^{N}\dyad{\pnuno}{\vnuno}=\1.
\label{211}\end{align}

We are now in a position to prove that $\vnuno$ and $\pnuno$ are eigenstates of the other adjoints of $H$. 
Indeed, we have
\begin{align}
\qty(H^{\flat_0})^\dagger\vnuno=E_n\vnuno, 
\qquad 
H^{\flat_0}\pnuno=\overline{E}_{n}\,\pnuno,
\label{212}\end{align}
for all $n$. 
For example, using \eqref{28}, we have
\begin{align}
\qty(H^{\flat_0})^\dagger\vnuno=S_\varphi\supzero HS_\psi\supzero S_\varphi\supzero\vno=S_\varphi\supzero H\vno=E_nS_\varphi\supzero\vno=E_n\vnuno.
\end{align}
We can prove the second eigenvalue equation in \eqref{212} in a similar way. 

Equations~\eqref{212} combined with \eqref{29bis} show that by applying $S_\varphi\supzero$ and $S_\psi\supzero$ to $\vno$ and $\pno$, we can introduce four sets of eigenspace for four new Hamiltonians in total, two adjoints of $H$ isospectral to $H^\dag$ and two of $H^\dag$ isospectral to $H$; see 
Table~\ref{tab1}. 

As before in Eqs.~\eqref{218-1}-\eqref{222-1}, these features are consequences of the existence of several intertwining relations between the operators involved. 
In the present finite-dimensional situation, these intertwining equations are similarity conditions, since all the intertwining operators are invertible. 
However, this invertibility could be lost when we try to extend our settings to an infinite-dimensional system, while the intertwining relations usually remain true (modulo, at most, some constraints on the domains of the operators~\cite{baginbagbook}). We hope to consider this aspect in a future paper.

To summarize so far, we defined two new adjoints $\sharp_0$ and $\flat_0$, with which we generated four new Hamiltonians in total, the $\sharp_0$-adjoint and the $\flat_0$-adjoint of $H$, namely $H^{\sharp_0}$ and $H^{\flat_0}$, and those of $H^\dag$,
namely $\qty(H^\dag)^{\sharp_0}=\qty(H^{\flat_0})^\dag$ and $\qty(H^\dag)^{\flat_0}=\qty(H^{\sharp_0})^\dag$.
The former two are isospectral to $H^\dag$, while the latter two are isospectral to $H$.
This is summarized in 
Table~\ref{tab1}, together with new isospectral Hamiltonians we will find below.
\begin{table}
\caption{New isospectral Hamiltonians with their eigenvectors that we introduce down to the third iteration.
The Hamiltonians listed on the top half are all isospectral to $H$, while those on the bottom half to $H^\dag$.
On the first generation, only two sets of eigenvectors of the four Hamiltonians are new, namely $\{\varphi_n^{(1)}\}$ and $\{\psi_n^{(1)}\}$, hence producing only four new Hamiltonians in the second generation. 
This anomaly corresponds to the existence of the pseudo-Hermiticity $H=H^{\sharp_0}$.
Each of the four Hamiltonians in the second generation has its own new set of eigenvectors, hence producing eight new Hamiltonians in the third generation.
Sixteen new Hamiltonians are in the fourth generation, which we left out from the table.}
\label{tab1}
\begin{center}
\begin{tabular}{lclll}
Eigenvalues & generation & Hamiltonian & Eigenvectors & Another notation\\
\hline
$\{E_n\}$ & ---& $H$ & $\{\vno\}$ & \\
\cline{2-5}
& 1 & $\qty(H^\dag)^{\flat_0}=\qty(H^{\sharp_0})^\dag$ & $\{S_\psi\supzero\vno\}$ &  $=\{\psi_n\supzero\}$  \\
\cline{2-5}
& 1 & $\qty(H^\dag)^{\sharp_0}=\qty(H^{\flat_0})^\dag$ & $\{S_\varphi\supzero\vno\}$ & $=\{\varphi_n\supone\}$\\
\cline{2-5}
& 2 & $\qty(H^\dag)^{\flat_1}=\qty(H^{\sharp_1})^\dag$ & $\{S_\psi\supone\vno\}$ &  $=\{\varphi_n^{(2,\alpha)}\}$\\
\cline{2-5}
& 2 & $\qty(H^\dag)^{\sharp_1}=\qty(H^{\flat_1})^\dag$ & $\{S_\varphi\supone\vno\}$ &  $=\{\varphi_n^{(2,\beta)}\}$ \\
\cline{2-5}
& 3 & $\qty(H^\dag)^{\flat_{2,\alpha}}=\qty(H^{\sharp_{2,\alpha}})^\dag$ & $\{S_\psi^{(2,\alpha)}\vno\}$ & $=\{\varphi_n^{(3,\alpha\alpha)}\}$ \\
\cline{2-5}
& 3 & $\qty(H^\dag)^{\sharp_{2,\alpha}}=\qty(H^{\flat_{2,\alpha}})^\dag$ & $\{S_\varphi^{(2,\alpha)}\vno\}$ & $=\{\varphi_n^{(3,\alpha\beta)}\}$ \\
\cline{2-5}
& 3 & $\qty(H^\dag)^{\flat_{2,\beta}}=\qty(H^{\sharp_{2,\beta}})^\dag$ & $\{S_\psi^{(2,\beta)}\vno\}$ & $=\{\varphi_n^{(3,\beta\alpha)}\}$\\
\cline{2-5}
& 3 & $\qty(H^\dag)^{\sharp_{2,\beta}}=\qty(H^{\flat_{2,\beta}})^\dag$ & $\{S_\varphi^{(2,\beta)}\vno\}$ & $=\{\varphi_n^{(3,\beta\beta)}\}$\\
\hline
$\{\overline{E}_{n}\}$ & --- & $H^\dag$ & $\{\psi_n\supzero\}$ &\\
\cline{2-5}
& 1 & $H^{\sharp_0}$ & $\{S_\varphi\supzero\pno\}$ & $=\{\varphi_n\supzero\}$ \\
\cline{2-5}
& 1 & $H^{\flat_0}$ & $\{S_\psi\supzero\pno\}$ & $=\{\psi_n\supone\}$  \\
\cline{2-5}
& 2 & $H^{\sharp_1}$ & $\{S_\varphi\supone\pno\}$&  $=\{\psi_n^{(2,\alpha)}\}$  \\
\cline{2-5}
& 2 & $H^{\flat_1}$ & $\{S_\psi\supone\pno\}$ &  $=\{\psi_n^{(2,\beta)}\}$ \\
\cline{2-5}
& 3 & $H^{\sharp_{2,\alpha}}$ & $\{S_\varphi^{(2,\alpha)}\pno\}$  & $=\{\psi_n^{(3,\alpha\alpha)}\}$\\
\cline{2-5}
& 3 & $H^{\flat_{2,\alpha}}$ & $\{S_\psi^{(2,\alpha)}\pno\}$  & $=\{\psi_n^{(3,\alpha\beta)}\}$ \\
\cline{2-5}
& 3 & $H^{\sharp_{2,\beta}}$ & $\{S_\varphi^{(2,\beta)}\pno\}$  & $=\{\psi_n^{(3,\beta\alpha)}\}$ \\
\cline{2-5}
& 3 & $H^{\flat_{2,\beta}}$ & $\{S_\psi^{(2,\beta)}\pno\}$  & $=\{\psi_n^{(3,\beta\beta)}\}$ \\
\hline
\end{tabular}
\end{center}
\end{table}

\subsection{Second iteration of generating two new adjoints}

The biorthonormal basis  $\qty(\F_\varphi\supone ,\F_\psi\supone )$ can now be used, in analogy with  $\qty(\F_\varphi\supzero,\F_\psi\supzero)$, to define two new operators $S_\varphi\supone $ and $S_\psi\supone $, two new scalar products $\scp<.,.>_{\varphi(1)}$ and $\scp<.,.>_{\psi(1)}$ (with their related norms), and two new adjoint maps $\sharp_1$ and $\flat_1$. 
As we will show, these adjoint maps then give rise to four new Hamiltonians with known eigensystems. 
The procedure can be iterated further, producing even more new Hamiltonians whose eigenvalues and eigensystems can also be easily derived.
In each iteration, we define new metric operators out of eigenvectors in the previous iteration, define new scalar products out of the metric operators, and thereby define new adjoints. 
The eigenvectors of the new adjoints of the Hamiltonians then produce further new metric operators.
As is indicated in 
Table~\ref{tab1}, the number of new Hamiltonians is doubled in every iteration down from the second generation.
We track the procedure one by one in the following.

First, let us put
\begin{align}
S_\varphi\supone :=\sum_{n=1}^{N}\dyad{\vnuno}{\vnuno}, 
\qquad 
S_\psi\supone :=\sum_{n=1}^{N}\dyad{\pnuno}{\pnuno}.
\label{213}\end{align}
Their properties reflect those of $S_\varphi\supzero$ and  $S_\psi\supzero$.
They are both self-adjoint and strictly positive. 
They also satisfy the analogous of \eqref{25}:
\begin{align}
S_\varphi\supone \pnuno=\vnuno, 
\qquad 
S_\psi\supone \vnuno=\pnuno, 
\qquad\mbox{and}\qquad 
S_\varphi\supone =\qty(S_\psi\supone )^{-1}.
\label{214}
\end{align}
Again in analogy with our previous definitions, we have the new scalar products
\begin{align}
\scp<f,g>_{\varphi(1)}:=\scp<S_\varphi\supone f,g>, 
\qquad 
\scp<f,g>_{\psi(1)}:=\scp<S_\psi\supone f,g>,
\label{215}\end{align}
and the new adjoints
\begin{align}
\scp<Xf,g>_{\varphi(1)}=\scp<f,X^{\flat_1}g>_{\varphi(1)}, 
\qquad 
\scp<Xf,g>_{\psi(1)}=\scp<f,X^{\sharp_1}g>_{\psi(1)}
\label{216}\end{align}
for all $f,g\in\Hil$.  
For all  $X\in B(\Hil)$ we find that 
\begin{align}
X^{\flat_1}=S_\psi\supone  X^\dagger S_\varphi\supone , 
\qquad 
X^{\sharp_1}=S_\varphi\supone  X^\dagger S_\psi\supone ,
\label{217}\end{align}
so that $X^{\flat_1}=\qty(S_\psi\supone )^2 X^{\sharp_1} \qty(S_\varphi\supone)^2$, and 
\begin{align}\label{237-1}
\qty(X^{\flat_1})^\dagger=\qty(X^\dagger)^{\sharp_1},
\qquad
\qty(X^{\sharp_1})^\dagger=\qty(X^\dagger)^{\flat_1}.
\end{align} 
It is clear that the analogous of Lemma \ref{lemma1} can now be stated. For instance, we can show that  $\qty(X^{\flat_1})^\dagger=\qty(X^\dagger)^{\flat_1}$ if and only if   $\qty[X,\qty(S_\varphi\supone )^2]=0$.

It is possible to find now that the operators $S_{\varphi}\supone$ and $S_\psi\supone$ are simply the cube of their zeroth counterparts $S_{\varphi}\supzero$ and $S_{\psi}\supzero$:
		\begin{align}
		S_{\varphi}\supone =\qty(S_{\varphi}\supzero)^3, 
		\qquad 
		S_{\psi}\supone =\qty(S_{\psi}\supzero)^3.
		\label{218}\end{align}
For instance, recalling the definitions of $S_{\varphi}\supzero$ and $S_{\varphi}\supone $, we have the following for a generic $f\in\Hil$: 
\begin{align}
S_{\varphi}\supone f&=\sum_{n=1}^{N}\scp<\vnuno,f>\,\vnuno=\sum_{n=1}^{N}\scp<S_{\varphi}\supzero\vno,f>\,S_{\varphi}\supzero\vno=S_{\varphi}\supzero\sum_{n=1}^{N}\scp<\vno,S_{\varphi}\supzero f>\,\vno 
\nonumber\\
&=S_{\varphi}\supzero S_{\varphi}\supzero\qty(S_{\varphi}\supzero f)=\qty(S_{\varphi}\supzero)^3 f.
\end{align}
Hence we deduce that $\sharp_1$ and $\sharp_0$ are related as well as $\flat_1$ and $\flat_0$ are. 
In fact, we have
\begin{align}
X^{\sharp_1}=\qty(S_\varphi\supzero)^2 X^{\sharp_0} \qty(S_{\psi}\supzero)^2,
\qquad 
X^{\flat_1}=\qty(S_\psi\supzero)^2 X^{\flat_0} \qty(S_{\varphi}\supzero)^2, 
\label{219}\end{align}
for all $X\in B(\Hil)$.

\subsection{Four new Hamiltonians in the second iteration}

Going back to the Hamiltonians, we first observe that starting from $H$ and $H^\dagger$, we have found four new operators $\qty(H^{\flat_0}, H^{\sharp_0})$ as well as $\qty((H^\dagger)^{\flat_0}, \qty(H^\dagger)^{\sharp_0})$, whose eigensystems are connected to those of $\qty(H,H^\dagger)$. 
It is now natural to define four more operators, again constructed out of $\qty(H,H^\dagger)$ but using the {\em new} adjoints: $\qty(H^{\flat_1},H^{\sharp_1})$ and $\qty((H^\dagger)^{\flat_1},\qty(H^\dagger)^{\sharp_1})$. 
Of course, we should check whether the new operators are different from the previous ones, which is exactly what we find in some concrete situations as we demonstrate in Section \ref{sectexe}. 

We stress that we can construct eigenvectors for the new Hamiltonians $H^{\flat_1}$, $H^{\sharp_1}$, $\qty(H^\dagger)^{\flat_1}$ and $\qty(H^\dagger)^{\sharp_1}$ easily. 
In fact, using the equality $S_\varphi\supone =\qty(S_\psi\supone )^{-1}$, we can show that
\begin{align}
H^{\sharp_1}S_\varphi\supone \pno=\overline{E}_{n}\,S_\varphi\supone \pno, 
\qquad 
H^{\flat_1}S_\psi\supone \pno=\overline{E}_{n}\,S_\psi\supone \pno,
\label{220}\end{align}
as well as 
\begin{align}
\qty(H^\dagger)^{\flat_1} S_\psi\supone \vno={E_n}\,S_\psi\supone \vno, 
\qquad 
\qty(H^\dagger)^{\sharp_1} S_\varphi\supone \vno={E_n}\,S_\varphi\supone \vno.
\label{221}\end{align}
We will give notations dedicated to these eigenstates below in \eqref{225}.
The new Hamiltonians that we have obtained in the second iteration is additionally listed in 
Table~\ref{tab1}.

Let us summarize our iteration procedure so far. 
\begin{enumerate}
\renewcommand{\theenumi}{(\roman{enumi})}

\item 
We start from a pair of non-self-adojoint Hamiltonian $(H,H^\dag)$.
Their eigenvectors form the biorthonormal basis
\begin{align}\label{246}
\E&=\qty(\{\vno\},\{\pno\}).
\end{align}

\item 
As the start of the first iteration, we constructed out of \eqref{246} the two intertwining operators $S_\varphi\supzero$ and $S_\psi\supzero$, which defined two new adjoints $\flat_0$ and $\sharp_0$.

\item 
The two new adjoints introduced four new Hamiltonians $\qty(\qty(H^{\flat_0})^\dagger, H^{\flat_0})$ and $\qty(\qty(H^{\sharp_0})^\dagger, H^{\sharp_0})$, both of which pair are isospectral to  $(H,H^\dag)$.

\item 
As the eigenvectors of the four new Hamiltonians, we introduced the two new biorthogonal bases
\begin{align}\label{247}
\E_{\sharp_0}&=\qty(\{S_\psi\supzero\vno\},\{S_\varphi\supzero\pno\}),
\\\label{248}
\E_{\flat_0}&=\qty(\{S_\varphi\supzero\vno\},\{S_\psi\supzero\pno\}).
\end{align}
We can see from \eqref{25} that $\E_{\sharp_0}$ actually coincides with  $\E$  with the role of the two bases exchanged.
It is also clear from \eqref{210} that $\E_{\flat_0}=\qty(\{\varphi_n\supone\},\{\psi_n\supone\})=\qty(\F_\varphi\supone ,\F_\psi\supone )$.

\item 
As the start of the second iteration, we constructed out of \eqref{248} the two intertwining operators $S_\varphi\supone$ and $S_\psi\supone$, which defined two adjoints $\flat_1$ and $\sharp_1$.
We could do the same out of \eqref{247}, but since it is equivalent to $\E$, it would only produce the old intertwining operators $S_\psi\supzero$ and $S_\varphi\supzero$ as well as the old adjoints $\sharp_0$ and $\flat_0$.
This will be different in the third iteration.

\item
The two new adjoints introduced four new Hamiltonians $\qty(\qty(H^{\flat_1})^\dagger, H^{\flat_1})$ and $\qty(\qty(H^{\sharp_1})^\dagger, H^{\sharp_1})$, both of which pair are isospectral to  $(H,H^\dag)$.

\item 
As the eigenvectors of the four new Hamiltonians, we introduced the two new biorthogonal bases
\begin{align}\label{249}
\E_{\sharp_1}&=\qty(\{S_\psi\supone \vno\},\{S_\varphi\supone \pno\}),
\\\label{250}
\E_{\flat_1}&=\qty(\{S_\varphi\supone \vno\},\{S_\psi\supone \pno\}).
\end{align}
\end{enumerate}
The next subsection presents the third iteration.

\subsection{Third iteration of generating four new adjoints}

%

It is now clear how to proceed in the third iteration.
The only difference from the second iteration is that we now have two new biorthogonal bases $\E_{\sharp_1}$ and $\E_{\flat_1}$ after the second iteration, whereas we had only one new basis $\E_{\flat_0}$ after the first iteration because $\E_{\sharp_0}$ was equivalent to $\E$.
Because of this difference, we will now come up with four new adjoints, which will let us introduce eight new Hamiltonians.

Obviously the number of Hamiltonians produced in this way increases at each iteration, we would have eight new adjoints in the fourth iteration, but we will stop our analysis here, since the structure is already very rich, as we will show in a moment. 
Of course, natural questions are the following: Does this procedure converge to some {\em limiting} Hamiltonian? Does this happens at least in some particular system? 
Are some basic (physical) properties of the original Hamiltonian $H$ preserved by our construction? 
Do these operators, scalar products, and adjoints have some physical meaning? Some of these questions will be considered here, while others are postponed to a future work.

To simplify the notation in the third iteration, let us introduce the following notation for the vectors in $\E_{\sharp_1}$ and  $\E_{\flat_1}$:
\begin{align}
\begin{cases}
	\varphi_n^{(2,\alpha)}=S_\psi\supone \vno, 
	\qquad 
	\psi_n^{(2,\alpha)}=S_\varphi\supone \pno, \\
	\varphi_n^{(2,\beta)}=S_\varphi\supone \vno, 
	\qquad 
	\psi_n^{(2,\beta)}=S_\psi\supone \pno, 
\end{cases}
\label{225}\end{align} 
so that $\E_{\sharp_1}=\qty(\{\varphi_n^{(2,\alpha)}\},\{\psi_n^{(2,\alpha)}\})$ and $\E_{\flat_1}=\qty(\{\varphi_n^{(2,\beta)}\},\{S_\psi\supone \pno\})$.
As we mentioned above, both bases are generally different from the other orthonormal bases $\E$, $\E_{\sharp_0}$ and $\E_{\flat_0}$.

In analogy to what we have done before for other biorthonormal bases in \eqref{24} and \eqref{213}, we now introduce the following four new intertwining operators:
\begin{align}
\begin{cases}
	S_\varphi^{(2,\alpha)}=\sum_n \dyad{\varphi_n^{(2,\alpha)}}{\varphi_n^{(2,\alpha)}}, 
	\qquad 
	S_\psi^{(2,\alpha)}=\sum_n \dyad{\psi_n^{(2,\alpha)}}{\psi_n^{(2,\alpha)}}, \\
	\\
	S_\varphi^{(2,\beta)}=\sum_n \dyad{\varphi_n^{(2,\beta)}}{\varphi_n^{(2,\beta)}}, 
	\qquad 
	S_\psi^{(2,\beta)}=\sum_n \dyad{\psi_n^{(2,\beta)}}{\psi_n^{(2,\beta)}}. \\
\end{cases}
\label{226}\end{align} 
These operators are positive, self-adjoint, and satisfy formulas similar to those in \eqref{25}:
\begin{align}
S_\varphi^{(2,\gamma)}\psi_n^{(2,\gamma)}=\varphi_n^{(2,\gamma)}, 
\qquad 
S_\psi^{(2,\gamma)}\varphi_n^{(2,\gamma)}=\psi_n^{(2,\gamma)}, 
\qquad\mbox{and}\qquad 
S_\varphi^{(2,\gamma)}=\qty(S_\psi^{(2,\gamma)})^{-1},
\label{226-1}
\end{align}
for all $n$ and for $\gamma=\alpha,\beta$. It is also possible to check that all these operators can be written in terms of the original $S_\varphi\supzero$ and $S_\psi\supzero$ as follows:
\begin{align}
S_\varphi^{(2,\alpha)}=\qty(S_\psi\supzero)^5, 
\quad
S_\psi^{(2,\alpha)}=\qty(S_\varphi\supzero)^5, 
\quad 
S_\varphi^{(2,\beta)}=\qty(S_\varphi\supzero)^7, 
\quad 
S_\psi^{(2,\beta)}=\qty(S_\psi\supzero)^7.
\label{227}\end{align}

We can now introduce four more pairs of scalar products
\begin{align}
\scp<f,g>_{\varphi(2,\gamma)}=\scp<S_\varphi^{(2,\gamma)}f,g>, 
\qquad 
\scp<f,g>_{\psi(2,\gamma)}=\scp<S_\psi^{(2,\gamma)}f,g>,
\label{228}\end{align}
and four related pairs of adjoints
\begin{align}
\scp<Xf,g>_{\varphi(2,\gamma)}=\scp<f,X^{\flat_{2,\gamma}}g>_{\varphi(2,\gamma)}, 
\qquad 
\scp<Xf,g>_{\psi(2,\gamma)}=\scp<f,X^{\sharp_{2,\gamma}}g>_{\psi(2,\gamma)}
\label{229}\end{align}
for all $f,g\in\Hil$, $\gamma=\alpha,\beta$.  For all  $X\in B(\Hil)$, we find that 
\begin{align}
X^{\flat_{2,\gamma}}=S_\psi^{(2,\gamma)} X^\dagger S_\varphi^{(2,\gamma)}, 
\qquad 
X^{\sharp_{2,\gamma}}=S_\varphi^{(2,\gamma)} X^\dagger S_\psi^{(2,\gamma)}.
\label{230}\end{align}
We can deduce more results in complete analogy with our previous ones. 

The main point is that these new adjoints still produce Hamiltonians with the same eigenvalues as the original $H$ and $H^\dagger$, and with easily computable eigenstates. We obtain
\begin{align}
\begin{cases}
	H^{\flat_{2,\gamma}}S_\psi^{(2,\gamma)}\pno=\overline{E}_{n}S_\psi^{(2,\gamma)}\pno, \\
	H^{\sharp_{2,\gamma}}S_\varphi^{(2,\gamma)}\pno=\overline{E}_{n}S_\varphi^{(2,\gamma)}\pno, \\	
\end{cases}
\label{231}\end{align} 
and
\begin{align}
\begin{cases}
	\qty(H^\dagger)^{\flat_{2,\gamma}}S_\psi^{(2,\gamma)}\vno={E_n}S_\psi^{(2,\gamma)}\vno, \\
	\qty(H^\dagger)^{\sharp_{2,\gamma}}S_\varphi^{(2,\gamma)}\vno={E_n}S_\varphi^{(2,\gamma)}\vno, \\	
\end{cases}
\label{232}\end{align} 
for $\gamma=\alpha,\beta$, which have been already added to 
Table~\ref{tab1}.
We thus see that the chain of solvable Hamiltonians becomes longer, and can be made even longer simply using the new biorthonormal bases which are constructed step by step.

Let us summarize what we have done so far. 
We have used the eigenvectors $\vno$ of $H$ (different from $H^\dagger$, in principle) to construct the operator $S_\varphi\supzero$ and its inverse to introduce several new scalar products and their related adjoints, which generate new Hamiltonians all with the same eigenvalues as $H$ and $H^\dagger$. 
All these adjoints are matrices which are similar to $H$ or to $H^\dagger$; the similarity maps are given by odd powers of $S_\varphi\supzero$ or $S_\psi\supzero=(S_\varphi\supzero)^{-1}$. 
We can also use even powers, but this would  not be on the line of this paper, where the new Hamiltonians are just suitable adjoints of $H$ or of $H^\dagger$ constructed in a specific way. 
In fact, every time we have a positive invertible operator $T$ on $\Hil$, we can introduce two new scalar products $\scp<f,g>_T:=\scp<Tf,g>$ and $\scp<f,g>_{T^{-1}}:=\scp<T^{-1}f,g>$ as well as two new adjoints, as we have done in, \textit{e.g.}, formula \eqref{27}. 
For this reason it is important to stress here that the whole procedure is {\em generated} by $H$ alone, and there is no need to introduce any {\em external} operator like the $T$ above.

\section{Examples}\label{sectexe}

In this section we will first consider a simple example in two dimensions in order to demonstrate that the procedure proposed here is not trivial: 
the chain of the Hamiltonians that we deduce and their eigenvectors are all different in general. Then we consider much richer and more physically motivated examples.

\subsection{A $2\times 2$ upper triangular matrix}

Let us consider the following two-by-two matrix
\begin{align}
H=
\begin{pmatrix}
{E_1} & \alpha(E_2-E_1) \\
0 & {E_2} \\
\end{pmatrix},
\end{align}
where $\alpha\in\mathbb{R}$, while $E_1$ and $E_2$ are generally complex quantities. 
Its $\dagger$-adjoint in $\Hil=\mathbb{C}^2$ is simply
\begin{align}
H^\dagger=
\begin{pmatrix}
\overline{E}_{1} & 0 \\
\alpha(\overline{E}_{2}-\overline{E}_{1}) & \overline{E}_{2} \\
\end{pmatrix}.
\end{align}
The eigenvalues of $H$ are exactly $E_1$ and $E_2$, while $\overline{E}_{1}$ and $\overline{E}_{2}$ are those of $H^\dagger$.
The eigenvectors of $H$ and $H^\dagger$ can be easily found:
\begin{align}
\varphi_1\supzero&=
\begin{pmatrix}
1 \\
0 \\
\end{pmatrix},
&
\varphi_2\supzero&=
\begin{pmatrix}
\alpha \\
1 \\
\end{pmatrix},
&
\psi_1\supzero&=
\begin{pmatrix}
1 \\
-\alpha \\
\end{pmatrix},
&
\psi_2\supzero&=
\begin{pmatrix}
0 \\
1 \\
\end{pmatrix}.
\end{align}
We can confirm the biorthonormality $\scp<\varphi\supzero_k,\psi\supzero_l>=\delta_{k,l}$ and the eigenvalue equations $H\varphi\supzero_k=E_k\supzero\varphi_k$ and $H^\dagger
\psi\supzero_k=\overline{E}_{k}\psi\supzero_k$. 
It is clear that an interesting situation arises when $\alpha\neq0$. 
In fact, when $\alpha=0$, $\F_\varphi\supzero$ and $\F_\psi\supzero$ collapse to the canonical orthonormal basis in $\Bbb C^2$, namely $\F_e=\{e_1,e_2\}$ and both $H$ and $H^\dagger$ reduce to diagonal matrices, which are not necessarily equal, except when $E_1, E_2\in\mathbb{R}$.

The operators $S_\varphi\supzero$ and $S_\psi\supzero$ in \eqref{24} are given by
\begin{align}
S_\varphi\supzero&=
\begin{pmatrix}
1+\alpha^2 & \alpha \\
\alpha & 1 \\
\end{pmatrix},
&
S_\psi\supzero&=
\begin{pmatrix}
1 & -\alpha \\
-\alpha & 1+\alpha^2 \\
\end{pmatrix},
\end{align}
which are clearly self-adjoint because $\alpha\in\mathbb{R}$. 
We can also confirm that they are inversion matrices to each other: $S\supzero_\varphi S\supzero_\psi=S\supzero_\psi S\supzero_\varphi=\1$. 
Moreover, $S\supzero_\psi\varphi\supzero_n=\psi\supzero_n$ and $S\supzero_\varphi\psi\supzero_n=\varphi\supzero_n$,
which confirms \eqref{25}. 

We can compute the adjoints defined in \eqref{27} using \eqref{28}:
\begin{align}
H^{\sharp_0}&=S_\varphi\supzero H^\dagger S\supzero_\psi=
\begin{pmatrix}
\overline{E}_{1} & \alpha(\overline{E}_{2}-\overline{E}_{1}) \\
0 & \overline{E}_{2} \\
\end{pmatrix}
\\
H^{\flat_0}&=S\supzero_\psi H^\dagger S\supzero_\varphi=
\begin{pmatrix}
(1+\alpha^2)\overline{E}_{1}-\alpha^2(2+\alpha^2)\overline{E}_{2} & \alpha(1+\alpha^2)(\overline{E}_{1}-\overline{E}_{2}) \\
-\alpha(1+\alpha^2)(2+\alpha^2)(\overline{E}_{1}-\overline{E}_{2}) & -\alpha^2(2+\alpha^2)\overline{E}_{1}+(1+\alpha^2)^2\overline{E}_{2} \\
\end{pmatrix}.
\end{align}
It is clear that $H=H^{\sharp_0}$ if and only if $E_1$ and $E_2$ are real. 
This is in agreement with Theorem \ref{thm1} for the pseudo-Hermiticity.
On the other hand, it is not so evident at a first glance that $H^\dagger=(H^\dagger)^{\flat_0}$ if and only if $E_1$ and $E_2$ are real. 
Nonetheless, we can show it from the second identity in \eqref{29}, since $H^\dagger-(H^\dagger)^{\flat_0}=H^\dagger-(H^{\sharp_0})^\dagger=(H-H^{\sharp_0})^\dagger=0$ if and only if $E_1$ and $E_2$ are real.

Furthermore, if we introduce the new vectors
\begin{align}
\psi_1\supone &=S\supzero_\psi\psi\supzero_1=
\begin{pmatrix}
1+\alpha^2 \\
-\alpha(2+\alpha^2) \\
\end{pmatrix},
&
\psi\supone _2&=S\supzero_\psi\psi\supzero_2=
\begin{pmatrix}
-\alpha \\
1+\alpha^2 \\
\end{pmatrix},
\\
\varphi\supone _1&=S\supzero_\varphi\varphi\supzero_1=
\begin{pmatrix}
1+\alpha^2 \\
\alpha \\
\end{pmatrix},
&
\varphi\supone _2&=S\supzero_\psi\varphi\supzero_2=
\begin{pmatrix}
\alpha(2+\alpha^2) \\
1+\alpha^2 \\
\end{pmatrix}
\end{align}
as in \eqref{210}, we see their biorthogonality $\scp<\psi\supone _k,\varphi\supone _l>=\delta_{k,l}$. 
For $\alpha=0$, the sets $\F_{\psi}\supone $ and $\F_{\varphi}\supone $ both collapse again to $\F_e$.
Using now $\F_{\psi}\supone $ and $F_{\varphi}\supone $, we introduce the  operators
\begin{align}
S_{\psi}\supone &=
\begin{pmatrix}
\alpha^4+3 \alpha^2+1 & -\alpha \alpha^4+4 \alpha^2+3 \\
-\alpha \alpha^4+4 \alpha^2+3 & \alpha^6+5 \alpha^4+6 \alpha^2+1 \\
\end{pmatrix},
\\
S_{\varphi}\supone &=
\begin{pmatrix}
\alpha^6+5 \alpha^4+6 \alpha^2+1 & \alpha \alpha^4+4 \alpha^2+3 \\
\alpha \alpha^4+4 \alpha^2+3 & \alpha^4+3 \alpha^2+1 \\
\end{pmatrix}
\end{align}
as in \eqref{213}.
Again, they are both self-adjoint (with respect to $\dagger$), positive, and inverse to each other, 
$S_{\varphi}\supone =\qty(S_{\psi}\supone )^{-1}$.
It is also easy to confirm that $S\supone _{\varphi}=\qty(S\supzero_\varphi)^3$ and $S\supone _{\psi}=\qty(S\supzero_\psi)^3$ in \eqref{218}. 

After some algebra, we find the adjoints defined in \eqref{216} of the Hamiltonian,$H^{\sharp_1}$ and $H^{\flat_1}$, in the forms
\begin{footnotesize}
\begin{align}
	H^{\sharp_1}=
	\begin{pmatrix}
	\alpha^4+3 \alpha^2+1^2 \overline{E}_{1}-\alpha^2 \alpha^2+1 \alpha^2+2 \alpha^2+3 \overline{E}_{2} & -\alpha \alpha^4+3 \alpha^2+1 \alpha^4+4 \alpha^2+3 \overline{E}_{1}-\overline{E}_{2} \\
	\alpha \alpha^6+5 \alpha^4+7 \alpha^2+2 \overline{E}_{1}-\overline{E}_{2} & \alpha^2+1 \alpha^2+2 \alpha^2+3 \overline{E}_{2}-\overline{E}_{1} \alpha^2+\overline{E}_{2} \\
	\end{pmatrix}
\end{align}
\end{footnotesize}
and
\begin{align}
H^{\flat_1}=
\begin{pmatrix}
h_{1,1} & h_{1,2} \\
h_{2,1} & h_{2,2} \\
\end{pmatrix},
\end{align}
where
\begin{align}
h_{1,1}&= \alpha^6+5 \alpha^4+6 \alpha^2+1^2 \overline{E}_{1}-\alpha^2 \alpha^2+1 \alpha^2+2 \alpha^2+3 \alpha^4+4 \alpha^2+2 \overline{E}_{2},
\\
h_{1,2}&=\alpha \alpha^4+4 \alpha^2+3 \alpha^6+5 \alpha^4+6 \alpha^2+1 \overline{E}_{1}-\overline{E}_{2},
\\
h_{2,1}&= -\alpha \alpha^2+2 \alpha^4+4 \alpha^2+2 \alpha^6+5 \alpha^4+6 \alpha^2+1 \overline{E}_{1}-\overline{E}_{2},
\\
h_{2,2}&=\alpha^2+1 \alpha^2+2 \alpha^2+3 \alpha^4+4 \alpha^2+2 \overline{E}_{2}-\overline{E}_{1} \alpha^2+\overline{E}_{2}.
\end{align}
Notice that for generic $\alpha$, the two new Hamiltonians $H^{\sharp_1}$ and $H^{\flat_1}$ are  different from all the Hamiltonians considered before. 
The eigenvectors in \eqref{220} and \eqref{221} are specifically given as follows:
\begin{align}
S_\varphi\supone \psi_1\supzero&=
\begin{pmatrix}
\alpha^4+3 \alpha^2+1 \\
\alpha \alpha^2+2 \\
\end{pmatrix},
&
S_\varphi\supone \psi_2\supzero&=
\begin{pmatrix}
\alpha \alpha^4+4 \alpha^2+3 \\
\alpha^4+3 \alpha^2+1 \\
\end{pmatrix},
\\
S_\psi\supone \varphi_1\supzero&=
\begin{pmatrix}
\alpha^4+3 \alpha^2+1 \\
-\alpha \alpha^4+4 \alpha^2+3 \\
\end{pmatrix},
&
S_\psi\supone \varphi_2\supzero&=
\begin{pmatrix}
-\alpha \alpha^2+2 \\
\alpha^4+3 \alpha^2+1 \\
\end{pmatrix},
\\
S_\psi\supone \psi_1\supzero&=
\begin{pmatrix}
\alpha^6+5 \alpha^4+6 \alpha^2+1 \\
-\alpha \alpha^6+6 \alpha^4+10 \alpha^2+4 \\
\end{pmatrix},
&
S_\psi\supone \psi_2\supzero&=
\begin{pmatrix}
-\alpha \alpha^4+4 \alpha^2+3 \\
\alpha^6+5 \alpha^4+6 \alpha^2+1 \\
\end{pmatrix}
\\
S_\varphi\supone \varphi_1\supzero&=
\begin{pmatrix}
\alpha^6+5 \alpha^4+6 \alpha^2+1 \\
\alpha \alpha^4+4 \alpha^2+3 \\
\end{pmatrix}
&
S_\varphi\supone \varphi_2\supzero&=
\begin{pmatrix}
\alpha \alpha^6+6 \alpha^4+10 \alpha^2+4 \\
\alpha^6+5 \alpha^4+6 \alpha^2+1 \\
\end{pmatrix},
\end{align}
which are all different one from another for general $\alpha$ and are  also different from all the vectors introduced before. 
A direct computation shows, for instance, the biorthogonality
\begin{align}
\scp<S_\varphi\supone \psi_k\supzero,S_\psi\supone \varphi_l\supzero> = \scp<S_\psi\supone \psi_k\supzero,S_\varphi\supone \varphi_l\supzero> = \delta_{k,l},
\end{align}
as we deduced in Section \ref{sect2}.

In this way we construct several sets of biorthogonal bases of $\Hil$, each one consisting of eigenstates of some matrix related to the
original Hamiltonian $H$. 
We notice that the procedure is essentially independent of the fact that  $E_1$ and $E_2$ are real or complex.
The main difference lies in the fact that when they are all real, the number of isospectral operators doubles.

Incidentally, we observe that (almost) all the matrices and vectors which are found during the computations are different from those we find at earlier steps. This is important to clarify the non triviality of our procedure.


\subsection{An example of off-diagonal non-self-adjointness: the clean and dirty Hatano-Nelson model}
\label{subsec3.2}

In the present subsection, we consider various versions of the Hatano-Nelson (HN) model~\cite{Hatano96,Hatano97}, in which the non-self-adjointness comes from the asymmetry in real off-diagonal elements.
The HN model was originally considered for an effective model of a statistical physical system of magnetic flux lines in a type-II superconductor.
We first consider two versions under open and periodic boundary conditions without any diagonal elements.
We then add diagonal elements to periodic systems.
We will see that addition of few matrix elements drastically changes the development of new isospectral Hamiltonians.


\subsubsection{An open clean system with asymmetric hopping}

Though the original HN model contained impurity potentials, we first consider the clean version defined on a chain of $L$ lattice sites, labeled by $x$: 
\begin{align}\label{322}
H_{HN}&:=-\hop \sum_{x=1}^{L-1}
\qty(e^g \dyad{x+1}{x}+e^{-g}\dyad{x}{x+1})
\\\label{322-1}
&=-\hop\begin{pmatrix}
0 &  e^{-g} & & & & \\
 e^{g} & 0 &  e^{-g} & & & \\
&  e^{g} & 0 & \ddots & & \\
& & \ddots & \ddots &   e^{-g} & \\
& & &   e^{g} & 0 &   e^{-g} \\
& & & &   e^{g} & 0 \\
\end{pmatrix},
\end{align}
where $x$ refers to a lattice site, $\hop ,g\in\mathbb{R}$ and we have put $\hbar=1$ to simplify the notation. 
$H_{HN}$ would be a self-adjoint operator if $g=i\theta$ were a pure imaginary number;
then $\theta\in\mathbb{R}$ would be referred to as the Peierls phase; see \eqref{327} below.

A similar matrix was recently considered in Ref.~\cite{santos} in order to describe an asymmetric tunneling between two energy-degenerate sites coupled by a non-reciprocal interaction without dissipation. 
Later in Ref.~\cite{bagaop2015a}, the same model was used in connection with the right definition of transition probabilities for systems driven by non self-adjoint Hamiltonian. 
In its simplified version in their notation, the non-self-adjoint Hamiltonian was given by
\begin{align}\label{323}
H:=- g' 
\begin{pmatrix}
0 & 1-k \\
1+k & 0 \\
\end{pmatrix},
\end{align}
where $g'\in\Bbb R$ and $k\in]-1,1[$. 
The $2\times2$-version of the Hamiltonian~\eqref{322} and the one~\eqref{323} are equivalent to each other under the parameter correspondence
\begin{align}
{\hop}^2={g'}^2(1-k^2),
\qquad
e^{2g}=\frac{1+k}{1-k}.
\end{align}

The Hamiltonian~\eqref{322} can be transformed into a real symmetric matrix as follows; 
\begin{align}\label{325}
S^{-1}H_{HN}S=\eval{H_{HN}}_{g=0},
\qquad\mbox{where}\qquad
S=\begin{pmatrix}
1 &  & & & &\\
& e^{g} & & & &\\
& & e^{2g} & & & \\
& & & \ddots & & \\
& & & & e^{(L-2)g} & \\
& & & & & e^{(L-1)g} \\
\end{pmatrix}.
\end{align}
This is physically understood in the following way~\cite{Hatano96,Hatano97}.
A vector potential due to a magnetic field is often represented in the discretized space by the Peierls substitution in the form 
\begin{align}\label{327}
H_\textrm{Peierls}= -\hop
\begin{pmatrix}
0 & e^{i\theta} & & & & \\
 e^{-i\theta} & 0 &  e^{i\theta} & & & \\
&   e^{-i\theta} & 0 & \ddots & & \\
& & \ddots & \ddots &   e^{i\theta} & \\
& & &   e^{-i\theta} & 0 &   e^{i\theta} \\
& & & &   e^{-i\theta} & 0 \\
\end{pmatrix},
,
\end{align}
where $\theta=eAa$ with $e$ the electron charge, $A$ the vector potential and $a$ the lattice constant. 
If there is no loop in the lattice, however, we should be able to gauge it out with the gauge transformation $e^{-i\theta x}$, where $x$ 
is the lattice label $x=1,2, \cdots$.
The symmetrization in \eqref{325} is nothing but gauging out the imaginary gauge field $\theta=ig$ with the gauge transformation $e^{gx}$.
This is a physical background of the symmetrization \eqref{325}.

Let us put $\hop=1$ for brevity hereafter.
After the symmetrization in \eqref{325}, we find the real eigenvalues of the Hamiltonian~\eqref{322-1} as follows:
\begin{align}
E_n:=-2\cos k_n,
\qquad
\mbox{where}
\quad
k_n=\frac{n\pi}{L+1}
\end{align}
for $n=1,2,\cdots,L$. 
Figure~\ref{fig1}(a) exemplifies the spectrum for $L=11$, $\hop=1$, $g=0.1$; note, however, that the spectrum actually does not depend on $g$. {As a consequence, the eigenvalues for  $H_{HN}$ and ${H_{HN}}^\dag$ coincide. In fact, since the eigenvalues of $H_{HN}$ do not depend on $g$, these eigenvalues coincide with those we get for (\ref{322}) taking $g=0$. But this matrix is self-adjoint: $H_{HN}|_{g=0}=H_{HN}^\dagger|_{g=0}$. Hence its eigenvalues must all be real.}

\begin{figure}
\begin{minipage}[t]{0.45\textwidth}
\vspace{0mm}
\includegraphics[width=\textwidth]{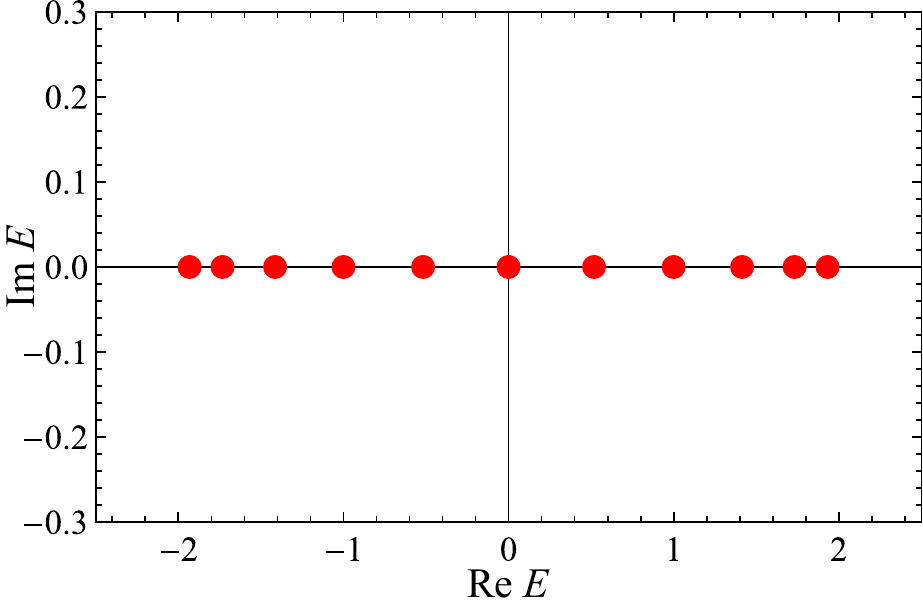}
\centering (a)
\end{minipage}
\hfill
\begin{minipage}[t]{0.45\textwidth}
\vspace{0mm}
\includegraphics[width=\textwidth]{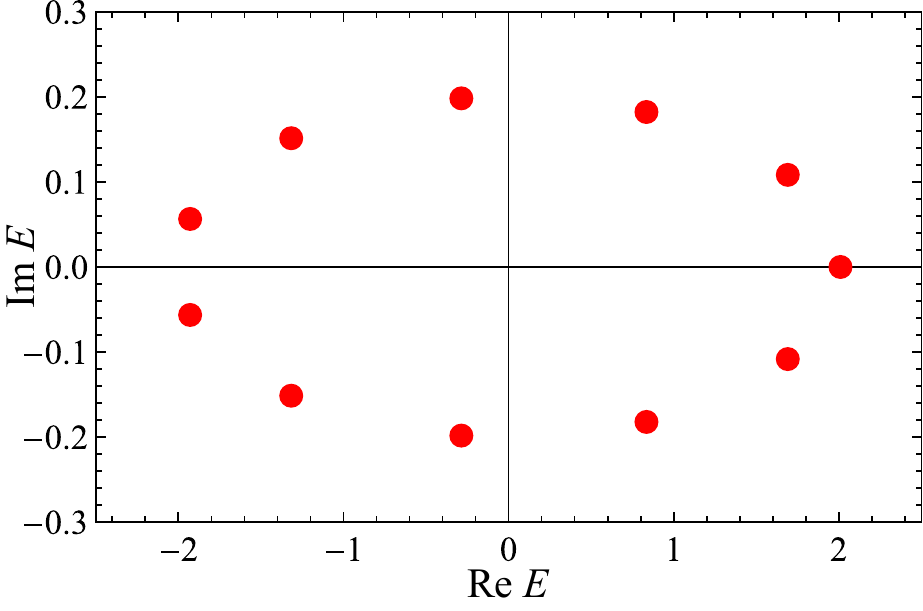}
\centering (b)
\end{minipage}
\\[\baselineskip]
\begin{minipage}[t]{0.45\textwidth}
\vspace{0mm}
\includegraphics[width=\textwidth]{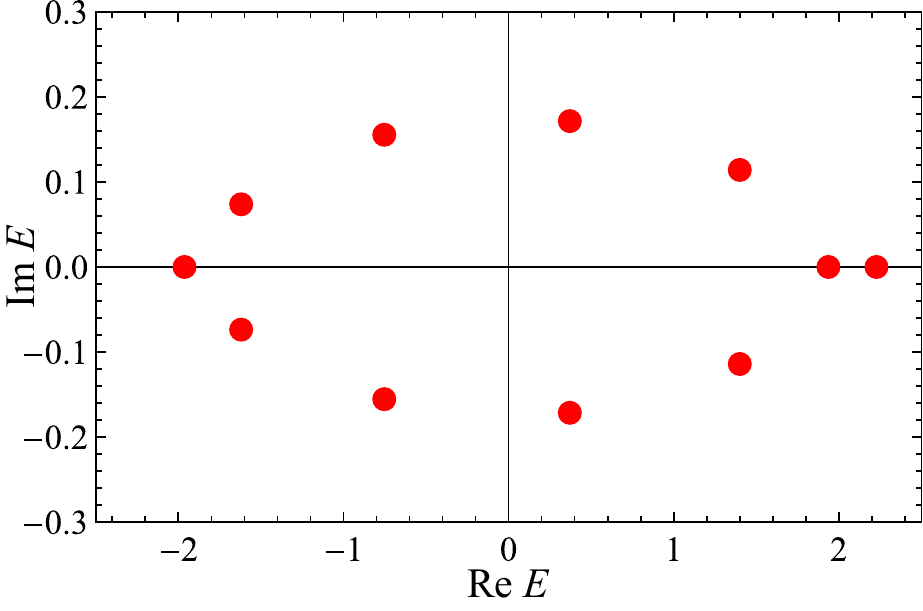}
\centering (c)
\end{minipage}
\hfill
\begin{minipage}[t]{0.45\textwidth}
\vspace{0mm}
\includegraphics[width=\textwidth]{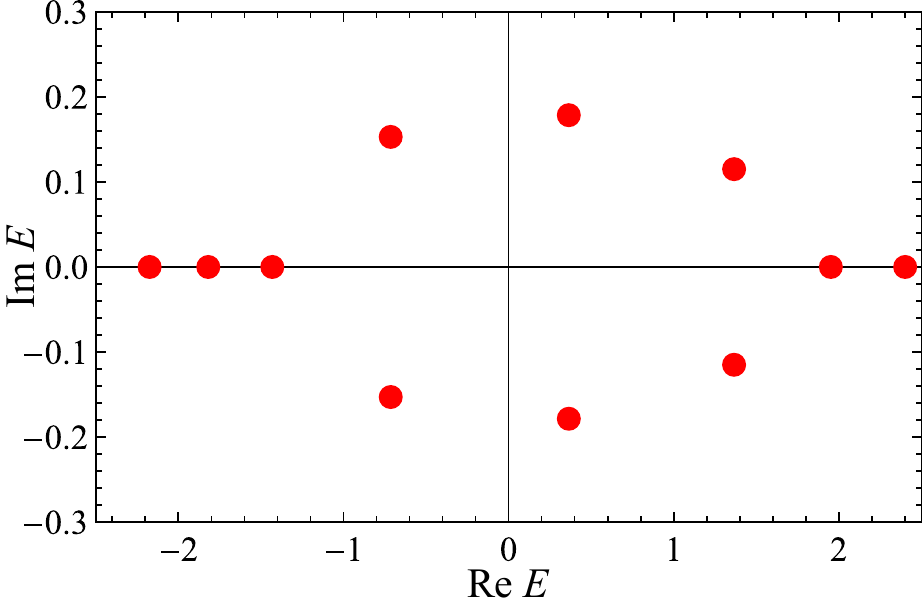}
\centering (d)
\end{minipage}
\\[\baselineskip]
\centering
\includegraphics[width=0.45\textwidth]{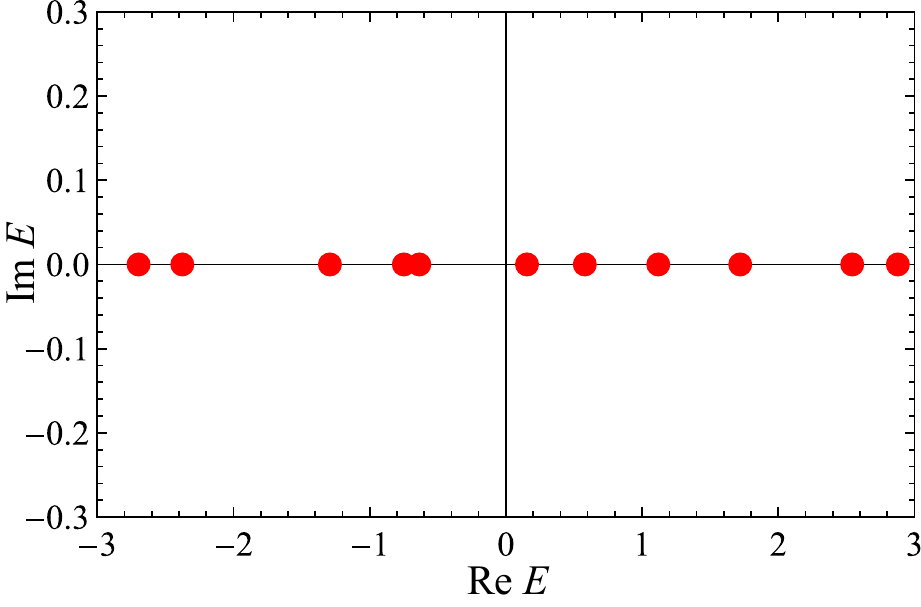}
\\
\centering(e)
\caption{The eigenvalue distribution plotted on a complex plane for the models throughout Subsect.~\ref{subsec3.2}, namely, (a)~\eqref{322}, (b)~\eqref{343}, (c)~\eqref{357}, (d)~\eqref{357-1} in a moderately dirty case, and (e)~\eqref{357-1} in a quite dirty case, all with  $L=11$, $\hop=1$, $g=0.1$:
in (c), we set only one impurity at $x=6$ with $v_6=1$;
in (d), we chose the potential of each site randomly from the box distribution $[-1,1]$, while in (e) from $[-2.5,2.5]$.}
\label{fig1}
\end{figure}

The corresponding eigenvectors of $H_{HN}$ and ${H_{HN}}^\dag$ are respectively given by
\begin{align}
\ket{\varphi_n\supzero}&:=\sqrt{\frac{2}{L+1}}\sum_{x=1}^L e^{g(x-1)}\sin\qty(k_n x)\ket{x},
\\
\ket{\psi_n\supzero}&:=\sqrt{\frac{2}{L+1}}\sum_{x=1}^L e^{-g(x-1)}\sin\qty(k_n x)\ket{x},
\end{align}
where $x=1,2,\cdots,L$ are lattice points. 
%
We can physically understand these eigenvectors in terms of the imaginary gauge transformation too~\cite{Hatano96,Hatano97} and mathematically in terms of the similarity transformation~\eqref{325}.
The eigenvectors satisfy the biorthonormality $\scp<\varphi_m^{(0)},\psi_n^{(0)}>=\delta_{m,n}$ for $m,n=1,2,\cdots,L$,
where the scalar product $\scp<.,.>$ here is the ordinary product in $\Hil=\mathbb{C}^L$.
A simple algebra also shows the resolution of identity $\sum_{n=1}^L\dyad{\varphi_n}{\psi_n}=\sum_{n=1}^L\dyad{\psi_n}{\varphi_n}=\1$. 

The intertwining operators defined in \eqref{24} are given by the following simple diagonal matrices
\begin{align}\label{329}
S_\varphi\supzero
=
\begin{pmatrix}
1 & & & & \\
 & e^{2g} & & &\\
& & e^{4g} & & \\
& & & \ddots & \\
& & & & e^{2(L-1)g}
\end{pmatrix},
\qquad 
S_\psi\supzero
=
\begin{pmatrix}
1 & & & & \\
 & e^{-2g} & & &\\
& & e^{-4g} & & \\
& & & \ddots & \\
& & & & e^{-2(L-1)g}
\end{pmatrix}.
\end{align}
We observe that the former is $S^2$ and the latter is $S^{-2}$, where $S$ is the symmetrizing similarity transformation in \eqref{325}.

As expected, we indeed find $S_\varphi=S_\psi^{-1}=S_\varphi^\dagger$ and see that they are positive operators.
The adjoints defined in \eqref{27} and \eqref{28} of the Hamiltonian~\eqref{322} as well as its $\dagger$-adjoint are given by
\begin{align}
{H_{HN}}^{\flat_0}&=
\begin{pmatrix}
0 & -e^{3g} & & & \\
-e^{-3g} & 0 & -e^{3g} & & \\
& -e^{-3g} & 0 & \ddots & \\
& & \ddots & \ddots & - e^{3g} \\
& & & - e^{-3g} & 0  \\
\end{pmatrix},
\\
{H_{HN}}^{\sharp_0}&=
\begin{pmatrix}
0 & -e^{-g} & & & \\
-e^{g} & 0 & -e^{-g} & & \\
& -e^{g} & 0 & \ddots & \\
& & \ddots & \ddots & - e^{-g} \\
& & & - e^{g} & 0  \\
\end{pmatrix},
\\
\qty({H_{HN}}^\dagger)^{\flat_0}&=
\begin{pmatrix}
0 & -e^{g} & & & \\
-e^{-g} & 0 & -e^{g} & & \\
& -e^{-g} & 0 & \ddots & \\
& & \ddots & \ddots & - e^{g} \\
& & & - e^{-g} & 0  \\
\end{pmatrix},
\\
\qty({H_{HN}}^\dagger)^{\sharp_0}&=
\begin{pmatrix}
0 & -e^{-3g} & & & \\
-e^{3g} & 0 & -e^{-3g} & & \\
& -e^{3g} & 0 & \ddots & \\
& & \ddots & \ddots & - e^{-3g} \\
& & & - e^{3g} & 0  \\
\end{pmatrix},
\end{align}
They satisfy the identities \eqref{29} as well as the pseudo-Hermiticity \eqref{223-1}.
Comparing these new Hamiltonians with the symmetrization procedure in Eq.~\eqref{325}, we observe that the generation of the new Hamiltonians can be referred to as the asymmetrization;
using powers of the symmetrizing similarity transformation $S$ in \eqref{325}, we make the Hamiltonian $H_{HN}$ more asymmetric.
We do not see any self-adjointness beyond the pseudo-Hermiticity $H^{\sharp_0}=H$ because other subsequent adjoints are too asymmetric to be equal to $H$; see below.
The Hamiltonian will become even more asymmetric in the next iteration.

The next round starts with \eqref{210}, which gives the new vectors 
\begin{align}
\ket{\varphi_n\supone}:=\sqrt{\frac{2}{L+1}}e^{3g(x-1)}\sin\qty(k_n x)\ket{x},
\quad
\ket{\psi_n\supone}:=\sqrt{\frac{2}{L+1}}e^{-3g(x-1)}\sin\qty(k_n x)\ket{x},
\end{align}
where $x=1,2,\cdots,L$.
They are again biorthogonal and satisfy the resolution of identity.
The new intertwining operators \eqref{213} are then given by
\begin{align}
S_\varphi\supone
=
\begin{pmatrix}
1 & & & & \\
 & e^{6g} & & &\\
& & e^{12g} & & \\
& & & \ddots & \\
& & & & e^{6(L-1)g}
\end{pmatrix},
\qquad 
S_\psi\supone
=
\begin{pmatrix}
1 & & & & \\
 & e^{-6g} & & &\\
& & e^{-12g} & & \\
& & & \ddots & \\
& & & & e^{-6(L-1)g}
\end{pmatrix},
\end{align}
which indeed satisfy \eqref{218} with the intertwining operators \eqref{329} in the first generation.
The new adjoints are given by
\begin{align}
{H_{HN}}^{\flat_1}&=
\begin{pmatrix}
0 & -e^{7g} & & & \\
-e^{-7g} & 0 & -e^{7g} & & \\
& -e^{-7g} & 0 & \ddots & \\
& & \ddots & \ddots & - e^{7g} \\
& & & - e^{-7g} & 0  \\
\end{pmatrix},
\\
{H_{HN}}^{\sharp_1}&=
\begin{pmatrix}
0 & -e^{-5g} & & & \\
-e^{5g} & 0 & -e^{-5g} & & \\
& -e^{5g} & 0 & \ddots & \\
& & \ddots & \ddots & - e^{-5g} \\
& & & - e^{5g} & 0  \\
\end{pmatrix},
\\
\qty({H_{HN}}^\dagger)^{\flat_1}&=
\begin{pmatrix}
0 & -e^{5g} & & & \\
-e^{-5g} & 0 & -e^{5g} & & \\
& -e^{-5g} & 0 & \ddots & \\
& & \ddots & \ddots & - e^{5g} \\
& & & - e^{-5g} & 0  \\
\end{pmatrix},
\\
\qty({H_{HN}}^\dagger)^{\sharp_1}&=
\begin{pmatrix}
0 & -e^{-7g} & & & \\
-e^{7g} & 0 & -e^{-7g} & & \\
& -e^{7g} & 0 & \ddots & \\
& & \ddots & \ddots & - e^{-7g} \\
& & & - e^{7g} & 0  \\
\end{pmatrix},
\end{align}
which satisfy \eqref{237-1}.

We can go further but we stop here.
A point to note is that while the Hamiltonian becomes more and more asymmetric, the tridiagonal structure of the original Hamiltonian \eqref{322} is kept in all its adjoints.
Because the hopping elements become more asymmetric as we proceed in the iteration, the eigenvector $\phi_n^{(\nu)}$ become more biased to the right end as $e^{mg(x-1)}$ with a larger integer $m$, whereas the eigenvector $\psi_n^{(\nu)}$ become more biased to the left end as $e^{-mg(x-1)}$.


\subsubsection{A periodic clean system with asymmetric hopping}
\label{sec3.2.2}
The situation discussed so far changes drastically when we require the periodic boundary condition to the model \eqref{322}:
\begin{align}\label{343}
H_{HN}&:=-\hop \sum_{x=1}^{L}
\qty(e^g \dyad{x+1}{x}+e^{-g}\dyad{x}{x+1})
\end{align}
with the identification $\ket{x+L}=\ket{x}$.
(The Hamiltonian is now different from \eqref{322} but we use the same notation $H_{HN}$ for brevity.
This notation applies only throughout the present subsection.)
This adds only two off-diagonal elements to the Hamiltonian on the top-right and bottom-left corners:
\begin{align}\label{345}
H_{HN}&=\begin{pmatrix}
0 & -\hop e^{-g} & & & & -\hop e^{g} \\
-\hop e^{g} & 0 & -\hop e^{-g} & & & \\
& -\hop  e^{g} & 0 & \ddots & & \\
& & \ddots & \ddots & -\hop  e^{-g} & \\
& & & -\hop  e^{g} & 0 & -\hop  e^{-g} \\
-\hop e^{-g} & & & & -\hop  e^{g} & 0 \\
\end{pmatrix}.
\end{align}
The addition of these elements produces translational invariance in the Hamiltonian;
that is, the Hamiltonian is invariant under the shift operation $x\to x+1$ thanks to the periodic boundary condition $\ket{x+L}=\ket{x}$.
Since the system is now a loop, we cannot gauge out the vector potential $\theta=ig$.
The existence of the imaginary vector potential implies that an imaginary magnetic field threads through the ring.

Because of the translational symmetry, we can diagonalize the matrix and its $\dag$-adjoint by means of the Fourier transformation,
in other words, by the following $g$-independent eigenvectors, which are actually equal to each other:
\begin{align}\label{345-1}
\ket{\varphi_n\supzero}=\ket{\psi_n\supzero}&=\frac{1}{\sqrt{L}}\sum_{x=1}^L e^{ik_nx}\ket{x},
\qquad\mbox{where}\quad
k_n=\frac{2n\pi}{L}
\end{align}
for $n=1,2,3,\cdots,L$, with the respective eigenvalues
\begin{align}\label{346}
E_n&=-2\cos(k_n+ig)=-2\cos k_n\cosh g-2i\sin k_n\sinh g,
\\
\overline{E}_{n}&=-2\cos(k_n-ig)=-2\cos k_n\cosh g+2i\sin k_n\sinh g,
\end{align}
 where we set $\hop=1$ for brevity. We notice that they are now complex. {We notice also that now $E_n$ depends explicitly on $g$, and that each $E_n$ turns out to be real when $g=0$. This is in agreement with the fact that, when $g=0$, the Hamiltonian in (\ref{345}) is Hermitian.}
The fact that the eigenvectors are plane waves indicates that the wave function circles around the periodic system.
Indeed, because of the imaginary magnetic field threading through the loop of the periodic system, the imaginary current flows around the system~\cite{Hatano96,Hatano97}.

The eigenvalue spectrum~\eqref{346} is exemplified in Fig.~\ref{fig1}(b) for $L=11$, $\hop=1$, $g=0.1$.
All eigenvalues generically sit on the ellipse 
\begin{align}\label{346-1}
\qty(\frac{\Re E_n}{\cosh g})^2+\qty(\frac{\Im E_n}{\sinh g})^2=4.
\end{align}

Because $\varphi_n\supzero=\psi_n\supzero$, both the intertwining operators \eqref{24} reduce to
\begin{align}\label{347}
S_\varphi\supzero=S_\psi\supzero=\sum_{n=1}^L\dyad{\varphi_n\supzero}{\psi_n\supzero}=\1.
\end{align}
Therefore, no new adjoints are generated.
This is due to the fact that the two basis sets in \eqref{345-1} collapse into a single orthonormal basis despite the fact that the Hamiltonian is not Hermitian. 
This collapse happens because of the translational invariance produced by the two elements on the corners in \eqref{345}.
Indeed, the Hamiltonian~\eqref{345} commutes with the shift operator as $H_{NH}S=SH_{NH}$, where
\begin{align}
S:=\sum_{x=1}^L\dyad{x+1}{x}
=\begin{pmatrix}
0 &  & & & & 1 \\
1 & 0 &  & & & \\
& 1 & 0 &  & & \\
& & \ddots & \ddots &  & \\
& & & 1 & 0 &  \\
 & & & & 1 & 0 \\
\end{pmatrix}.
\end{align}
Such Hamiltonians are generally diagonalized by the Fourier transform operator
\begin{align}
F&=\sum_{x=1}^{L}\sum_{y=1}^L e^{ik (x+y-2)}\dyad{x}{y}
\\
&=\frac{1}{\sqrt{L}}
\begin{pmatrix}
1 & 1 & 1 & \cdots & 1 \\
1 & e^{ik} & e^{2ik} & \cdots & e^{ik(L-1)} \\
1 & e^{2ik} & e^{4ik} & \cdots & e^{2ik(L-1)} \\
\vdots & \vdots & \vdots & \ddots & \vdots \\
1 & e^{ik(L-1)} & e^{2ik(L-1)} &\cdots & e^{ik(L-1)^2}
\end{pmatrix},
\end{align}
where $k=2\pi/L$.

%
%
%
%
%
%
%
%
%

Because of the relation~\eqref{347}, we have
\begin{align}
H^{\flat_0}=H^{\sharp_0}=H^\dag;
\end{align}
in particular, the latter equation implies that the present Hamiltonian \eqref{345} is \textit{not} pseudo-Hermitian: $H^{\sharp_0}\neq H$,
which is in agreement with Theorem \ref{thm1} because the eigenvalues~\eqref{346} are complex numbers, for $g\neq0$.

\subsubsection{A periodic system with asymmetric hopping and impurities}


The situation further changes drastically when we add one element on the diagonal of the periodic system~\eqref{345}.
The additional diagonal element physically represents an impurity potential and hence means the loss of the translational invariance.
Let us now numerically treat the Hamiltonian
\begin{align}\label{357}
H_{HN}&:=-\hop\sum_{x=1}^{L}
\qty(e^g \dyad{x+1}{x}+e^{-g}\dyad{x}{x+1})
+v_{x_0}\dyad{x_0},
\end{align}
where $\hop, g, v_{x_0}\in \mathbb{R}$.
(The Hamiltonian is again different from \eqref{322} and \eqref{343}, but we still use the same notation $H_{HN}$ for brevity.
This notation applies only throughout the present subsection.)

The spectrum is exemplified in Fig.~\ref{fig1}(c) for $L=11$, $\hop=1$, $g=0.1$ with the impurity potential set to $v_6=1$ at $x_0=6$, the medium point of the lattice.
When compared to the spectrum in Fig.~\ref{fig1}(b), most of the eigenvalues remain complex while we have two more real eigenvalues.
Figure~\ref{fig2}, on the other hand, shows the development of the elements of the isospectral Hamiltonians generated.
\begin{figure}
\begin{minipage}[t]{0.3\textwidth}
\vspace{0mm}
\includegraphics[width=\textwidth]{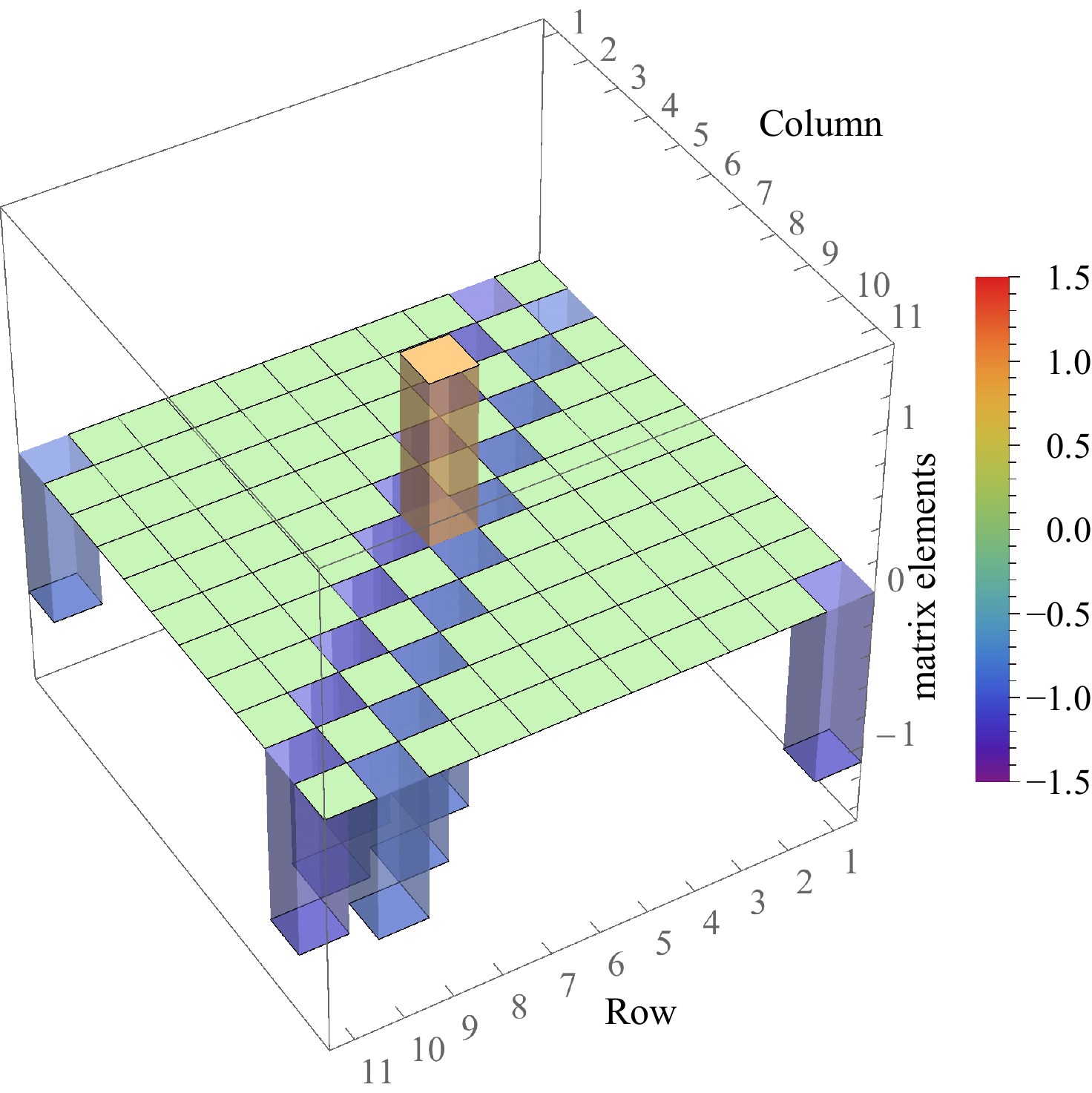}
\centering (a)
\end{minipage}
\hfill
\begin{minipage}[t]{0.3\textwidth}
\vspace{0mm}
\includegraphics[width=\textwidth]{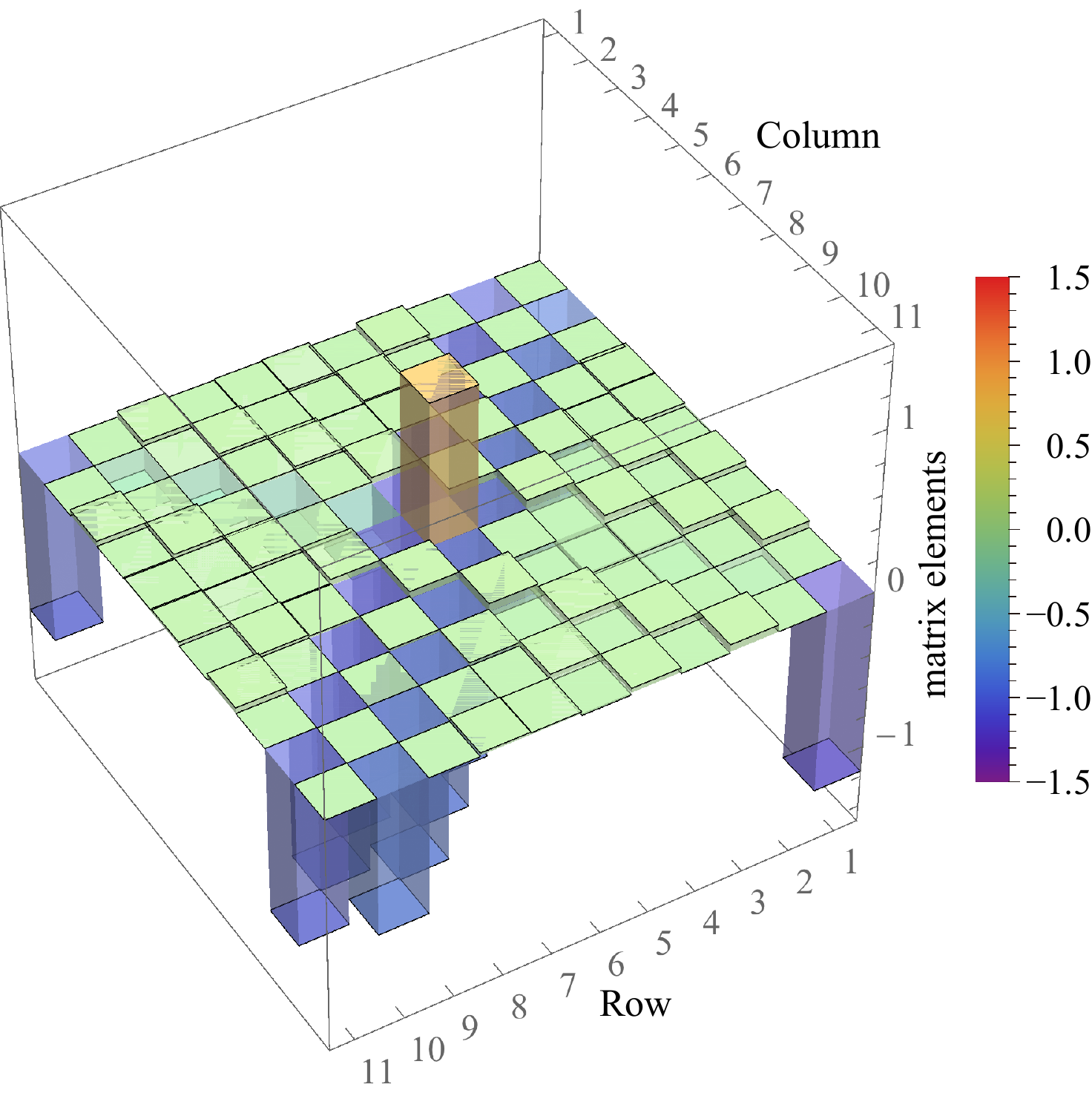}
\centering (b)
\end{minipage}
\hfill
\begin{minipage}[t]{0.3\textwidth}
\vspace{0mm}
\includegraphics[width=\textwidth]{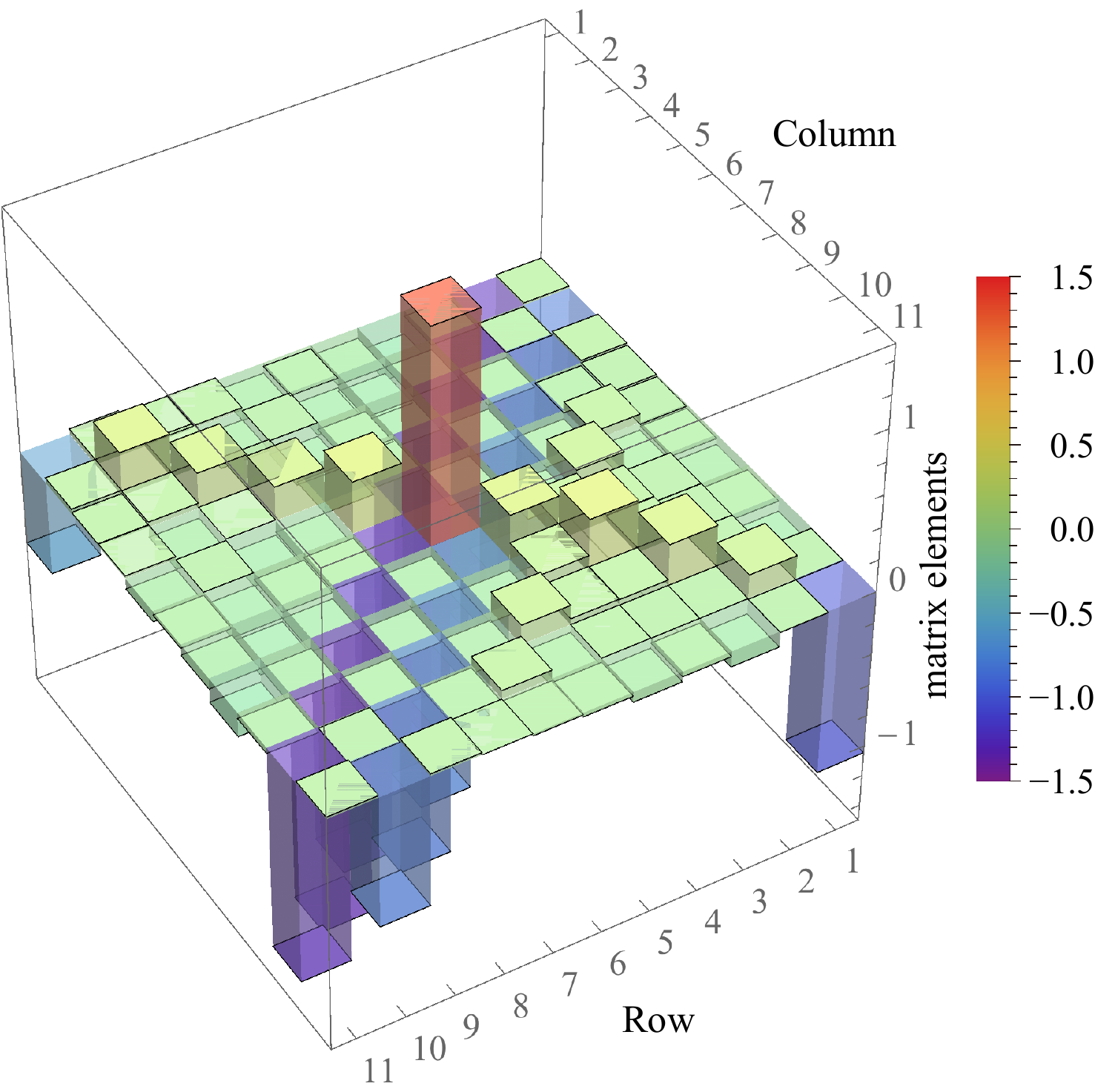}
\centering (c)
\end{minipage}
\vspace{\baselineskip}
\\
\begin{minipage}[t]{0.3\textwidth}
\vspace{0mm}
\includegraphics[width=\textwidth]{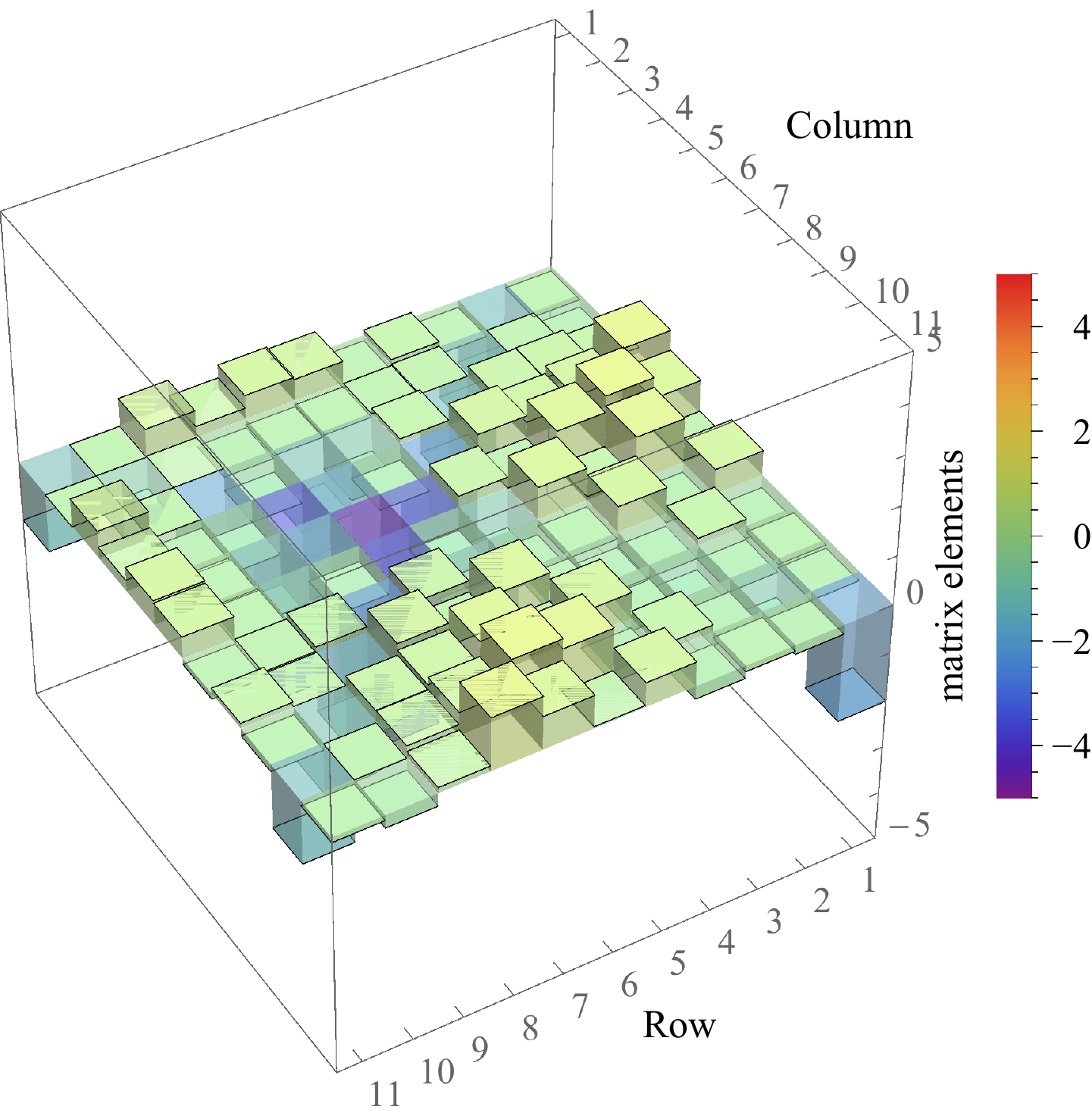}
\centering (d)
\end{minipage}
\hfill
\begin{minipage}[t]{0.3\textwidth}
\vspace{0mm}
\includegraphics[width=\textwidth]{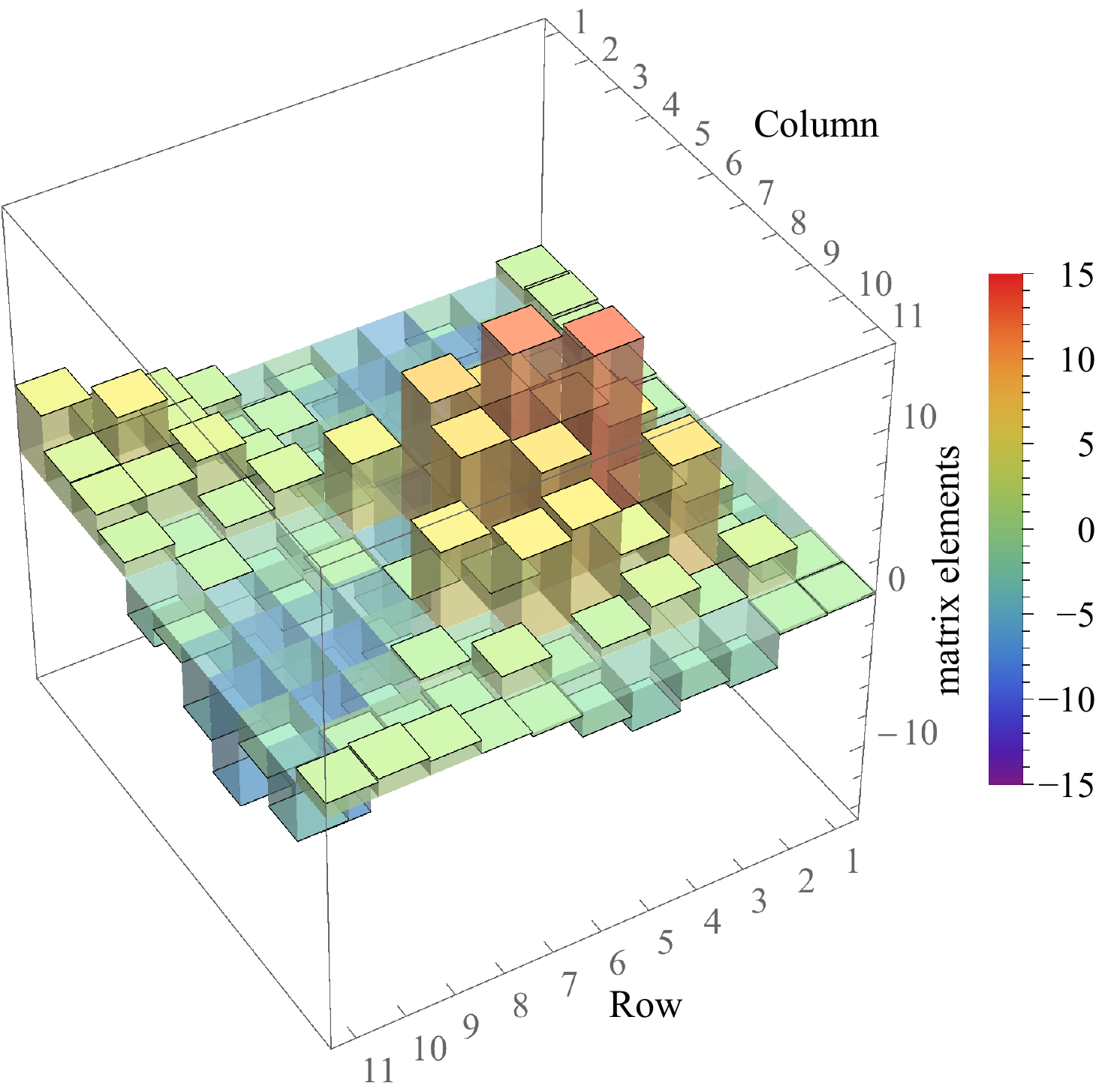}
\centering (e)
\end{minipage}
\hfill
\begin{minipage}[t]{0.3\textwidth}
\vspace{0mm}
\includegraphics[width=\textwidth]{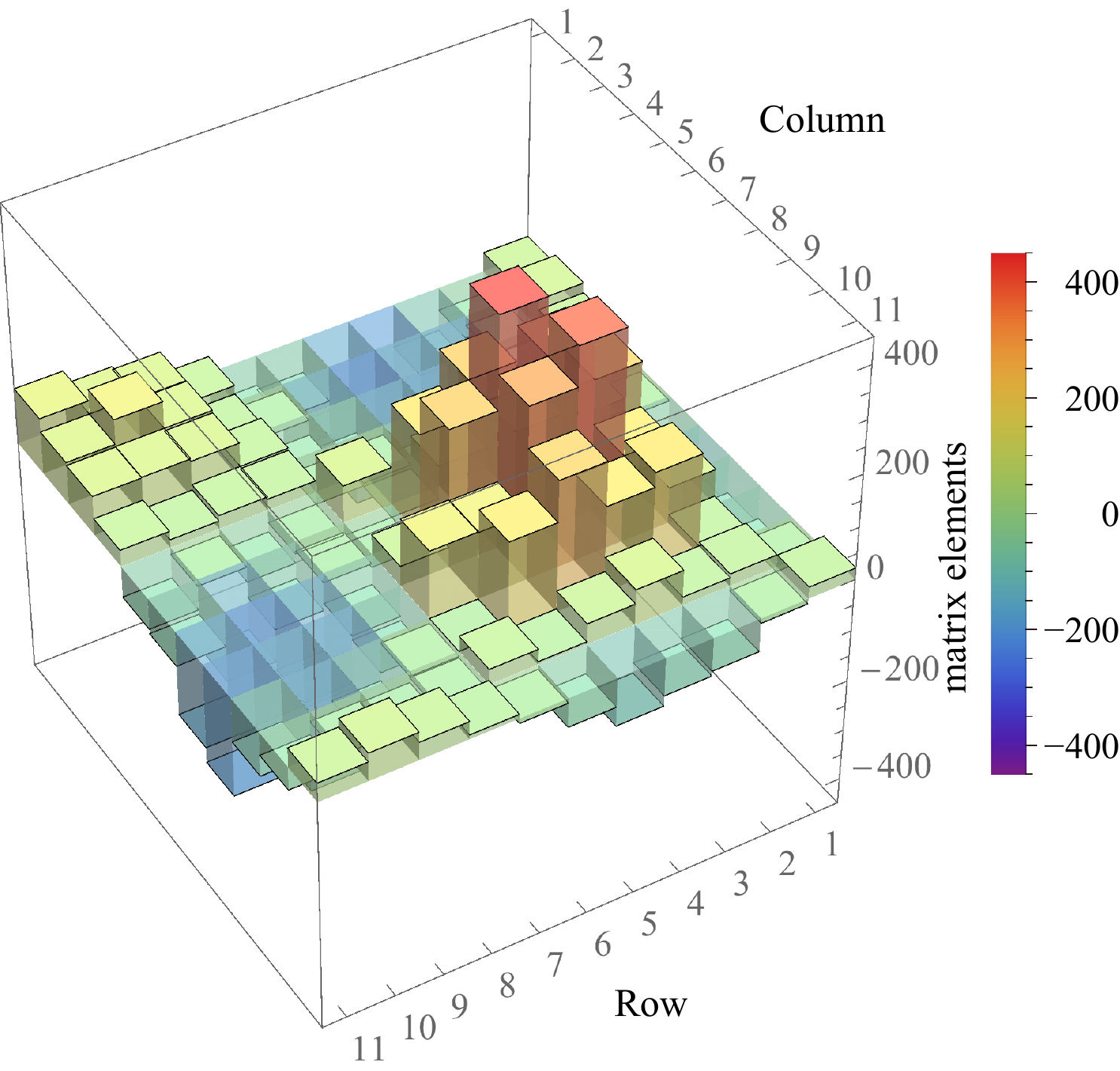}
\centering (f)
\end{minipage}
\vspace{\baselineskip}
\\
\begin{minipage}[t]{0.3\textwidth}
\vspace{0mm}
\includegraphics[width=\textwidth]{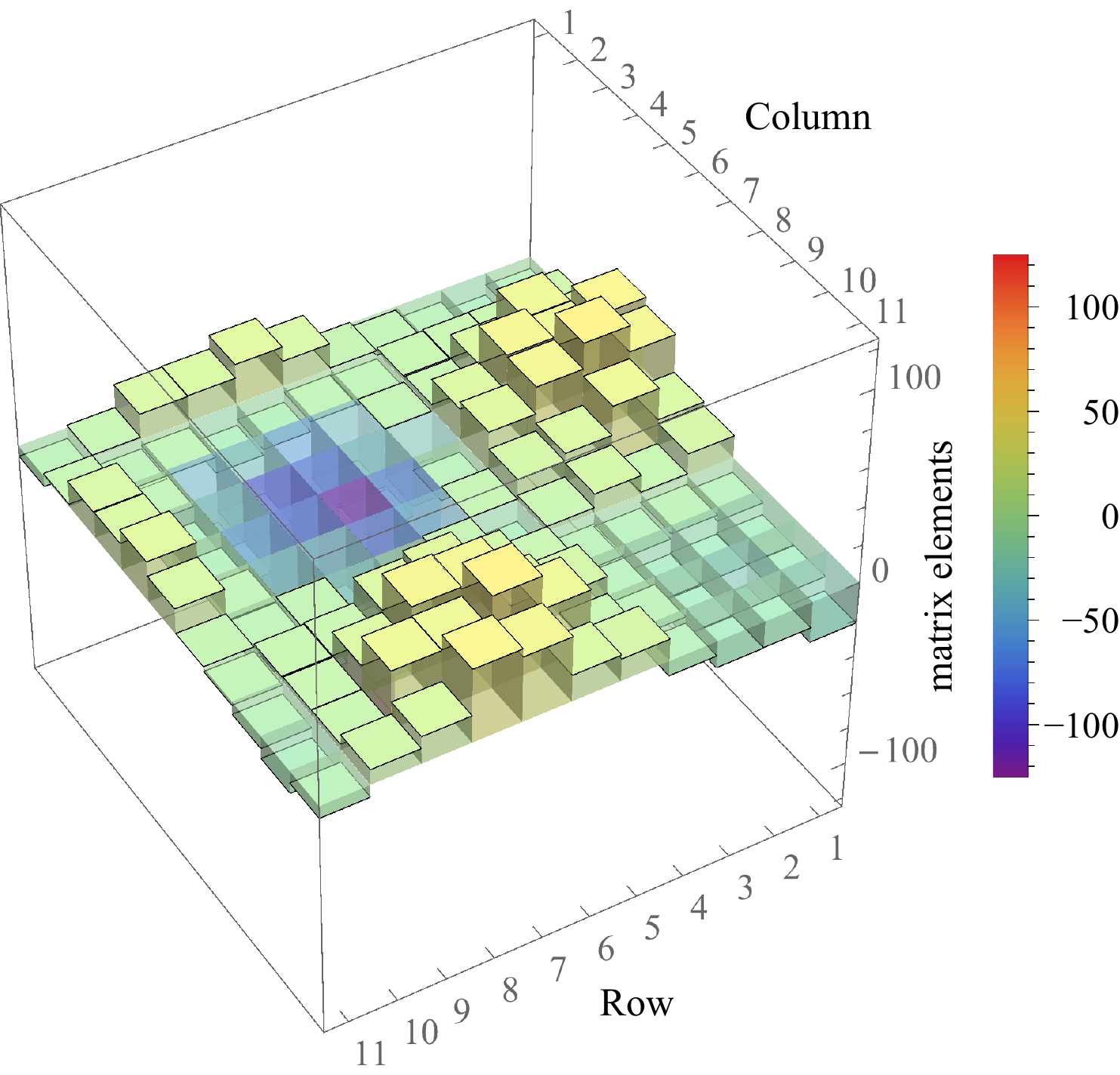}
\centering (g)
\end{minipage}
\hfill
\begin{minipage}[t]{0.3\textwidth}
\vspace{0mm}
\includegraphics[width=\textwidth]{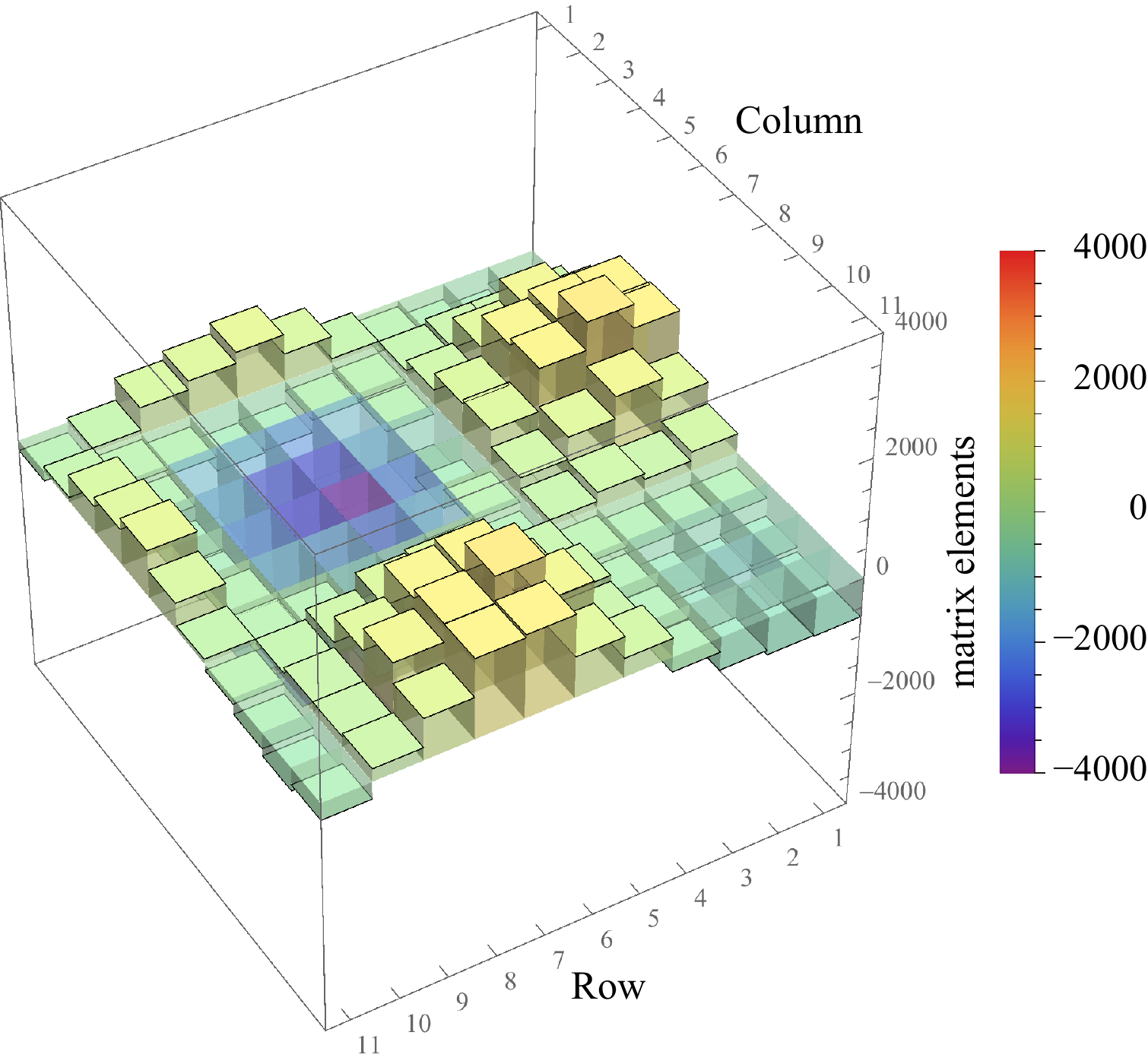}
\centering (h)
\end{minipage}
\hfill
\begin{minipage}[t]{0.3\textwidth}
\vspace{0mm}
\includegraphics[width=\textwidth]{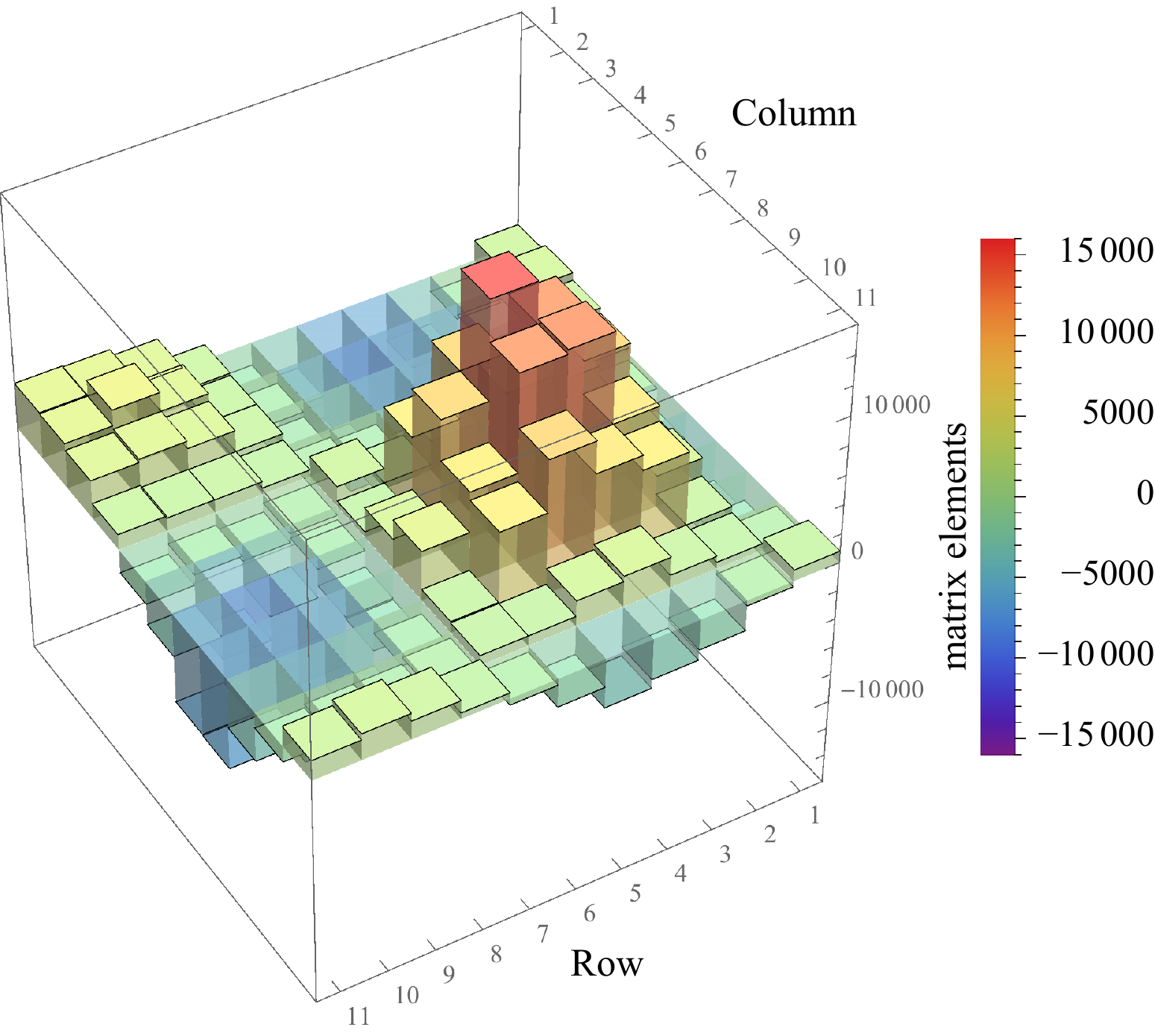}
\centering (i)
\end{minipage}
\caption{The elements of the Hamiltonians developed from the Hamiltonian~\eqref{357}.
We set $L=11$, $\hop=1$, $g=0.1$ with the impurity potential set to $v_6=1$ at $x_0=6$.
The elements of the original and subsequent isospectral Hamiltonians are all real: (a) $H$, 
(b) $\qty(H^\dag)^{\flat_0}$, (c) $\qty(H^\dag)^{\sharp_0}$,
(d) $\qty(H^\dag)^{\flat_1}$, (e) $\qty(H^\dag)^{\sharp_1}$,
(f) $\qty(H^\dag)^{\flat_{2,\alpha}}$, (g) $\qty(H^\dag)^{\sharp_{2,\alpha}}$,
(h) $\qty(H^\dag)^{\flat_{2,\alpha}}$, (i) $\qty(H^\dag)^{\sharp_{2,\beta}}$.
Note that the vertical scales may be different from each other.
}
\label{fig2}
\end{figure}
We can observe that once the translational invariance is broken, far-off-diagonal elements quickly develop non-zero values;
some of the elements grow rapidly as we advance the iteration.
We do not see any self-adjointness here, including the pseudo-Hermiticity;
even $\qty(H^\dag)^{\flat_0}$ in Fig.~\ref{fig2}(b) has nonzero elements at $(2,9)$, $(3,8)$, and so on. 

We next introduce diagonal random potentials over the entire system:
\begin{align}\label{357-1}
H_{HN}&:=-\hop\sum_{x=1}^{L}
\qty(e^g \dyad{x+1}{x}+e^{-g}\dyad{x}{x+1})
+\sum_{x=1}^L v_x\dyad{x},
\end{align}
where $\hop, g, v_{x}\in \mathbb{R}$ for all $x$.
We choose each of $v_x$ randomly from the box distribution $[-V,V]$.

This is the setting of the original HN model.
It was discovered~\cite{Hatano96,Hatano97} for this model that the hopping asymmetry and the randomness compete with each other.
More precisely, an eigenvector with a complex eigenvalue extends over the system, signifying a current circulating around the system as the plane wave solution~\eqref{345-1} does in the clean limit.
On the other hand, an eigenvector with a real eigenvalue (except for ones being real owing to the symmetry, such as the one on the right end of the spectrum in Fig.~\ref{fig1}(b)) is localized due to the random potential.
In short, whether an eigenvector is localized or extended corresponds to whether its eigenvalue is real or complex.

In the clean limit \eqref{343}, the eigenvalues, as is exemplified in Fig.~\ref{fig1}(b), are complex except the one being real because of the mirror symmetry with respect to the real axis. 
As we turn on and hike up the amplitude $V$ of the random potential, eigenvalues of each complex-conjugate pair approach to each other, collide with each other one pair by one pair, and respectively turn into two real eigenvalues.
In the dirty limit, all eigenvalues become real even though the Hamiltonian is still non-Hermitian because of the asymmetric hopping. 

Here we numerically check that the Hamiltonian is not pseudo-Hermitian in a moderately dirty case with complex eigenvalues still remaining, but it is so in a quite dirty case with real eigenvalues only.
In the first case, we used the parameter set $L=11$, $\hop=1$, $g=0.1$, and $V=1$ as in Fig.~\ref{fig1}(d); we find that only five out of eleven eigenvalues are real in this particular random realization.
The elements of the generated Hamiltonians are shown in Fig.~\ref{fig3}.
\begin{figure}
\begin{minipage}[t]{0.3\textwidth}
\vspace{0mm}
\includegraphics[width=\textwidth]{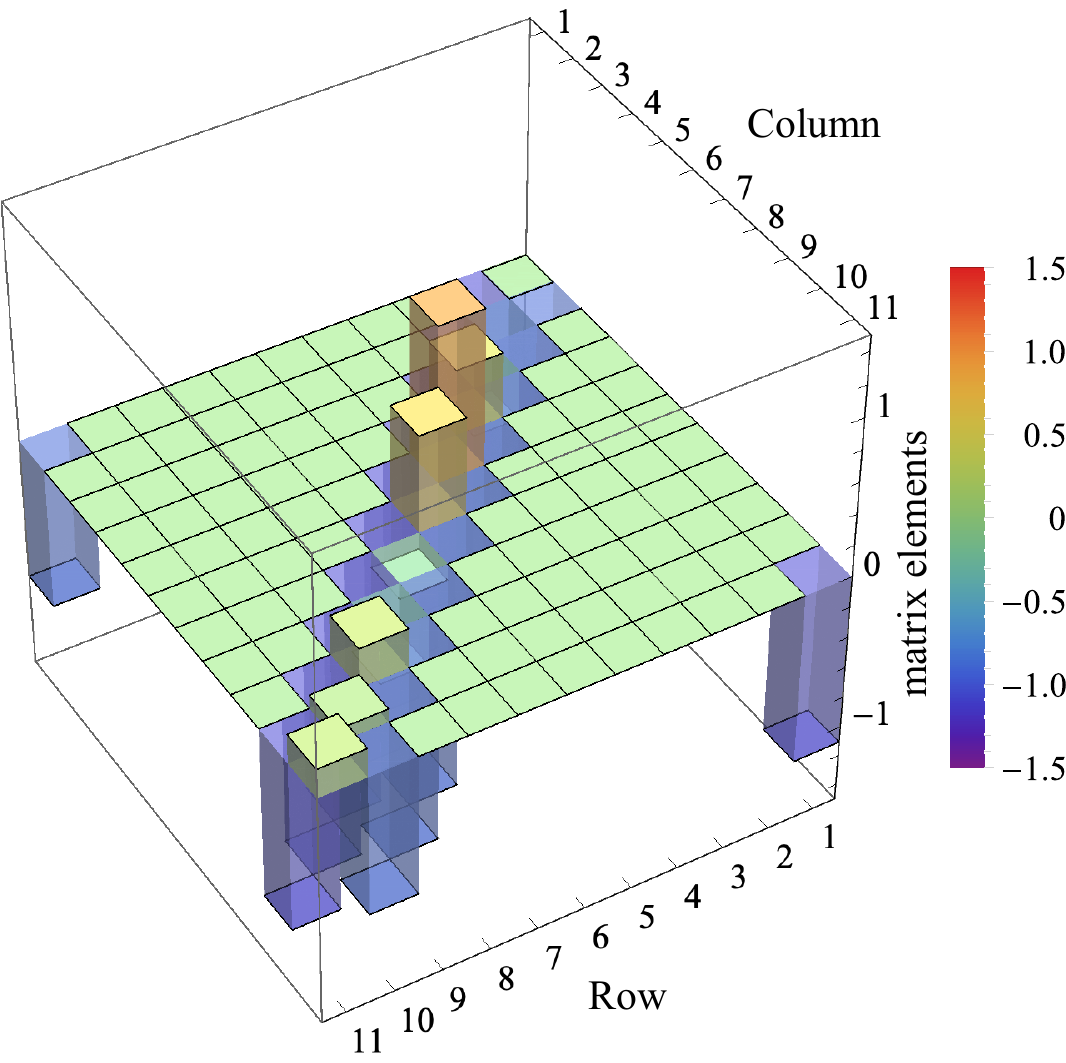}
\centering (a)
\end{minipage}
\hfill
\begin{minipage}[t]{0.3\textwidth}
\vspace{0mm}
\includegraphics[width=\textwidth]{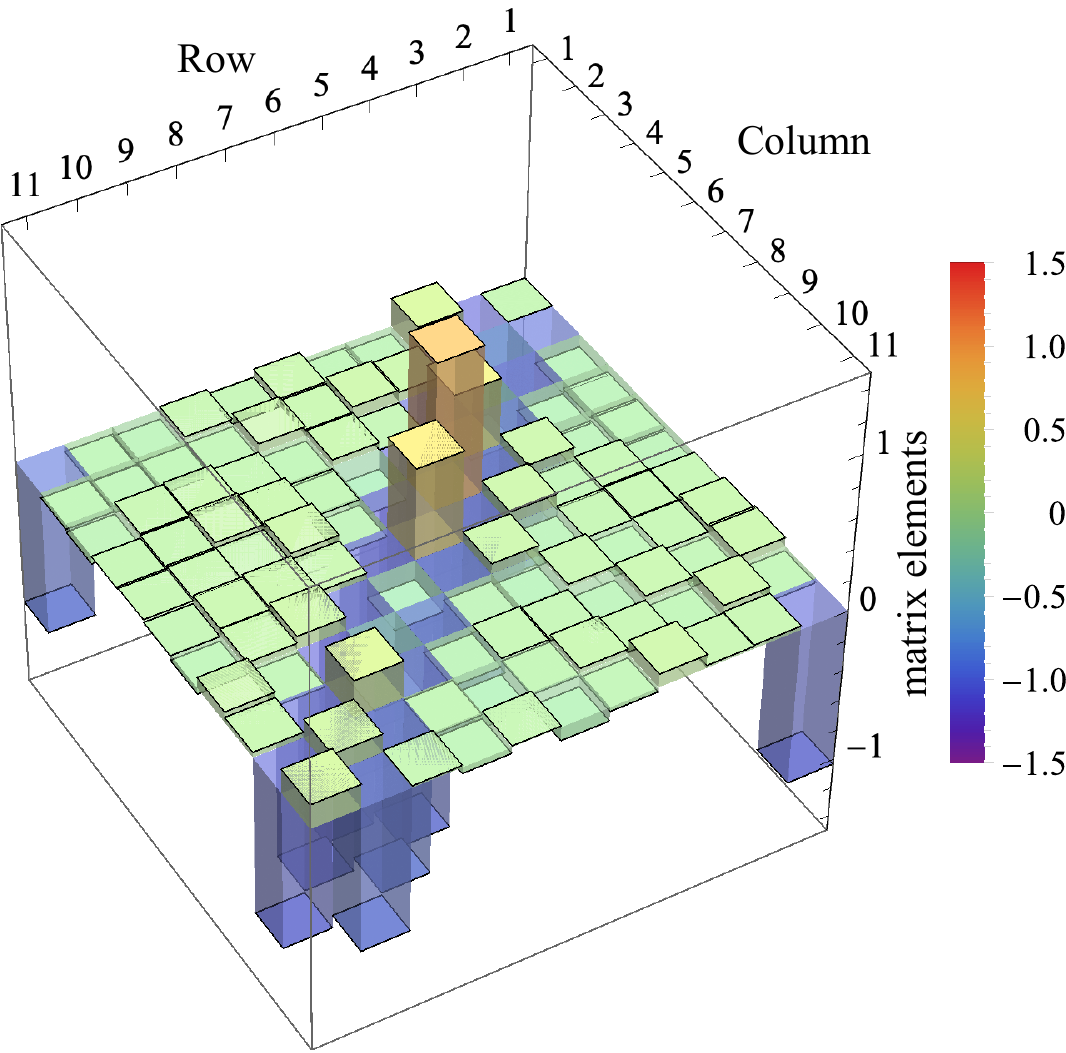}
\centering (b)
\end{minipage}
\hfill
\begin{minipage}[t]{0.3\textwidth}
\vspace{0mm}
\includegraphics[width=\textwidth]{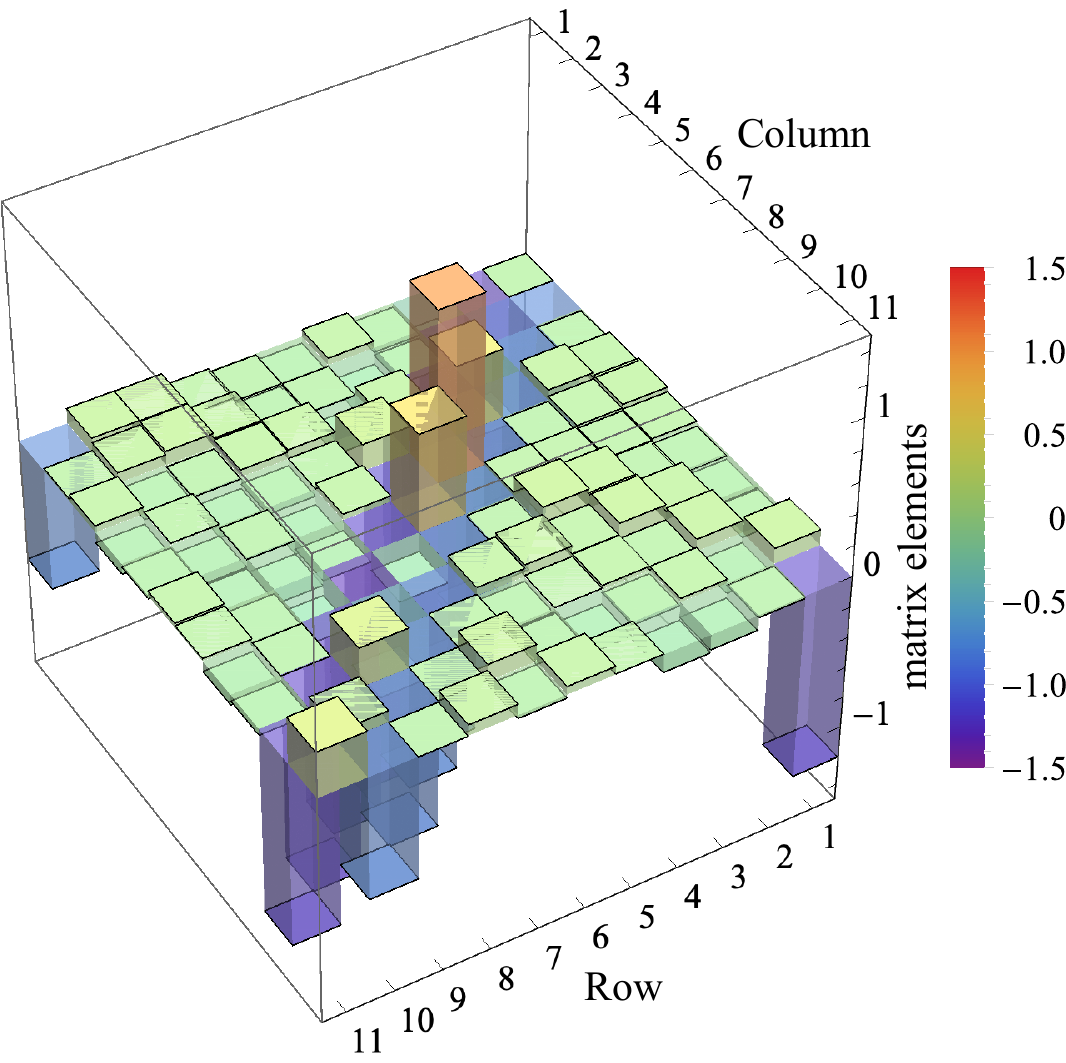}
\centering (c)
\end{minipage}
\vspace{\baselineskip}
\\
\begin{minipage}[t]{0.3\textwidth}
\vspace{0mm}
\includegraphics[width=\textwidth]{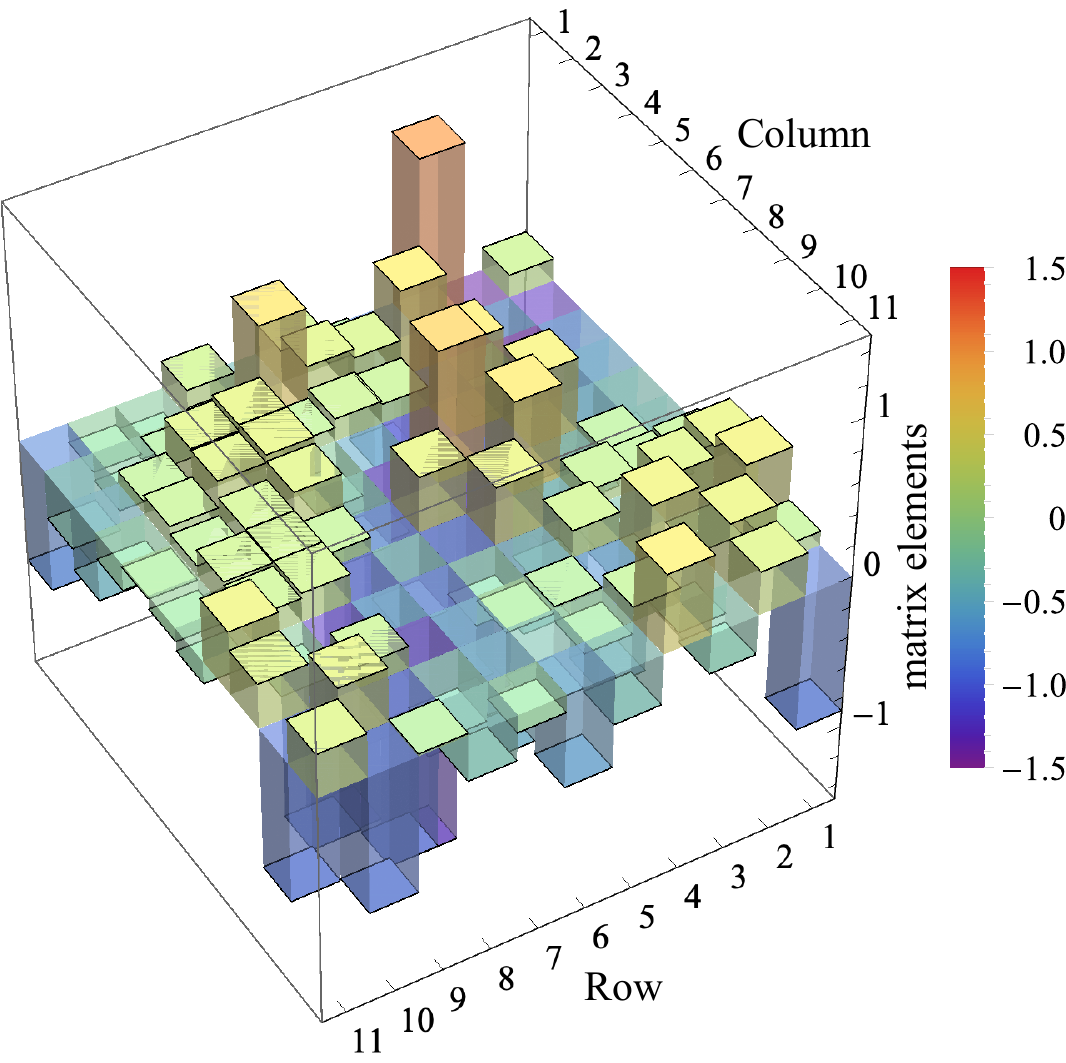}
\centering (d)
\end{minipage}
\hfill
\begin{minipage}[t]{0.3\textwidth}
\vspace{0mm}
\includegraphics[width=\textwidth]{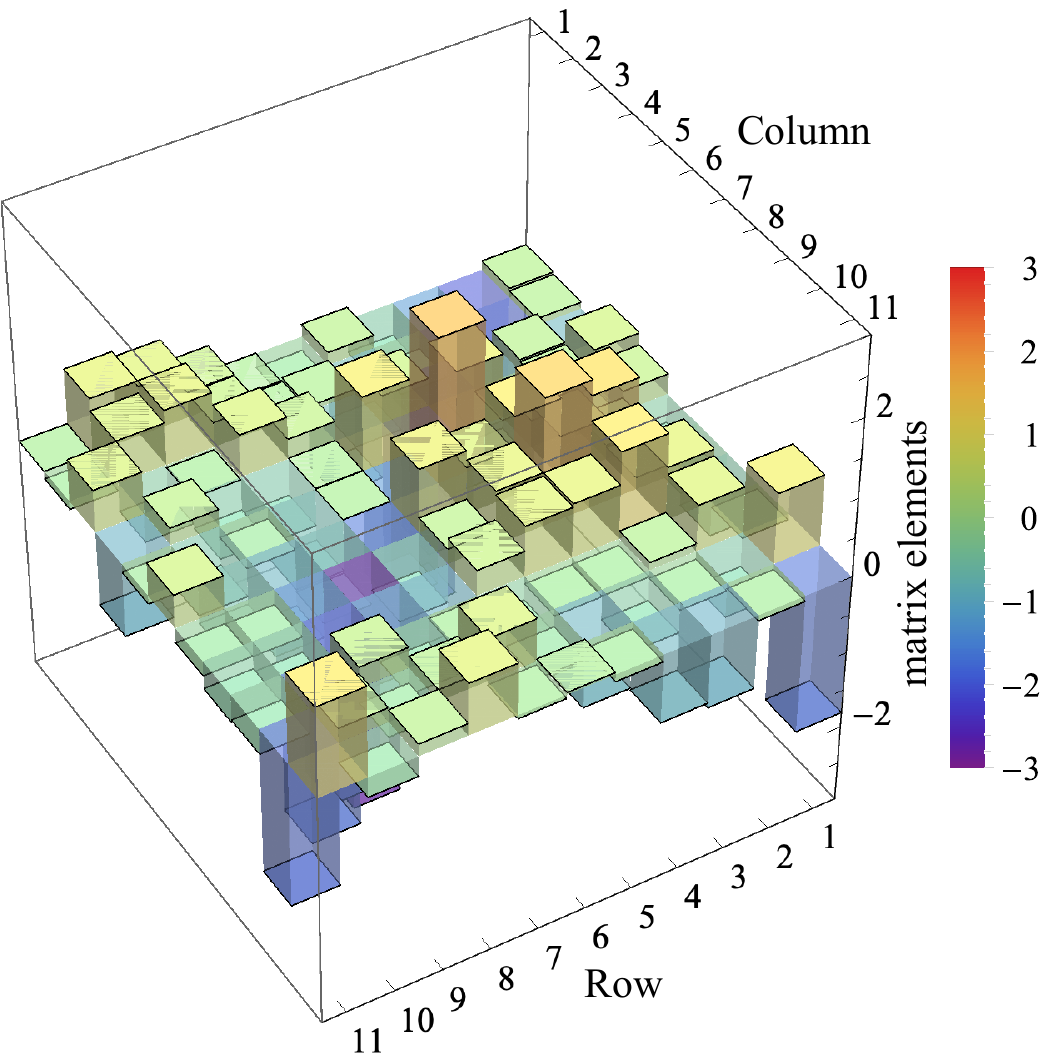}
\centering (e)
\end{minipage}
\hfill
\begin{minipage}[t]{0.3\textwidth}
\vspace{0mm}
\includegraphics[width=\textwidth]{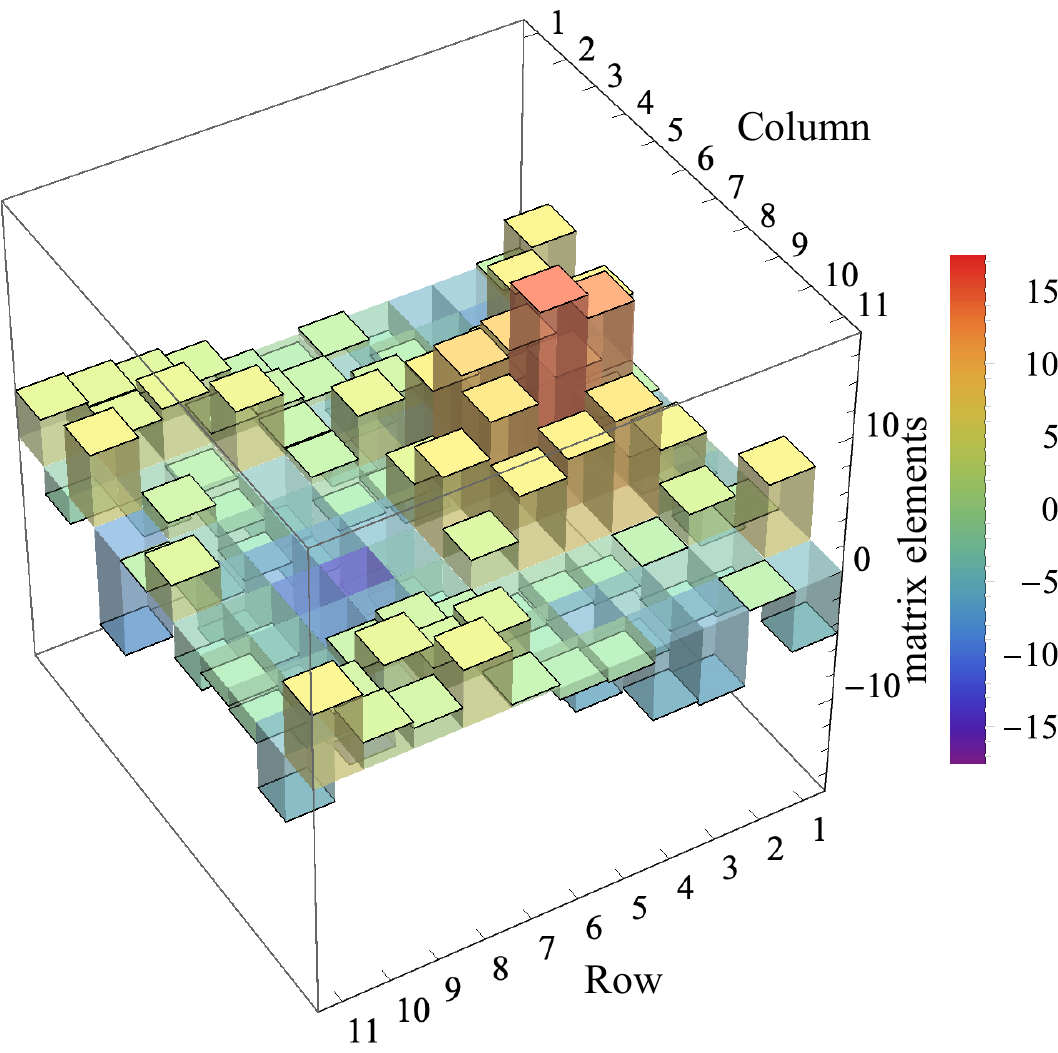}
\centering (f)
\end{minipage}
\vspace{\baselineskip}
\\
\begin{minipage}[t]{0.3\textwidth}
\vspace{0mm}
\includegraphics[width=\textwidth]{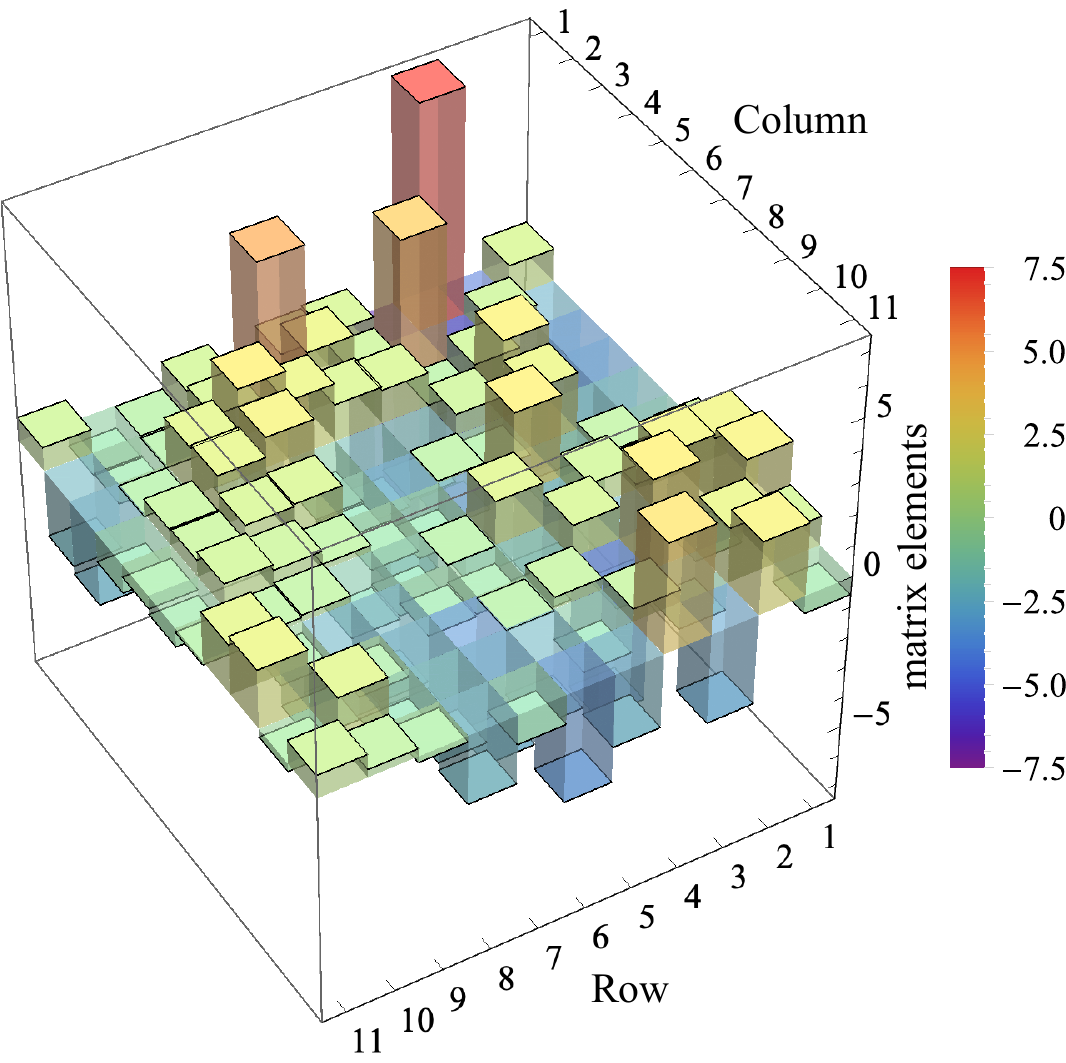}
\centering (g)
\end{minipage}
\hfill
\begin{minipage}[t]{0.3\textwidth}
\vspace{0mm}
\includegraphics[width=\textwidth]{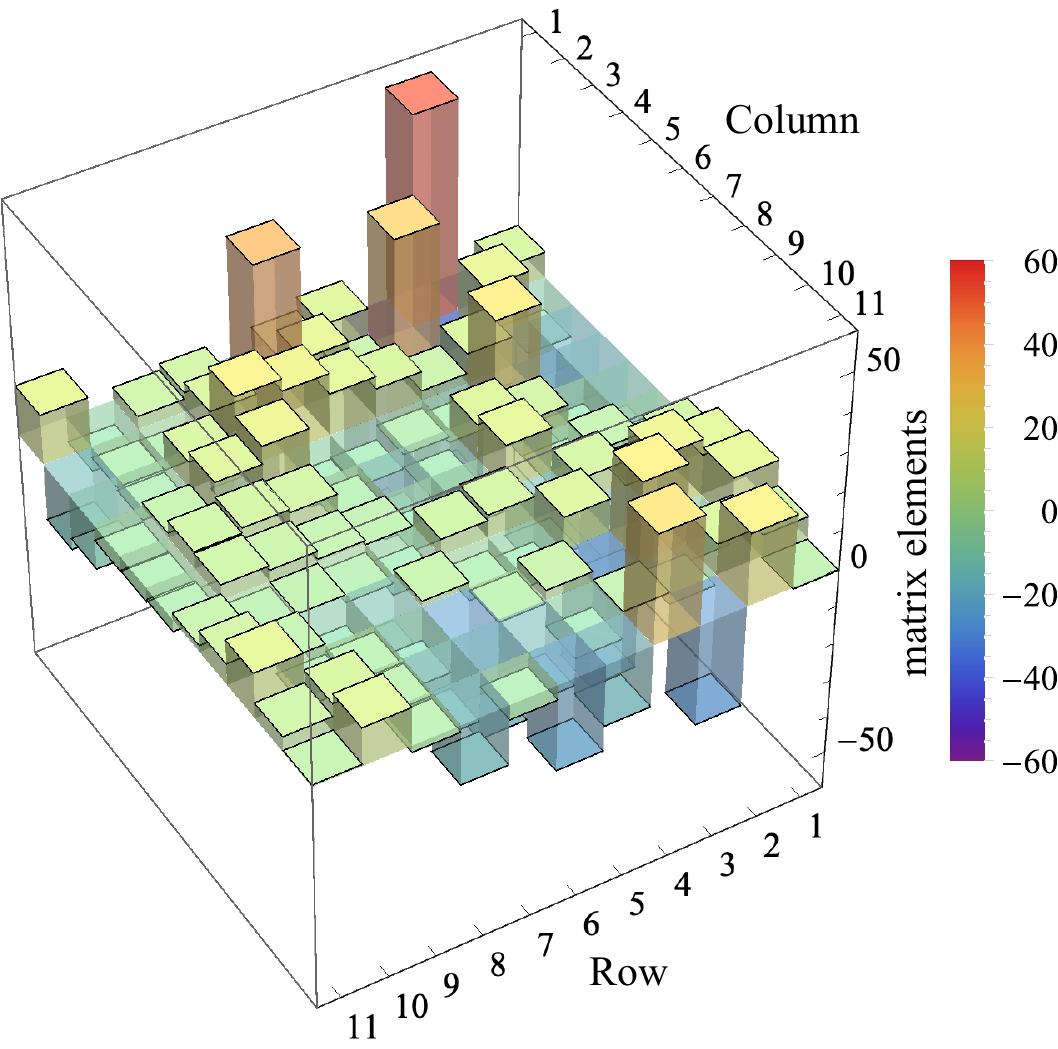}
\centering (h)
\end{minipage}
\hfill
\begin{minipage}[t]{0.3\textwidth}
\vspace{0mm}
\includegraphics[width=\textwidth]{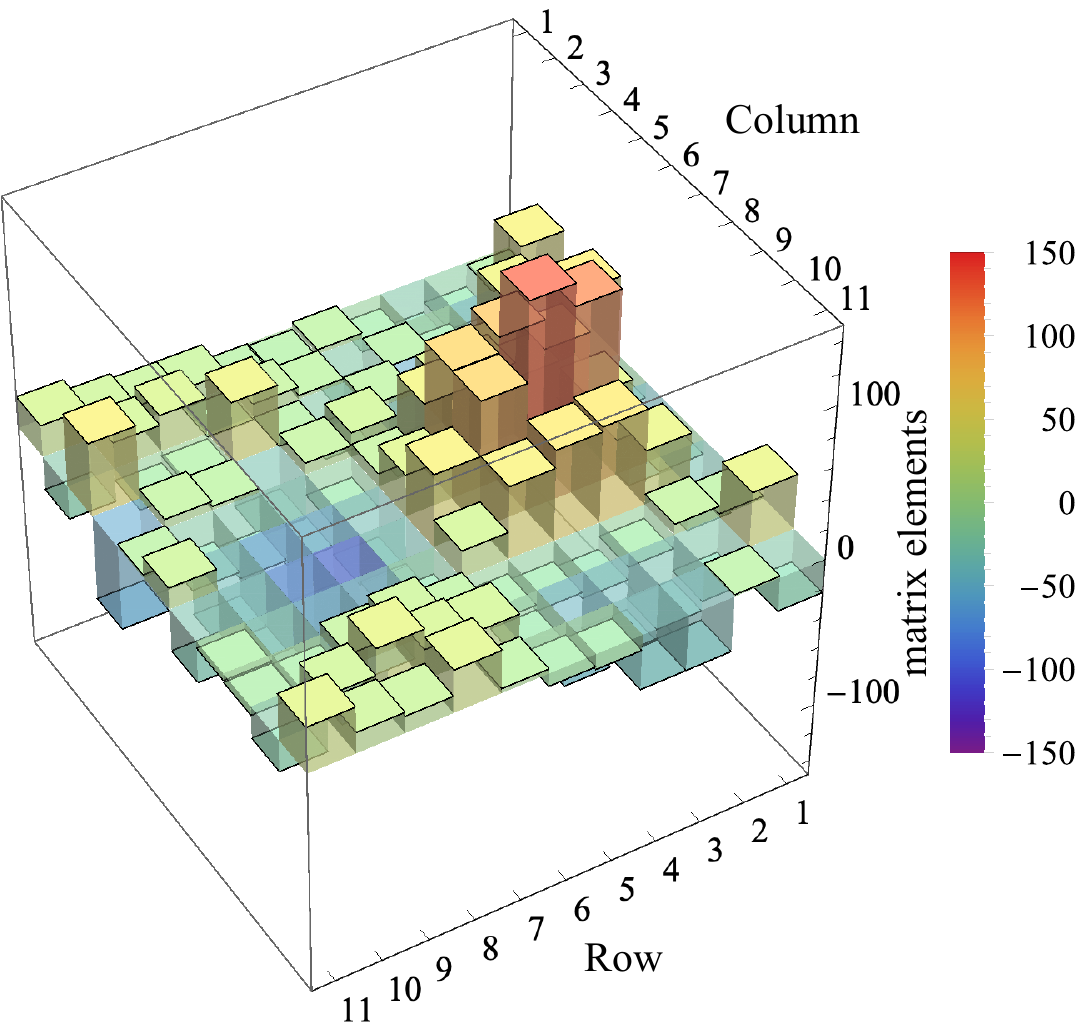}
\centering (i)
\end{minipage}
\caption{The elements of the Hamiltonians developed from the Hamiltonian~\eqref{357-1}.
We set $L=11$, $\hop=1$, $g=0.1$ with the potential at each site randomly chosen from the box distribution $[-1,1]$.
The elements of the original and subsequent isospectral Hamiltonians are all real: (a) $H$, 
(b) $\qty(H^\dag)^{\flat_0}$, (c) $\qty(H^\dag)^{\sharp_0}$,
(d) $\qty(H^\dag)^{\flat_1}$, (e) $\qty(H^\dag)^{\sharp_1}$,
(f) $\qty(H^\dag)^{\flat_{2,\alpha}}$, (g) $\qty(H^\dag)^{\sharp_{2,\alpha}}$,
(h) $\qty(H^\dag)^{\flat_{2,\alpha}}$, (i) $\qty(H^\dag)^{\sharp_{2,\beta}}$.
Note that the vertical scales may be different from each other.
}
\label{fig3}
\end{figure}
Since some of the eigenvalues are complex, we do not see any self-adjointness here either, including the pseudo-Hermiticity.
The Hamiltonian $\qty(H^\dag)^{\flat_0}$ has non-zero elements at which the original Hamiltonian $H$ does not.

In the second case, we used the parameter set $L=11$, $\hop=1$, $g=0.1$, and $V=2.5$ as in Fig.~\ref{fig1}(e); 
because of the stronger random amplitude than the case of Fig.~\ref{fig1}(d), all eigenvalues are now real;
the randomness has won over the asymmetric hopping and made all eigenvectors localized.
Thanks to the reality of all the eigenvalues, we finally find the pseudo-Hermiticity $H^{\sharp_1}=H$ in this case.
The magnitude of the matrix elements of the subsequent adjoints grow much more rapidly 
than in Figs.~\ref{fig2} and~\ref{fig3}, being too large to be equal to $H$.
The results throughout the present subsection may suggest that the magnitude of the matrix elements grow more rapidly when the eigenvalues are all real than when some of them are complex.

\subsection{An example of diagonal non-self-adjointness: the $PT$-symmetric model}
\label{subsect3.3}

In this subsection, we consider models of non-self-adjointness that originates in the diagonal elements. 
Specifically, we consider two types of $PT$-symmetric models.

The $PT$-symmetry was originally introduced in order to replace the Hermiticity as the condition of the reality of eigenvalues~\cite{BB98}; 
it was argued that the Hermiticity may be a too mathematical condition, not a physical one, whereas the universe should be built up under physical conditions, not mathematical ones.
Indeed, the Hermiticity is only a sufficient condition for the reality of eigenvalues, not a necessary one.
On the other hand, the combination of an anti-linear symmetry with any linear symmetries, including the $PT$-symmetry, can have a parameter region of real eigenvalues.

When the Hamiltonian has a time-reversal symmetry $T$, which is an anti-linear symmetry, as in $[H,T]=0$, we can deduce the following~\cite{Hatano19}.
Suppose that we have a solution $H\psi=E\psi$ without degeneracy. 
This is followed by $TH\psi=TE\psi$.
The left-hand side is equal to $HT\psi$ because of the commutability $[H,T]=0$, while the right-hand side is equal to $E^\ast T\psi$ because of the anti-linearity of $T$, which yields $HT\psi=E^\ast T\psi$.
We can thereby conclude that either (i) $E$ is real and $T\psi\propto \psi$ or (ii) $E$ is complex and there is a complex-conjugate eigenvalue $E^\ast$ with the eigenfunction $T\psi$, namely the time-reversal of the eigenfunction $\psi$ of the eigenvalue $E$.
This argument is intact when we instead consider the $PT$-symmetry~\cite{BB98} or any combinations of $T$ and linear symmetries.

In the two examples of $PT$-symmetric models below, we will indeed see regions of real and complex eigenvalues, which are often called the $PT$-unbroken and $PT$-broken phases.
In both cases, we will see a skewed version of our argument when all eigenvalues are pure imaginary.

\subsubsection{The simplest $2\times 2$ matrix}
\label{subusubsec3.2.1}

Let us now consider a simple version of the $PT$-symmetric system on a two lattice points.
The Hamiltonian and its $\dag$-adjoint are given by
\begin{align}\label{363}
H_{PT}:=\begin{pmatrix}
i\gamma & -\hop \\
-\hop & -i\gamma
\end{pmatrix},
\qquad
{H_{PT}}^\dag=\begin{pmatrix}
-i\gamma & -\hop \\
-\hop & i\gamma
\end{pmatrix},
\end{align}
where $\hop,\gamma\in\mathbb{R}$.
The $PT$-symmetry in this case is given as follows.
First, the $P$ (parity) symmetry is the mirror symmetry with respect to the exchange of the first and second rows and columns.
This linear operation is expressed by 
\begin{align}
P=\begin{pmatrix}
0 & 1 \\
1 & 0
\end{pmatrix}.
\end{align}
On the other hand, the anti-linear $T$ symmetry is given by complex conjugation $\ast$.
Combining them, we indeed confirm that $H_{PT}$ is $PT$-symmetric  in the sense of $[H,PT]=0$:
\begin{align}
&(PT)H_{PT}\vec{f}
=\begin{pmatrix}
0 & 1 \\
1 & 0
\end{pmatrix}
\begin{pmatrix}
-i\gamma & -\hop \\
-\hop & i\gamma
\end{pmatrix}
\begin{pmatrix}
f_1^\ast \\
f_2^\ast
\end{pmatrix}
=\begin{pmatrix}
 -\hop & i\gamma  \\
 -i\gamma &-\hop 
\end{pmatrix}
\begin{pmatrix}
f_1^\ast \\
f_2^\ast
\end{pmatrix}
\\
&=H_{PT}(PT)f=\begin{pmatrix}
i\gamma & -\hop \\
-\hop & -i\gamma
\end{pmatrix}
\begin{pmatrix}
0 & 1 \\
1 & 0
\end{pmatrix}
\begin{pmatrix}
f_1^\ast \\
f_2^\ast
\end{pmatrix}
\end{align}
for all $\vec{f}^T=\qty(f_1,f_2)$ with $f_1,f_2\in\mathbb{C}$. 
We find that the eigenvalues $\pm\sqrt{{\hop}^2-\gamma^2}$ are both real numbers when $\abs{\hop}>\abs{\gamma}$ but $\pm i\sqrt{\gamma^2-{\hop}^2}$ are pure imaginary ones when $\abs{\hop}<\abs{\gamma}$.
We can understand this physically in the following way.
For explanatory purposes, let us set $\hop=0$ and $\gamma>0$ for the moment.
In the time-dependent frame of the Schr\"{o}dinger equation $i\hbar \partial\Psi(t)/\partial t=H_{PT}\Psi(t)$, the amplitude of the two-component vector $\Psi(t)$ exponentially increases at the first site as $\qty(\Psi(t))_1\propto \exp(\gamma t)$, while it exponentially decreases at the second site as $\qty(\Psi(t))_2\propto \exp(-\gamma t)$.
We can attribute these exponential changes to injection of the probability onto the first site, say the source, and the removal of the probability from the second site, say the drain.
The imaginary eigenvalues $\pm i\gamma$ in the case of $\hop=0$ signify the instability of exponential growth and decay at the respective sites.
As we turn on the coupling $\hop$ between the two sites, the probability injected on the source starts to flow to the drain to be removed.
The instability is thus gradually remedied as the imaginary part of the eigenvalues $\pm i\sqrt{\gamma^2-{\hop}^2}$ diminishes, and finally is resolved when $\hop$ exceeds $\gamma$, as is signified by the real eigenvalues $\pm\sqrt{{\hop}^2-\gamma^2}$ evolving out of the imaginary eigenvalues.
In short, the condition $\abs{\hop}>\abs{\gamma}$ for the reality of the eigenvalues accomplishes a steady state in which the probability injected on the source constantly flows to the drain and removed.

We begin our analysis with the case of real eigenvalues, namely in the $PT$-unbroken phase.
For brevity, let us set $\hop=1$ and $0<\gamma<1$ hereafter, and choose the ordering of the eigenvalues such that
\begin{align}
E_1=-\beta\qquad \mbox{and}\qquad E_2=+\beta,
\end{align}
where $\beta=\sqrt{1-\gamma^2}\in\mathbb{R}$.
We can choose the eigenvectors in the following forms:
\begin{align}
\varphi_1\supzero&:=\frac{1}{\sqrt{2\beta}}
\begin{pmatrix}
1\\
i\gamma+\beta\\
\end{pmatrix},
\qquad
\varphi_2\supzero:=\frac{-1}{\sqrt{2\beta}}
\begin{pmatrix}
i\gamma+\beta\\
-1\\
\end{pmatrix},
\\
\psi_1\supzero&:=\frac{1}{\sqrt{2\beta}}
\begin{pmatrix}
i\gamma+\beta\\
1\\
\end{pmatrix},
\qquad
\psi_2\supzero:=\frac{1}{\sqrt{2\beta}}
\begin{pmatrix}
-1\\
i\gamma+\beta\\
\end{pmatrix}.
\end{align}
We can easily confirm that they satisfy the biorthonormality $\scp<\psi_m\supzero,\varphi_n\supzero>=\delta_{m,n}$ for $m,n=1,2$ as well as the resolution of the identity $\sum_{n=1}^2 \dyad{\varphi_n\supzero}{\psi_n\supzero}=\mathbb{I}$.

We then find the intertwining operators in \eqref{24} as
\begin{align}
S_\varphi\supzero=\frac{1}{\beta}
\begin{pmatrix}
1 & -i\gamma \\
i\gamma & 1 \\
\end{pmatrix},
\qquad
S_\psi\supzero=\frac{1}{\beta}
\begin{pmatrix}
1 & i\gamma \\
-i\gamma & 1 \\
\end{pmatrix}.
\end{align}
They are obviously inverse matrices to each other.
They define the two adjoints in \eqref{27} with \eqref{28} yielding
\begin{align}\label{368}
{H_{PT}}^{\flat_0}=\frac{-1}{\beta^2}
\begin{pmatrix}
i\gamma(3+\gamma^2) & 1+3\gamma^2 \\
1+3\gamma^2 & -i\gamma(3+\gamma^2) \\
\end{pmatrix},
\qquad
{H_{PT}}^{\sharp_0}=
\begin{pmatrix}
i\gamma & -1 \\
-1 & -i\gamma \\
\end{pmatrix},
\end{align}
which indicates the pseudo-Hermiticity ${H_{PT}}^{\sharp_0}=H_{PT}$.
We can also confirm the relations \eqref{29}.

The vectors \eqref{210}, which start the second iteration, are given by
\begin{align}
\varphi_1\supone&=\frac{1}{\sqrt{2\beta^3}}
\begin{pmatrix}
1-i\gamma\beta+\gamma^2\\
2i\gamma +\beta\\
\end{pmatrix},
\qquad
\varphi_2\supone=\frac{1}{\sqrt{2\beta^3}}
\begin{pmatrix}
-2i\gamma - \beta\\
1-i\gamma\beta+\gamma^2\\
\end{pmatrix},
\\
\psi_1\supone&=\frac{1}{\sqrt{2\beta^3}}
\begin{pmatrix}
2i\gamma +\beta\\
1-i\gamma\beta+\gamma^2\\
\end{pmatrix},
\qquad
\psi_2\supone=\frac{-1}{\sqrt{2\beta^3}}
\begin{pmatrix}
1-i\gamma\beta+\gamma^2\\
-2i\gamma - \beta\\
\end{pmatrix}.
\end{align}
They produce the further two intertwining operators \eqref{213} as in
\begin{align}
S_\varphi\supone=\frac{1}{\beta^3}
\begin{pmatrix}
1+3\gamma^2 & -i\gamma(3+\gamma^2) \\
i\gamma(3+\gamma^2) & 1+3\gamma^2 \\
\end{pmatrix},
\qquad
S_\psi\supone=\frac{1}{\beta^3}
\begin{pmatrix}
1+3\gamma^2 & i\gamma(3+\gamma^2) \\
-i\gamma(3+\gamma^2) & 1+3\gamma^2 \\
\end{pmatrix},
\end{align}
which in turn produce the further two adjoints
\begin{align}\label{372}
{H_{PT}}^{\flat_1}&=\frac{-1}{\beta^6}
\begin{pmatrix}
i\gamma(7+35\gamma^2+21\gamma^4+\gamma^6) & 1+21\gamma^2+35\gamma^4+7\gamma^6 \\
1+21\gamma^2+35\gamma^4+7\gamma^6 & -i\gamma(7+35\gamma^2+21\gamma^4+\gamma^6) \\
\end{pmatrix},
\\\label{373}
{H_{PT}}^{\sharp_1}&=\frac{-1}{\beta^4}
\begin{pmatrix}
-i\gamma(5+10\gamma^2+\gamma^4) & 1+5\gamma^2+10\gamma^2 \\
1+5\gamma^2+10\gamma^2 & i\gamma(5+10\gamma^2+\gamma^4) \\
\end{pmatrix}.
\end{align}
We can again confirm the relations \eqref{237-1}. 
We stop here for the expressions become more complicated and we do not find self-adjointness other than the pseudo-Hermiticity, but we note that the $PT$-symmetry defined at the beginning of this subsection is conserved throughout the process, that is, every Hamiltonian generated here commute with the $PT$ operation. 

Let us now move to the case of pure imaginary eigenvalues, namely the $PT$-broken phase.
We set $\hop=1$ but $\gamma>1$ hereafter, and choose the ordering of the eigenvalues such that
\begin{align}
E_1=i\kappa\qquad \mbox{and}\qquad E_2=-i\kappa,
\end{align}
where $\kappa=\sqrt{\gamma^2-1}\in\mathbb{R}$. This is the case discussed in the Remark 1 in Section \ref{subsec22}.
We note here that
\begin{align}
\tilde{H}_{PT}:=iH_{PT}=\begin{pmatrix}
-\gamma & -i\hop \\
-i\hop & \gamma 
\end{pmatrix}
\end{align}
has two real eigenvalues $\mp\kappa$.


Before exploring the alley of $\tilde{H}_{PT}$, we first give the generation of the isospectral Hamiltonians of $H_{PT}$.
We can choose the eigenvectors in the following forms:
\begin{align}
\varphi_1\supzero&:=\frac{1}{\sqrt{2i\kappa}}
\begin{pmatrix}
1\\
i\gamma-i\kappa\\
\end{pmatrix},
\qquad
\varphi_2\supzero:=\frac{1}{\sqrt{2i\kappa}}
\begin{pmatrix}
i\gamma-i\kappa\\
-1\\
\end{pmatrix},
\\
\psi_1\supzero&:=\frac{-1}{\sqrt{-2i\kappa}}
\begin{pmatrix}
i\gamma+i\kappa\\
1\\
\end{pmatrix},
\qquad
\psi_2\supzero:=\frac{1}{\sqrt{-2i\kappa}}
\begin{pmatrix}
-1\\
i\gamma+i\kappa\\
\end{pmatrix}.
\end{align}
We can again confirm the biorthonormality $\scp<\psi_m\supzero,\varphi_n\supzero>=\delta_{m,n}$ for $m,n=1,2$ as well as the resolution of the identity $\sum_{n=1}^2 \dyad{\varphi_n\supzero}{\psi_n\supzero}=\mathbb{I}$.

The intertwining operators in \eqref{24} are given by
\begin{align}
S_\varphi\supzero=\frac{1}{\kappa(\gamma+\kappa)}
\begin{pmatrix}
\gamma & -i \\
i & \gamma \\
\end{pmatrix},
\qquad
S_\psi\supzero=\frac{\gamma+\kappa}{\kappa}
\begin{pmatrix}
\gamma & i \\
-i & \gamma \\
\end{pmatrix}.
\end{align}
They are indeed inverse matrices to each other.
They then define the two adjoints in \eqref{27} with \eqref{28} as follows:
\begin{align}\label{379}
{H_{PT}}^{\flat_0}=\frac{-1}{\kappa^2}
\begin{pmatrix}
i\gamma(3+\gamma^2) & 1+3\gamma^2 \\
1+3\gamma^2 & -i\gamma(3+\gamma^2) \\
\end{pmatrix},
\qquad
{H_{PT}}^{\sharp_0}=
\begin{pmatrix}
-i\gamma & 1 \\
1 & i\gamma \\
\end{pmatrix}.
\end{align}
The Hamiltonian is therefore \textit{not} pseudo-Hermitian because ${H_{PT}}^{\sharp_0}=-H_{PT}$, but meanwhile let us note that the Hamiltonians~\eqref{379} are both negatives of the corresponding Hamiltonians in the $PT$-unbroken case~\eqref{368} because $\kappa^2=-\beta^2$.
We can confirm that the relations~\eqref{29} are satisfied.

The vectors \eqref{210}, which start the second iteration, are given by
\begin{align}
\varphi_1\supone&=\frac{1}{\sqrt{2\kappa^3}(\gamma+\kappa)}
\begin{pmatrix}
e^{-i\pi/4}\qty(2\gamma-\kappa)\\
e^{i\pi/4} \qty(1-\gamma\kappa+\gamma^2)\\
\end{pmatrix},
\qquad
\varphi_2\supone=\frac{1}{\sqrt{2\kappa^3}(\gamma+\kappa)}
\begin{pmatrix}
e^{i\pi/4} \qty(1-\gamma\kappa+\gamma^2)\\
-e^{-i\pi/4}\qty(2\gamma-\kappa)\\
\end{pmatrix},
\\
\psi_1\supone&=\frac{1}{\sqrt{2\kappa^3}}
\begin{pmatrix}
e^{-i\pi/4} (\gamma+\kappa)(1+\gamma\kappa+\gamma^2)\\
e^{i\pi/4} (1-3\gamma\kappa-3\gamma^2)\\
\end{pmatrix},
\qquad
\psi_2\supone=\frac{1}{\sqrt{2\kappa^3}}
\begin{pmatrix}
e^{i\pi/4} (1-3\gamma\kappa-3\gamma^2)\\
-e^{-i\pi/4} (\gamma+\kappa)(1+\gamma\kappa+\gamma^2)\\
\end{pmatrix}.
\end{align}
They produce the further two intertwining operators \eqref{213} as in
\begin{align}
S_\varphi\supone&=\frac{\gamma-\kappa}{\kappa^3(\gamma+\kappa)^2}
\begin{pmatrix}
\gamma(3+\gamma^2)  & -i(1+3\gamma^2) \\
i(1+3\gamma^2) & \gamma(3+\gamma^2) \\
\end{pmatrix},
\\
S_\psi\supone&=\frac{-3\gamma-\kappa+4\gamma^2(\gamma+\kappa)}{\kappa^3}
\begin{pmatrix}
\gamma(3+\gamma^2)  & i(1+3\gamma^2) \\
-i(1+3\gamma^2) & \gamma(3+\gamma^2) \\
\end{pmatrix},
\end{align}
which in turn produce the further two adjoints
\begin{align}\label{384}
{H_{PT}}^{\flat_1}&=\frac{-1}{\kappa^6}
\begin{pmatrix}
i\gamma(7+35\gamma^2+21\gamma^4+\gamma^6) & 1+21\gamma^2+35\gamma^4+7\gamma^6 \\
1+21\gamma^2+35\gamma^4+7\gamma^6 & -i\gamma(7+35\gamma^2+21\gamma^4+\gamma^6) \\
\end{pmatrix},
\\\label{385}
{H_{PT}}^{\sharp_1}&=\frac{1}{\kappa^4}
\begin{pmatrix}
-i\gamma(5+10\gamma^2+\gamma^4) & 1+5\gamma^2+10\gamma^2 \\
1+5\gamma^2+10\gamma^2 & i\gamma(5+10\gamma^2+\gamma^4) \\
\end{pmatrix}.
\end{align}
We again note that the Hamiltonians~\eqref{384}-\eqref{385} are both sign inversions of the corresponding Hamiltonians in the $PT$-unbroken case~\eqref{372}-\eqref{373} because $\beta^6=-\kappa^6=(1-\gamma)^3$ while $\beta^4=\kappa^4=(1-\gamma)^2$. 

Let us now take advantage of the fact that $\tilde{H}_{PT}:=iH_{PT}$ has real eigenvalues.
Since $\tilde{H}_{PT}$ and $H_{PT}$ share the eigenvectors, all $\varphi_n^{(\nu)}$, $\psi_n^{(\nu)}$, $S_\varphi^{(\nu)}$ and $S_\psi^{(\nu)}$ remain the same.
The difference appears in the generated Hamiltonians;
they are all multiplied by $-i$ because they are generated out of ${\tilde{H}_{PT}}^\dag=-i{H_{PT}}^\dag$:
\begin{align}
{\tilde{H}_{PT}}^{\flat_0}&=-i{H_{PT}}^{\flat_0},
\qquad
{\tilde{H}_{PT}}^{\sharp_0}=-i{H_{PT}}^{\sharp_0},
\\
{\tilde{H}_{PT}}^{\flat_1}&=-i{H_{PT}}^{\flat_1},
\qquad
{\tilde{H}_{PT}}^{\sharp_1}=-i{H_{PT}}^{\sharp_1}
\end{align}
The pseudo-Hermiticity ${\tilde{H}_{PT}}^{\sharp_0}=\tilde{H}_{PT}$ is now recovered,
which may open up a possible path to the formulation of the present theory for the $PT$-broken phase with complex eigenvalues.

\subsubsection{$PT$ chain with alternate hopping}\label{sect3.3.2}

We finally consider a model with alternating hoppings (called the Su-Schrieffer-Heeger model~\cite{SSH}) with $PT$-alternating potentials~\cite{Rudner09}:
\begin{align}
H_{RL}:=-\hop\sum_{x=1}^{2L} \qty[1-\delta(-1)^x] \qty(\dyad{x+1}{x}+\dyad{x+1}{x})-\sum_{x=1}^L i\gamma (-1)^x\dyad{x},
\end{align}
where $\hop,\delta,\gamma\in\mathbb{R}$. 
We set periodic boundary conditions: $\ket{x+2L}=\ket{x}$.
The hopping element alternates between $1+\delta$ and $1-\delta$, while the potential alternates between $i\gamma$ and $-i\gamma$.
The matrix representation is given for \textit{e.g.}\ $2L=8$ by
\begin{align}\label{374}
H_{RL}=\qty(\begin{array}{|cc|cc|cc|cc|}
\cline{1-2}
+i\gamma & -u & \multicolumn{5}{c}{} & \multicolumn{1}{c}{-w} \\
-u & -i\gamma & -w & \multicolumn{5}{c}{} \\
\cline{1-4} 
\multicolumn{1}{c}{} & -w & +i\gamma & -u & \multicolumn{4}{c}{} \\
\multicolumn{1}{c}{} & & -u & -i\gamma & -w & \multicolumn{3}{c}{} \\
\cline{3-6}
\multicolumn{3}{c}{} & -w & +i\gamma & -u & & \multicolumn{1}{c}{} \\
\multicolumn{4}{c|}{} & -u & -i\gamma & -w & \multicolumn{1}{c}{} \\
\cline{5-8}
\multicolumn{5}{c}{} & -w & +i\gamma & -u \\
\multicolumn{1}{c}{-w} & \multicolumn{5}{c|}{} & -u & -i\gamma \\
\cline{7-8}
\end{array}),
\end{align}
where
\begin{align}
u=\hop (1+\delta) \qquad \mbox{and} \qquad w=\hop(1-\delta).
\end{align}
In other words, a pair of two neighboring sites forms a unit cell, as is indicated by the squares in \eqref{374};
each unit-cell Hamiltonian has the same form as the $2\times 2$ example $H_{PT}$ in Eq.~\eqref{363}, only the off-diagonal elements $-\hop$ being replaced by $-u$.
Each unit cell is then coupled to its neighboring unit cells by the hopping $-w$.

This inspires us to rewrite Eq.~\eqref{374} in the form
\begin{align}\label{377}
H_{RL}&=\sum_{\xi=1}^L \qty[i\gamma\qty(\dyad{\xi,A}-\dyad{\xi,B})-u\qty(\dyad{\xi,B}{\xi,A}+\dyad{\xi,A}{\xi,B})]
\nonumber\\
&+\sum_{\xi=1}^L\qty[-w\qty(\dyad{\xi+1,A}{\xi,B}+\dyad{\xi,B}{\xi+1,A})],
\end{align}
where
\begin{align}
\ket{\xi,A}=\ket{x=2\xi-1}\qquad\mbox{and}\qquad \ket{\xi,B}=\ket{x=2\xi},
\end{align}
for $\xi=1,2,3,\cdots,L$ denote the unit cells.
The periodic boundary conditions now read $\ket{\xi+L,A/B}=\ket{\xi,A/B}$.
The first term of the expression~\eqref{377} represents the Hamiltonian of each unit cell, whereas the second term represents the coupling between two neighboring cells.

Since this Hamiltonian has translational symmetry based on the shift operation $\ket{x}\to\ket{x\pm2}$, or $\ket{\xi,A/B}\to\ket{\xi\pm1,A/B}$, we can apply the Fourier transformation to block-diagonalize it.
Specifically, by using the basis set defined by
\begin{align}
\ket{k_n,C}&:=\frac{1}{\sqrt{L}}\sum_{\xi=1}^L e^{ik_n \xi}\ket{\xi,C},
\qquad
\ket{\xi,C}:=\frac{1}{\sqrt{L}}\sum_{n=1}^L e^{-ik_n \xi}\ket{k_n,C},
\end{align}
where $k_n=2\pi/L$ for $n=1,2,3,\cdots,L$ and $C=A,B$, the Hamiltonian~\eqref{377} is expressed in the following form of the direct product:
\begin{align}
H_{RL}&=\sum_{n=1}^L \qty[i\gamma\qty(\dyad{k_n,A}-\dyad{k_n,B})-u\qty(\dyad{k_n,B}{k_n,A}+\dyad{k_n,A}{k_n,B})]
\nonumber\\
&+\sum_{n=1}^L\qty[-w\qty(e^{-ik_n}\dyad{k_n,A}{k_n,B}+e^{ik_n}\dyad{k_n,B}{k_n,A})]
\\
&=\bigotimes_{n=1}^L
\begin{pmatrix}
i\gamma & -u-we^{-ik_n} \\
-u-we^{ik_n} & -i\gamma
\end{pmatrix},
\end{align}
where the last expression represents each block of the block-diagonalized form of $H_{RL}$.

We can now proceed the analysis for each $2\times2$ block independently, slightly generalizing the analysis for the $2\times2$ Hamiltonian $H_{PT}$ in~\eqref{363}.
Let us denote a block by
\begin{align}\label{396}
H_{RL}^{(n)}:=\begin{pmatrix}
i\gamma & -\mu_n+i\nu_n \\
-\mu_n-i\nu_n & -i\gamma \\
\end{pmatrix}
=-\mu_n\sigma_x-\nu_n\sigma_y+i\gamma\sigma_z,
\end{align}
where
\begin{align}
\mu_n:=u+w\cos k_n,
\qquad
\nu_n:=w\sin k_n
\end{align}
are both real, and $\sigma_x$, $\sigma_y$, $\sigma_z$ denote the Pauli matrices.
We denote the eigenvalues by 
\begin{align}
E_1^{(n)}=-\varepsilon_n\qquad\mbox{and}\qquad
E_2^{(n)}=+\varepsilon_n, \qquad
\mbox{where}\qquad
\varepsilon_n=\sqrt{{\mu_n}^2+{\nu_n}^2-\gamma^2},
\end{align}
which can be real or pure imaginary, depending on the relative magnitudes of $\mu_n$, $\nu_n$ and $\gamma$.
Since $(\abs{u}-\abs{w})^2\leq {\mu_n}^2+{\nu_n}^2=u^2+2uw\cos k_n+w^2\leq (\abs{u}+\abs{w})^2$, the eigenvalues of all blocks are real when $\abs{\gamma}<\abs{\abs{u}-\abs{w}}=2\min(1,\abs{\delta})$, whereas all of them are pure imaginary when $\abs{\gamma}>\abs{\abs{u}+\abs{w}}=2\max(1,\abs{\delta})$.
In between them, namely $2\min(1,\abs{\delta})<\abs{\gamma}<2\max(1,\abs{\delta})$, the eigenvalues are real for some blocks and pure imaginary for other blocks with possibilities of zero eigenvalues at exceptional points.

Let us start with the case of all real eigenvalues: ${\mu_n}^2+{\nu_n}^2>\gamma^2$ for all $n=1,2,3,\cdots,L$.
For brevity, we drop the block label $n$ from the superscripts and the subscripts.
We can set the eigenvectors for each block as
\begin{align}
\varphi_1\supzero&:=\frac{1}{\sqrt{2(\mu-i\nu)\varepsilon}}
\begin{pmatrix}
\mu-i\nu\\
i\gamma+\varepsilon\\
\end{pmatrix},
\qquad
\varphi_2\supzero:=\frac{-1}{\sqrt{2(\mu+i\nu)\varepsilon}}
\begin{pmatrix}
i\gamma+\varepsilon\\
-\mu-i\nu\\
\end{pmatrix},
\\
\psi_1\supzero&:=\frac{1}{\sqrt{2(\mu+i\nu)\varepsilon}}
\begin{pmatrix}
i\gamma+\varepsilon\\
\mu+i\nu\\
\end{pmatrix},
\qquad
\psi_2\supzero:=\frac{1}{\sqrt{2(\mu-i\nu)\varepsilon}}
\begin{pmatrix}
-\mu+i\nu\\
i\gamma+\epsilon\\
\end{pmatrix}.
\end{align}
%
The intertwining operators in \eqref{24} are given for each block by
\begin{align}
S_\varphi\supzero=\frac{1}{\varepsilon}
\begin{pmatrix}
\sqrt{\mu^2+\nu^2} & -i\gamma\sqrt{\frac{\mu-i\nu}{\mu+i\nu}} \\
i\gamma\sqrt{\frac{\mu+i\nu}{\mu-i\nu}} & \sqrt{\mu^2+\nu^2} \\
\end{pmatrix},
\qquad
S_\psi\supzero=\frac{1}{\varepsilon}
\begin{pmatrix}
\sqrt{\mu^2+\nu^2} & i\gamma\sqrt{\frac{\mu-i\nu}{\mu+i\nu}} \\
-i\gamma\sqrt{\frac{\mu+i\nu}{\mu-i\nu}} & \sqrt{\mu^2+\nu^2} \\
\end{pmatrix}.
\end{align}
The two adjoints in \eqref{27} with \eqref{28} are then given by
\begin{align}
{H_{RL}}^{\flat_0}&=\frac{-1}{\varepsilon^2}
\begin{pmatrix}
i\gamma(3\mu^2+3\nu^2+\gamma^2) & (\mu-i\nu)(\mu^2+\nu^2+3\gamma^2) \\
(\mu+i\nu)(\mu^2+\nu^2+3\gamma^2) & -i\gamma(3\mu^2+3\nu^2+\gamma^2) \\
\end{pmatrix},
\\
{H_{RL}}^{\sharp_0}&=
\begin{pmatrix}
i\gamma & -\mu+i\nu \\
-\mu-i\nu & -i\gamma \\
\end{pmatrix},
\end{align}
the latter of which satisfies the pseudo-Hermiticity ${H_{RL}}^{\sharp_0}=H_{RL}$.

As we move to the case of all pure imaginary eigenvalues (${\mu_n}^2+{\nu_n}^2<\gamma^2$ for all $n=1,2,3,\cdots,L$), we analyze
\begin{align}\label{3104}
\tilde{H}_{RL}^{(n)}:=iH_{RL}^{(n)}
\begin{pmatrix}
-\gamma & -i\mu_n-\nu_n \\
-i\mu_n+\nu_n & \gamma \\
\end{pmatrix}
\end{align}
as we did at the end of Subsubsect.~\ref{subusubsec3.2.1} instead of \eqref{396}, because then the eigenvalues $\mp\sqrt{\gamma^2-{\mu_n}^2-{\nu_n}^2}$ are real for all $n=1,2,3,\cdots,L$.

We define the ordering of the eigenvalues as $E_1^{(n)}=-\omega$ and $E_2^{(n)}=+\omega_n$ with $\omega_n =\sqrt{\gamma^2-{\mu_n}^2-{\nu_n}^2}>0$.
As above, we drop the scripts $n$ hereafter.
The eigenvectors for each block then reads
\begin{align}
\varphi_1\supzero&:=\frac{1}{\sqrt{2i\omega(\mu-i\nu)}}
\begin{pmatrix}
\mu-i\nu\\
i(\gamma-\omega)\\
\end{pmatrix},
\qquad
\varphi_2\supzero:=\frac{1}{\sqrt{2i\omega(\mu+i\nu)}}
\begin{pmatrix}
i(\gamma-\omega)\\
-\mu-i\nu\\
\end{pmatrix},
\\
\psi_1\supzero&:=\frac{-1}{\sqrt{-2i(\omega\mu+i\nu)}}
\begin{pmatrix}
i(\gamma+\omega)\\
\mu+i\nu\\
\end{pmatrix},
\qquad
\psi_2\supzero:=\frac{1}{\sqrt{-2i\omega(\mu-i\nu)}}
\begin{pmatrix}
-\mu+i\nu\\
i(\gamma+\omega)\\
\end{pmatrix}.
\end{align}
%
The intertwining operators in \eqref{24} for each block are 
\begin{align}
S_\varphi\supzero=\frac{\gamma-\omega}{\sqrt{\mu^2+\nu^2}\omega}
\begin{pmatrix}
\gamma & -i(\mu-i\nu) \\
i(\mu+i\nu) & \gamma \\
\end{pmatrix},
\qquad
S_\psi\supzero=\frac{\gamma+\omega}{\sqrt{\mu^2+\nu^2}\omega}
\begin{pmatrix}
\gamma & i(\mu-i\nu) \\
-i(\mu+i\nu) & \gamma \\
\end{pmatrix}.
\end{align}
The two adjoints in \eqref{27} with \eqref{28} are
\begin{align}
{\tilde{H}_{RL}}^{\flat_0}&=\frac{-1}{\omega^2}
\begin{pmatrix}
\gamma(3\mu^2+3\nu^2+\gamma^2) & (\mu-i\nu)(\mu^2+\nu^2+3\gamma^2) \\
(\mu+i\nu)(\mu^2+\nu^2+3\gamma^2) & -i\gamma(3\mu^2+3\nu^2+\gamma^2) \\
\end{pmatrix},
\\
{\tilde{H}_{RL}}^{\sharp_0}&=
\begin{pmatrix}
-\gamma & -i\mu-\nu \\
-i\mu+\nu & \gamma \\
\end{pmatrix},
\end{align}
the latter of which again satisfies the pseudo-Hermiticity ${\tilde{H}_{RL}}^{\sharp_0}=\tilde{H}_{RL}$ even in the $PT$-broken region, at least when all the eigenvalues are purely imaginary. We refer to  Remark 1 in Section \ref{subsec22}  for some general results on this particular situation. 

\section{Conclusion}\label{sectconcl}

We have proposed a simple strategy to construct more and more Hamiltonians, not necessarily (Dirac) self-adjoint, whose eigenstates and eigenvalues can be automatically deduced out of a single operator $H$. The key idea is that the role of the metric operators defined in terms of the eigenstates of $H$ and $H^\dagger$ can be explored further than usually done in the literature, to produce many other (actually, infinite more) biorthogonal sets of vectors which, together, produce new scalar products and, as a consequences, new adjoints maps and new Hamiltonians which are isospectral to $H$ and $H^\dagger$, and whose eigenvectors can be deduced out of those of $H$.

We have applied our strategy to several physical systems, showing that the procedure is not trivial, meaning that the new operators and vectors we get are really different from those we started with. We have also briefly discussed the physical meaning of this approach, in our applications.

There are many aspects which, in ur opinion, deserve further analysis. First of all, the role of symmetries in this context. How a symmetry should be defined? And is a symmetry preserved by our construction? But, does it need to be preserved, first? Secondly, it would be interesting to extend the approach to infinite-dimensional Hilbert spaces. This will allow us to cover many physically relevant situations, which do not fit the assumptions used here. But other aspects are also relevant in the present finite-dimensional settings. For instance, the tridiagonal matrices appearing in Section \ref{subsec3.2} could be analyzed as in \cite{bgr}, in the large $L$ limit, and relations with biorthogonal polynomials could be deduced. Also, by factorizing them, the SUSY partners of the Hamiltonians considered in this paper could be studied, and their physical meaning could be considered. The following question looks also interesting to us: does the chain of Hamiltonians converge to some limiting operators? When and why? This is interesting, since it would imply that one such limiting Hamiltonian is a fixed point for our strategy, and open the door to the physical interpretation of such a result.

\section*{Acknowledgements}

F.B. acknowledges financial support from the University of Tokyo and from the Istituto Nazionale di Fisica Nucleare. F.B. also acknowledges  support from Palermo University and from the Gruppo Nazionale di Fisica Matematica of the I.N.d.A.M.

\renewcommand{\theequation}{A.\arabic{equation}}


\begin{thebibliography}{99}
	
	\bibitem{intop1} Kuru S., Tegmen A., Vercin A., {\em Intertwined isospectral potentials in an arbitrary dimension},
	J. Math. Phys, {\bf 42}, No. 8, 3344-3360, (2001) 
	
	\bibitem{intop2} Kuru S.,
	Demircioglu B., Onder M., Vercin A., {\em Two families of
		superintegrable and isospectral potentials in two dimensions}, J.
	Math. Phys, {\bf 43}, No. 5, 2133-2150, (2002)
	
	\bibitem{intop3}  Samani K. A., Zarei
	M., {\em Intertwined hamiltonians in two-dimensional curved spaces},
	Ann. of Phys., {\bf 316}, 466-482, (2005).
	
	\bibitem{bagintop} F. Bagarello {\em Mathematical aspects of intertwining
		operators: the role of Riesz bases},   J. Phys. A,  {\bf 43},  175203 (2010) 
	
	
	\bibitem{fern} F. M. Fernandez, {\em Symmetric quadratic Hamiltonians with
		pseudo-Hermitian matrix representation}, Ann. of Phys., {\bf 369}, 168-176 (2016)
	
	
	\bibitem{jun} G. Junker, {\em Supersymmetric methods in quantum and statistical physics},
	Springer-Verlag, Berlin (1992)
	
	\bibitem{coop} F. Cooper, A. Khare, U. Sukhatme, {\em Supersymmetry and quantum mechanics}, Word Scientific, Singapore, (2001)
	
	
	\bibitem{suku} C. V. Sukumar, {\em Supersymmetric quantum mechanics and its applications}, AIP Conference Proceedings {\bf 744}, 166 (2004)
	
	\bibitem{bagI} F. Bagarello {\em Extended SUSY quantum mechanics, intertwining operators and coherent states},
	Phys. Lett. A,  {\bf 372}, 6226-6231 (2008)
	
		\bibitem{dong} S.-H. Dong, {\em Factorization method in quantum mechanics},  Springer, Dordrecht (2007)
	
	
	\bibitem{Gamow28}
	G. Gamow, 
	{\em Zur Quantentheorie des Atomkernes (On quantum theory of atomic nuclei)},
	Z. Phys. {\bf 51}, 204-212 (1928)
	
	\bibitem{Feshbach58}
H. Feshbach,
{\em A unified theory of nuclear reactions},
Ann. Phys. (New York) {\bf 5}, 357-390 (1958)

	\bibitem{Peierls59}
	R. E. Peierls,
{\em Complex eigenvalues in scattering theory},
Proc. Roy. Soc. London A,
{\bf 253}, 16--36 (1959)


	
	\bibitem{Hatano96}
	N. Hatano and D. R. Nelson, {\em Localization Transitions in Non-Hermitian Quantum Mechanics}, Phys. Rev. Lett. {\bf 77}, 570-573 (1996)
	\bibitem{Hatano97}
	N. Hatano and D. R. Nelson, {\em Vortex pinning and non-Hermitian quantum mechanics}, Phys. Rev. B {\bf 56}, 8651-8673 (1997)
	
	\bibitem{Hatano98}
	N. Hatano,
	{\em Localization in non-Hermitian quantum mechanics and flux-line pinning in superconductors},
	Physica A {\bf 254}, 317--331 (1998).
	
		\bibitem{BB98}
		C.M. Bender and S. Boettcher,
		{\em Real Spectra in Non-Hermitian Hamiltonians Having PT Symmetry}, Phys. Rev. Lett. {\bf 80}, 5243-5247 (1998)
		
		\bibitem{benbook} C.M. Bender, {\em PT Symmetry in Quantum and Classical Physics}, World Scientific, Singapore,  2019
		
		
		\bibitem{ben}  C.M. Bender, {\em Making Sense of Non-Hermitian Hamiltonians}, Rep. Progr.  Phys., {\bf 70},  947-1018 (2007)

	\bibitem{Gong18}
Zongping Gong, Yuto Ashida, Kohei Kawabata, Kazuaki Takasan, Sho Higashikawa, and Masahito Ueda,
	{\em Topological Phases of Non-Hermitian Systems},
Phys. Rev. X {\bf 8}, 031079 (2018)
	
	\bibitem{Kawabata19}
Kohei Kawabata, Ken Shiozaki, Masahito Ueda, and Masatoshi Sato,
	{\em Symmetry and Topology in Non-Hermitian Physics},
Phys. Rev. X {\bf 9}, 041015 (2019)

\bibitem{Turker18}
Z.Turker, S.Tombuloglu, C.Yuce,
{\em PT symmetric Floquet topological phase in SSH model}
Phys. Lett. A, {\bf 382} 2013-2016 (2018)

\bibitem{Harter20}
A.K. Harter and N. Hatano,
{\em Real Edge Modes in a Floquet-modulated $PT$-symmetric SSH model},
arXiv:2006.16890
	
		\bibitem{chri} Christensen O., {\em An Introduction to Frames and Riesz Bases}, Birkh\"auser, Boston, (2003)
		
			\bibitem{mosta} A. Mostafazadeh, {\em Pseudo-Hermitian representation of Quantum Mechanics}, Int. J. Geom. Methods Mod. Phys. {\bf 7}, 1191-1306 (2010)
		
		
			\bibitem{rs} M. Reed and B. Simon, {\em Methods of Modern Mathematical Physics I: Functional analysis}, Academic Press, New York, (1980)
		
			\bibitem{baginbagbook} F. Bagarello, {\em Deformed canonical (anti-)commutation relations and non hermitian hamiltonians}, in {Non-selfadjoint operators in quantum physics: Mathematical aspects}, F. Bagarello, J. P. Gazeau, F. H. Szafraniek and M. Znojil Eds., John Wiley and Sons Eds.,  (2015)
		
		
		
			\bibitem{santos} R. B. B. Santos, V. R. Da Silva, {\em Non-hermitian model for asymmetrical tunneling}, Mod. Phys. Lett. B,  {\bf 28}, 1450223 (2014)
		
		
		
		\bibitem{bagaop2015a} F. Bagarello, {\em Some results on the dynamics and  transition probabilities for non self-adjoint hamiltonians},  Ann. of Phys., {\bf 356}, 171-184 (2015)
		
		
		
		
		
		
		
		
		
		
		
		
		
		
		
		
		
		
	
	
	
	
%
%
%
%
%
%
%
%
%
%
%
%
	
	
	
	
	
	
	
	
	
	
	
	
	%
	%
	%
	%
	
	\bibitem{CMB02}
	C. M. Bender, D. C. Brody and H. F. Jones, {\em Complex extension of quantum mechanics}, Phys. Rev. Lett. {\bf 89}, 270401 (2002)
	
	
	
	
	
	\bibitem{SSH}
	W. P. Su, J. R. Schrieffer, and A. J. Heeger,
	{\em Solitons in Polyacetylene}, Phys. Rev. Lett. {\bf 42}, 1698-1701 (1979)s
	
	\bibitem{Rudner09}
	M. S. Rudner and L. S. Levitov,
{\em	Topological Transition in a Non-Hermitian Quantum Walk}, Phys. Rev. Lett. {\bf 102}, 065703 (2009)
	
	\bibitem{Hatano19}
	N. Hatano and G. Ordonez, {\em Time-Reversal Symmetry and Arrow of Time in Quantum Mechanics of Open Systems}, Entropy {\bf 21}, 380 (2019)
	
	\bibitem{bgr} F. Bagarello, F. Gargano, F. Roccati, {\em Tridiagonality, supersymmetry and non self-adjoint Hamiltonians},  J. Phys. A,  {\bf 52}, 355203 (2019)
	
	
\end{thebibliography}
\end{document}